\documentclass[]{IEEEtran}
\usepackage{mathtools}
\usepackage{comment}
\usepackage{algorithm,algorithmicx,algpseudocode,amsmath,amssymb,amsthm}
\usepackage{float}
\usepackage{subfig}
\usepackage{bm,color,dsfont,pifont,hyperref}
\usepackage{makecell}
\usepackage{multirow}
\usepackage{scalerel}
\usepackage{url}
\usepackage{accents}
\usepackage{cite}

\usepackage{authblk}

\makeatletter
\renewcommand{\maketag@@@}[1]{\hbox{\m@th\normalsize\normalfont#1}}%
\makeatother

\newtheorem{fact}{Fact}
\newtheorem{corollary}{Corollary}

\newtheorem{definition}{Definition}

\newtheorem{lemma}{Lemma}
\newtheorem{theorem}{Theorem} 
\newtheorem{remark}{Remark}

\newcommand{\bfm}[1]{\bm{#1}}
\allowdisplaybreaks
\newcommand\numberthis{\addtocounter{equation}{1}\tag{\theequation}}

\newcommand{\norm}[1]{\left\lVert#1\right\rVert}
\renewcommand{\(}{\left(}
\renewcommand{\)}{\right)}

\def\widebar{\accentset{{\cc@style\underline{\mskip10mu}}}}
\def\la{\left\langle}
\def\ra{\right\rangle}

\def\ln{\left\|}
\def\rn{\right\|}

\def\A{\mathcal{A}}
\def\C{\mathbb{C}}
\def\M{\mathcal{M}}

\def\P{\mathcal{P}}
\def\T{\mathcal{T}}
\def\H{\mathcal{H}}
\def\D{\mathcal{D}}
\def\G{\mathcal{G}}
\def\R{\mathbb{R}}
\def\E{\mathbb{E}}

\def\I{\mathcal{I}}

\newcommand{\va}{\bm{a}}
\newcommand{\vb}{\bm{b}}

\newcommand{\ve}{\bm{e}}

\newcommand{\vh}{\bm{h}}

\newcommand{\vr}{\bm{r}}

\newcommand{\vx}{\bm{x}}

\newcommand{\vv}{\bm{v}}

\newcommand{\vy}{\bm{y}}
\newcommand{\vz}{\bm{z}}
\newcommand{\vtau}{\bm{\tau}}

\newcommand{\vZ}{\bm{Z}}

\newcommand{\vF}{\bm{F}}

\newcommand{\vL}{\bm{L}}
\newcommand{\vR}{\bm{R}}

\newcommand{\vM}{\bm{M}}
\newcommand{\vI}{\bm{I}}
\newcommand{\vU}{\bm{U}}
\newcommand{\vV}{\bm{V}}

\newcommand{\vQ}{\bm{Q}}

   \def\mA{\bfm A}  
  \def\mB{\bfm B} 
  \def\mC{\bfm C}  
    
  \def\mE{\bfm E}  
 \def\mF{\bfm F}  
    
  \def\mH{\bfm H} 
  \def\mI{\bfm I} 
   \def\mJ{\bfm J} 
   \def\mK{\bfm K}  
   \def\mL{\bfm L}
  \def\mM{\bfm M}  
   \def\mN{\bfm N}  
    
   \def\mP{\bfm P}  
  \def\mQ{\bfm Q}  
   \def\mR{\bfm R}

 \def\mV{\bfm V}  
    
   \def\mX{\bfm X}  
   \def\mY{\bfm Y}  
   \def\mZ{\bfm Z}

\newcommand{\br}{\bm{r}}

\newcommand{\bD}{\bm{D}}

\newcommand{\bF}{\bm{F}}

\newcommand{\bI}{\bm{I}}

\newcommand{\bP}{\bm{P}}
\newcommand{\bQ}{\bm{Q}}
\newcommand{\bR}{\bm{R}}

\newcommand{\bDelta}{\bm{\Delta}}

\newcommand{\bSigma}{\bm{\Sigma}}

\newcommand{\cM}{\mathcal{M}}

\newcommand{\cS}{\mathcal{S}}
\newcommand{\cT}{\mathcal{T}}

\newcommand{\cX}{\mathcal{X}}

\newcommand{\bcD}{\bm{\mathcal{D}}}
\newcommand{\bcE}{\bm{\mathcal{E}}}

\newcommand{\bcH}{\bm{\mathcal{H}}}

\newcommand{\bcS}{\bm{\mathcal{S}}}
\newcommand{\bcT}{\bm{\mathcal{T}}}

\newcommand{\bcW}{\bm{\mathcal{W}}}
\newcommand{\bcX}{\bm{\mathcal{X}}}

\newcommand{\bcZ}{\bm{\mathcal{Z}}}

\renewcommand{\baselinestretch}{1} 

\newcommand{\argmin}{\mathop{\mathrm{argmin}}}

\DeclareMathOperator{\bcdot}{\boldsymbol{\cdot}}
\DeclareMathOperator{\GL}{\mathrm{GL}}
\DeclareMathOperator{\op}{\mathsf{}}

\DeclareMathOperator{\vc}{\mathrm{vec}}
\DeclareMathOperator{\mulrank}{mulrank}
\DeclareMathOperator{\rank}{rank}
\DeclareMathOperator{\diag}{diag}
\DeclareMathOperator{\Real}{{\normalfont Re}}

\newcommand{\PB}[1]{\P_{B}(#1)}

\newcommand{\dist}[2]{{\normalfont\mbox{dist}}(#1,#2)}
\newcommand{\distQ}[2]{{\normalfont\mbox{dist}_Q}(#1,#2)}
\newcommand{\distP}[2]{{\normalfont\mbox{dist}_P}(#1,#2)}

\newcommand{\distsq}[2]{{\normalfont\mbox{dist}^2}(#1,#2)}

\makeatletter
\renewcommand*{\@fnsymbol}[1]{\ensuremath{\ifcase#1\or \dagger\or \ddagger\or \mathsection\or
    *\or \mathparagraph\or \|\or **\or \dagger\dagger
    \or \ddagger\ddagger \else\@ctrerr\fi}}
\makeatother
 \title{Fast and Provable Hankel Tensor Completion for Multi-measurement Spectral Compressed Sensing}
\author{Jinsheng Li, Xu Zhang, Shuang Wu, and Wei Cui
\thanks{This work was supported in part by the National Natural Science Foundation of China under Grant No. 62025103 and 62201053, the Postdoctoral Fellowship Program of CPSF under Grant No. GZC20232038 and the China Postdoctoral Science Foundation under Grant No. 2024M762521. \textit{(Corresponding author: Shuang Wu.)}}
\thanks{
J. Li, S. Wu, and W. Cui are with the School of Information and Electronics, Beijing Institute of
Technology, Beijing 100081, China (e-mail: 18612398166@163.com; ws900226@163.com; cuiwei@bit.edu.cn). 
X.~Zhang is with the School of Artificial Intelligence, Xidian University, Xi'an 710126, China (e-mail: zhang.xu@xidian.edu.cn).

}}

\begin{document}

\maketitle
\begin{abstract}
In this paper, we introduce a novel low-rank Hankel tensor completion approach to address the 
problem of 
multi-measurement spectral compressed sensing. 
By lifting the multiple 
signals to a Hankel tensor, 
we reformulate this problem into a low-rank Hankel tensor completion task, exploiting the spectral sparsity via the low multilinear rankness of the tensor. 
{Furthermore}, we design {a scaled gradient descent algorithm for Hankel tensor completion (ScalHT), which integrates the low-rank Tucker 
decomposition with the Hankel 
structure.} 
Crucially, we derive novel fast computational formulations that leverage the interaction between these two structures, achieving up to an $O(\min\{s,n\})$-fold improvement in storage and computational efficiency 
compared to the existing algorithms, {where $n$ is the length of signal, $s$ is the number of measurement vectors.} Beyond its practical efficiency, ScalHT is backed by rigorous theoretical guarantees: we establish both recovery and linear convergence guarantees, which, to the best of our knowledge, are the first of their kind for low-rank Hankel tensor completion. 
 Numerical simulations show that our method exhibits significantly lower computational and storage costs while delivering superior recovery performance compared to prior arts. 

\end{abstract}
\begin{IEEEkeywords}
Multiple measurement vectors, spectral compressed sensing, Hankel tensor completion, gradient descent
\end{IEEEkeywords}
\section{Introduction}
In this paper, we study the 
spectral compressed sensing problem with multiple measurements vectors  (MMV) \cite{Cotter2005}, 
which aims to reconstruct the multi-measurement spectral sparse signals from partial observations. 
Multi-measurement spectral sparse signals refer to multiple signals sharing the same sparse frequencies. These signals  
widely arises in applications such as power system monitoring \cite{Zhang2018,Gao2016}, wireless communication \cite{Barbotin2012}, target localization 
in radar systems \cite{Potter2010}, and direction-of-arrival (DOA) estimation in array processing \cite{Krim1996}. The multi-measurement spectral sparse signal $\{\vx_l\}_{l=0}^{s-1}$ 
is formed as:  
\begin{align*}
\vx_l(j) = \sum_{k=0}^{r-1} b_{k,l} e^{(\imath2\pi f_k-\tau_k)j}, 
\numberthis\label{eq:signal_model}
\end{align*}
where $j\in\{0,1,\cdots,n-1\}$, $l\in\{0,1,\cdots,s-1\}$ , $n$ is the length of signal, $s$ is the number of measurement vectors, and $r$ is the joint spectral sparsity,   $\imath=\sqrt{-1}$, $f_k\in[0,1)$ is the $k$-th normalized frequency, $\tau_k$ is the $k$-th damping factor,  $b_{k,l}\in \C$ is the amplitude for the $k$-th sinusoids component at the $l$-th measurement. 
We stack the multiple signals $\{\vx_l\}_{l=0}^{s-1}$ into a matrix as $\mX_\star=[\vx_{0},\cdots,\vx_{s-1}]^T\in\C^{s\times n}$ and $\mX_\star$ can be reformulated as: 
\begin{align*}
\mX_\star=\sum_{k=0}^{r-1}\vb_k\va(p_{k})^T, 
\numberthis\label{eq:Ground_truth}
\end{align*}
where $\vb_{k}=[b_{k,0},\cdots,b_{k,s-1}]^T\in\C^s$, 
$\va(p_k)=[1,p_k,\cdots,p_k^{n-1}]^T\in\C^n$ and $p_k=e^{(\imath2\pi f_k-\tau_k)}$. 

Due to hardware limitations such as sensor malfunction, sparse array design, and non-uniform sampling in the time domain, 
only a portion of the multiple signal ensemble can be observed in practice. 
A natural task, therefore, is to recover the target matrix $\mX_\star$ from its partial observations, i.e.,
\begin{align*}
\mbox{Find}\quad\mX
\quad\mbox{subject to}\quad\P_{\Omega}(\mX)=\sum_{(i,j)\in{\Omega}}\mX_\star(i,j)\ve_{\mathrm{s},i}(\ve_{\mathrm{n},j})^T,
\end{align*}
where $\ve_{{\mathrm{s}},i}$, $\ve_{{\mathrm{n}},j}$ is the canonical basis of $\R^s$ and $\R^n$  respectively,  
${\Omega}\subseteq\{0,1,\cdots,s-1\}{\times}\{0,1,\cdots,n-1\}$ 
is the index set, {$m=|\Omega|$ is the number of observations,} 
 and
 $\P_{\Omega}$ is a projection operator. 

 Gridless approaches were proposed to reconstruct multi-measurement spectrally sparse signals 
 as seen in \cite{Yang2016a,Li2016,Zhang2018}. In particular, \cite{Yang2016a,Li2016} introduced an atomic norm minimization (ANM) framework with multiple measurement vectors, which leverages convex optimization to solve the problem. However, these convex approaches are computationally expensive for large-scale problems. To address these challenges, AM-FIHT \cite{Zhang2018} is proposed by lifting the multiple signal ensemble  $\mX_\star\in\C^{s\times n}$ to a Hankel matrix $\H(\mX_\star)\in\C^{s n_1 \times n_2}$ where $n=n_1+n_2-1$. The spectral sparsity of the signal is then captured through the low-rank structure of the Hankel matrix, expressed as: $$\H(\mX_\star)=\mP_{L}\bm{\Gamma}\mP_R^T,$$ where $\mP_{L}\in\C^{sn_1\times sr}, \bm{\Gamma}\in\C^{sr\times r}$ and $\mP_R\in\C^{n_2\times r}.$ 
 This reformulation converts the problem into a low-rank Hankel matrix completion task as follows
 \begin{align*}
\min_{\mX\in\C^{s\times n}}&\la \P_{\Omega}\(\mX-\mX_\star\),(\mX-\mX_\star)\ra
\\&\mbox{s.t.}~\rank(\H(\mX))=r,
\end{align*}
 which is solved using the fast iterative hard thresholding (FIHT) algorithm \cite{Cai2019}. While this approach reduces computational costs relative to convex methods, 
 its computational complexity remains high, particularly when the number of multiple measurement vectors $s$ is large. 
 \begin{figure*}[!t]
		\centering
		\includegraphics[width=0.7\linewidth]{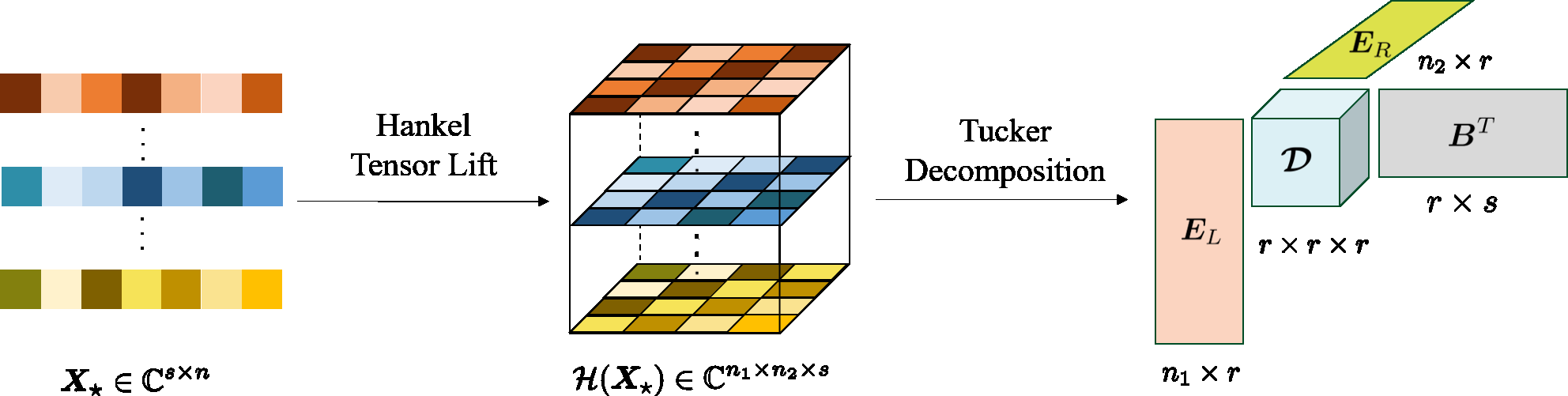}  
\caption{The multiple signal ensemble  $\mX_\star\in\C^{s\times n}$ is lifted to a 
Hankel tensor $\H(\mX_\star)\in\C^{n_1 \times n_2 \times s}$ that exhibits the low-rank Tucker decomposition.}
\label{fig:diagram_mulHT}
\end{figure*}


In this paper, we take a further step by lifting the multiple signals to a Hankel tensor and exploit the joint spectral sparsity via the low multilinear rankness of the tensor, as shown in Fig. \ref{fig:diagram_mulHT}. 
We define 
the Hankel tensor lifting operator  $\H:\C^{s\times n}\rightarrow\C^{n_1\times n_2\times s} 
 (n=n_1+n_2-1)$ \footnote{{We set $n_1=O(n)$ and $n_2= O(n)$ as  explained after Definition \ref{def:incoh}.}} 
 as 
\begin{align*}
    [\H(\mX_\star)](i,j,k)=\mX_\star(k,i+j), 
\end{align*}
where $i\in[n_1]$, $j\in[n_2]$, $k\in[s]$, $[n_1]$ is the set $\{0,1,\cdots,n_1-1\}$, and $[n_2]$, $[s]$ is similarly defined.   
The lifted tensor $\H(\mX_\star)\in\C^{n_1\times n_2\times s}$ naturally admits a low-rank Tucker decomposition \cite{Tucker1966}, which can be expressed as:\footnote{{Although the lifted tensor has low CP rank in essence, we choose the low-rank  Tucker decomposition as it is more stable, with better landscape and initialization strategy \cite{Frandsen2020,Tong2022}.}} 
\begin{align*}
\H(\mX_\star){=}\sum_{k=0}^{r-1} \va_{n_1}(p_k)\circ\va_{n_2}(p_k)\circ\vb_k{=}(\mE_L,\mE_R,\mB)\bcdot\bcD,  \numberthis \label{eq:Ground_truth_Tuckerdcp}
\end{align*}
where $\circ$ denotes the outer product, $\va_{n_1}(p_k)=[1,p_k,\cdots,p_k^{n_1-1}]^T$, 
 $\mE_L=\begin{bmatrix}
    \va_{n_1}(p_0),\cdots,\va_{n_1}(p_{r-1})
\end{bmatrix}\in\C^{n_1\times r}$, 
$\va_{n_2}(p_k)$ and
 $\mE_R\in\C^{n_2\times r}$ 
 are similarly defined,  $\mB=\begin{bmatrix}
    \vb_0,\cdots,\vb_{r-1}
\end{bmatrix}\in\C^{s\times r}$, and $\bcD\in\C^{r\times r \times r}$ is a core tensor, whose non-diagonal elements are zero and  
$\bcD(k,k,k)=1$ for $k=0,\cdots,r-1$. This reformulation converts the problem into a low multilinear rank Hankel tensor completion task as follows: 
 \begin{align*}
\min_{\mX\in\C^{s\times n}}&\la \P_{\Omega}\(\mX-\mX_\star\),(\mX-\mX_\star)\ra
\\&\mbox{s.t.}~\mulrank(\H(\mX))=\vr,\numberthis \label{eq:intro_HTC}
\end{align*}
where $\mulrank(\cdot)$ denotes the multilinear rank of a tensor and $\vr=(r,r,r)$.

\subsection{Contributions}
In this paper, we address multi-measurement spectral compressed sensing by introducing a novel low-rank Hankel tensor completion approach. Furthermore, 
we propose a scaled gradient descent \cite{Tong2021,Tong2022} algorithm based on low-rank Tucker decomposition with a particular emphasis on the Hankel tensor's structure to solve this task, named ScalHT.  

Our main contributions are listed as follows: 

\noindent1) We propose a novel algorithm, ScalHT, for solving multi-measurement spectral compressed sensing via low-rank Hankel tensor completion. ScalHT achieves up to an $O(\min\{s,n\})$-fold improvement in both storage and computation efficiency 
compared to ANM \cite{Yang2016a,Li2016} and AM-FIHT \cite{Zhang2018}, and up to an $O(n\cdot\min\{s,n\})$-fold improvement in computation efficiency compared to ScaledGD \cite{Tong2022}. Detailed comparisons with prior methods are summarized in Table \ref{tab:comp_methods} \footnote{The practical improvement in computation and storage efficiency also depends on $r$ and $\log(n)$, but we focus on their dependence on the ambient dimensions $s$ and $n$ since $r$ and $\log(n)$ $\ll \min\{s,n\}$.}.  
Numerical simulations demonstrate that ScalHT exhibits much lower computational and storage costs, with superior recovery performance compared to prior arts. 

 \noindent2) Beyond its practical efficiency, ScalHT is backed by rigorous theoretical guarantees:  both the recovery and linear convergence guarantees are established 
provided the number of observations is $O(sr^3\kappa^2\log(sn))$. Furthermore, to the best of our knowledge, these are the first rigorous theoretical guarantees ever established for low-rank Hankel tensor completion.

 \noindent3) Several technical innovations are introduced to achieve ScalHT's computation efficiency and theoretical guarantees, which are of independent interest. First, we formulate novel and efficient {computational techniques} that deeply leverage the interaction between {Hankel structure and low-rank tensor decomposition.} These techniques enable ScalHT to achieve a significantly lower per-iteration computational complexity of  $O(s+n)$. 
Second, we propose a provable sequential spectral initialization strategy and derive novel concentration results tailored to Hankel tensor sampling. These contributions are critical for establishing ScalHT's theoretical guarantees. 

 \begin{table*}[t]	
		\begin{center}
			
			\caption{
   Comparisons between algorithms towards multi-measurements spectral compressed sensing.  } 
			\begin{tabular}{c|c|c|c|c}
   \hline
				\textbf{Algorithms} & Computational complexity \footnotemark{} & Storage complexity   & Model formulation & Optimization method
                \\
	\hline			
				\hline
				ANM \cite{Li2016,Yang2016a}   &$O((sn)^3)$ &$O(sn)$  & Atomic norm with MMV  & CVX via interior point method 
               
     \\ 
				ScaledGD \cite{Tong2022} &$O(sn^2)$ &$O\(s+n\)$    &Low-rank tensor  &  Gradient descent
                  \\
				AM-FIHT \cite{Zhang2018} &$O(sn)$ &$O(sn)$  & 
                Low-rank Hankel matrix& Fast iterative hard thresholding
                \\
				ScalHT (ours) &$O\(\bm{s+n}\)$ &$O\(\bm{s+n}\)$    & Low-rank  Hankel tensor 
                &  Gradient descent
 \\   \hline
			\end{tabular}
			\label{tab:comp_methods}
		\end{center}	
	\end{table*}
\subsection{Related work}
When the number of multiple measurement vectors is one, multi-measurement spectral sparse signals reduce to a single spectrally sparse signal \cite{Cai2018,Zhang2021}. Conventional compressed sensing \cite{Candes2006,Donoho2006} could be applied to estimate its spectrum by assuming the frequencies lie on a uniform grid. However, these approaches often suffer from basis mismatch errors in practice. 
To avoid the mismatch error, some gridless methods such as atomic norm minimization (ANM) \cite{Tang2013} and Hankel matrix completion approaches \cite{Chen2014,Cai2018,Cai2019} were proposed. 

Returning to the multiple measurements case, traditional on-grid compressed sensing approaches addressed this problem  
through group sparsity \cite{Tropp2006,Tropp2006a,Mishali2008}, 
and then some gridless methods \cite{Yang2016a,Li2016,Zhang2018} were proposed. 
In particular, \cite{Yang2016a,Li2016} proposed the ANM 
to characterize the joint spectral sparsity, and applied convex optimization to solve this problem. However, 
such convex approaches incur a computational complexity of $O((sn)^3)$ and a storage complexity of $O(sn)$, which are high for large-scale problems. 
Inspired by the fast iterative hard thresholding (FIHT) algorithm \cite{Cai2019}, AM-FIHT \cite{Zhang2018} was introduced 
with fast convergence guarantees, which exploited the joint spectral sparsity via the low-rankness of the Hankel matrix. However, its computational complexity, $O(snr\log
(n)+snr^2)$, and storage complexity, $O(snr)$, remain high when the number of multiple measurements vectors $s$ is large. In contrast, our algorithm ScalHT exhibits a computational complexity as $O(nr^2\log(n)+nr^3+sr^2)$, and a storage complexity as $O((s+n)r+r^3)$, achieving up to $O(\min\{s,n\})$-fold improvement in computation and storage efficiency when treating $r$ and $\log(n)$ as constants. 
Detailed comparisons are listed in Table \ref{tab:comp_methods}. 
Recently, Wu et al. \cite{Wu2023a,Wu2024c,Wu2024b} proposed structured matrix embedding approaches that provide good accuracy, but these methods can't handle damped signals, and 
lack recovery and convergence rate guarantees. 
 

Fast and nonconvex gradient methods based on low-rank  factorization \cite{Chi2019review,Cai2018,Li2024,Li2024a} have garnered significant interest in recent years. In \cite{Cai2018}, a projected gradient descent (PGD) method was proposed for single-measurement spectral compressed sensing, while \cite{Li2024a} introduced a symmetric Hankel projected gradient descent (SHGD) that employs symmetric factorization, effectively reducing both computational and storage costs by nearly half.  
However, when the target matrix or tensor is ill-conditioned, gradient descent methods based on low-rank factorization tend to converge slowly. To address these challenges, scaled gradient descent (ScaledGD) methods \cite{Tong2021,Tong2022,Xu2023,Li2024} were developed to accelerate estimation for ill-conditioned matrices or tensors. 
\footnotetext{This refers to the computational complexity per iteration. }

While it might seem intuitive to directly apply ScaledGD \cite{Tong2022}—a method designed for low-rank tensor estimation—this approach faces several challenges in the context of low-rank Hankel tensor completion. {Specifically, the interplay between the Hankel structure and low-rank tensor decomposition in this problem poses challenges for both theoretical analysis and computation.} 
First, previous concentration inequalities, Lemmas 18-21, and off-diagonal spectral initialization in  \cite{Tong2022}, which apply to random tensor sampling, cannot be generalized to this Hankel tensor completion problem.  
     This is because the Hankel tensor sampling pattern introduces the dependence between the first and second dimensions of the tensor.   {To address this, we define the Hankel tensor basis in Definition \ref{def:HankelT} and establish concentration inequalities Lemmas \ref{lem:init_conc_m1}, \ref{lem:PTconc} under Hankel tensor sampling. Additionally, we propose a new sequential spectral initialization strategy in Algorithm \ref{alg:init} and provide theoretical guarantees for it in Lemma \ref{lem:init}.}  {Second, ScaledGD fails to consider the intrinsic Hankel structure of the tensor, leading to a high per-iteration computational complexity of $O(sn^2r)$. In contrast, leveraging our proposed  {Lemma~\ref{lem:HankelT_algebra},\ref{lem:HT_mul_1or2}}, which characterizes the interaction between Hankel structure and low-rank tensor decomposition, we transform the high-dimensional operations in tensor space into low-dimensional operations in low-complexity factor space. }  Furthermore, we introduce the fast computation techniques in Algorithm~\ref{alg:fast_comput}, achieving a significantly lower per-iteration complexity of $O(nr^2\log(n)+nr^3+sr^2)$ when $m=O(sr)$. For a detailed comparison between our algorithm and ScaledGD \cite{Tong2022}, please refer to Table \ref{tab:comp_methods}. 

Our work is of independent interest to low-rank Hankel 
 tensor completion tasks, which frequently arise in applications such as seismic reconstruction \cite{Qian2021,Trickett2013},  traffic estimation \cite{Wang2023a,Wang2021}, and image recovery \cite{Yamamoto2022,Yokota2018}, {offering a highly efficient optimization method for solving related problems.} 
Also, to the best of our knowledge, we are the first to establish rigorous recovery and linear convergence guarantees for the low-rank Hankel tensor completion problem. 

\noindent \textbf{Notations.} 
We denote vectors with bold lowercase letters, 
matrices with bold uppercase letters, tensors with bold calligraphic letters, and operators with calligraphic letters. For matrix $\vZ$, 
$\norm{\vZ}$ and $\norm{\vZ}_F$ denote its 
spectral norm, and Frobenius norm, respectively. Besides, define $\norm{\vZ}_{2,\infty}$ as the largest $\ell_2$-norm of its rows.   We define 
 the inner product of two matrices $\vZ_1$ and $\vZ_2$ as $\la\vZ_1,\vZ_2\ra=\mathrm{trace}(\vZ_1^H\vZ_2)$. For a tensor $\bcZ$, the inner product between two complex tensors is defined as 
$\langle\bcZ_{1},\bcZ_{2}\rangle = \sum_{i_1,i_2,i_3} \overline{\bcZ}_{1} (i_1,i_2,i_3) \bcZ_{2} (i_1,i_2,i_3).$ 
 We denote the identity matrix and operator as $\vI$ and $\I$, respectively. The adjoint of the operator $\A$ is denoted as $\A^*$. $\Real(\cdot)$ denotes the real part of a complex number. $\otimes$ denotes Kronecker product and $\circ$ denotes the outer product, for example: 
\begin{align*}
    [\vz_1\circ\vz_2\circ\vz_3](i_1,i_2,i_3)=\vz_1(i_1)\vz_2(i_2)\vz_3(i_3).
\end{align*}
All the numbering of the elements starts at zero. 
We denote $[n]$ as the set $\{0,1,\cdots,n-1\}$, where $n$ is a natural number.  For $a\in[n]$, $w_a$ is defined as the {cardinality} of the set  $\mathcal{W}_a=
 \{(j,k) | j+k=a, 0\leq j \leq n_1-1, 0\leq k\leq n_2-1 \} $. 
Next, we introduce additional notations for tensor algebra: 
\paragraph{Multilinear rank} 
The multilinear rank of a 
tensor 
$\bcZ$ is  defined as 
\begin{align*}
    \mulrank(\bcZ){=}(\rank(\cM_1(\bcZ)),\rank(\cM_2(\bcZ)),\rank(\cM_3(\bcZ))).
\end{align*}

\paragraph{Tensor matricization} 
Given a tensor $\bcZ\in\C^{n_1\times n_2\times s}$, the mode-$1$ matricization operation is defined as $$[\cM_1(\bcZ)]\big(i_1, i_2 + i_3 n_2\big) = \bcZ(i_1,i_2,i_3);$$ $\cM_2(\bcZ)$ and $\cM_3(\bcZ)$  are defined similarly. 
\paragraph{Mode-$i$ tensor product} The mode-$i$ product of a tensor $\bcZ\in\C^{n_1\times n_2\times s}$ and a matrix $\mM\in\C^{k\times n_1}$ is defined as, taking $i=1$ for example
\vspace{-5pt}
\begin{align*}
    [\bcZ\times_1\mM](j_1,i_2,i_3 )=\sum_{i_1}{\bcZ}(i_1,i_2,i_3)\mM(j_1,i_1),
    \vspace{-4pt}
\end{align*}
where $\bcZ\times_1\mM\in\C^{k\times n_2 \times s}$. Mode-$2$ and Mode-$3$ tensor product are similarly defined.  
\vspace{-3pt}
\paragraph{Tensor norms} 
The Frobenius norm is defined as $\|\bcZ\|_{F}=\sqrt{\langle\bcZ,\bcZ\rangle}$. 
With slight abuse of terminology, denote 
\begin{align*}
\sigma_{\max}(\bcZ) &= \max_{k=1,2,3} \sigma_{\max}(\cM_k(\bcZ)), \\ \sigma_{\min}(\bcZ) &= \min_{k=1,2,3} \sigma_{\min}(\cM_k(\bcZ)).
\end{align*} 

\paragraph{Tucker decomposition} 
 For a tensor $\bcZ$ with  $\mulrank(\bcZ)=(r_1,r_2,r_3)$, it has the following Tucker decomposition \cite{Tucker1966}:
\begin{align*}
    \bcZ=(\mL,\mR,\mV)\bcdot\bcS=\bcS\times_1\mL\times_2\mR\times_3\mV, 
\end{align*}
where $\bcS\in\C^{r_1\times r_2 \times r_3}$, $\mL\in\C^{n_1\times r_1}$, $\mR\in\C^{n_2\times r_2}$ and $\mV\in \C^{s\times r_3}$. Given a complex tensor $ \bcZ=(\mL,\mR,\mV)\bcdot\bcS$: 
\begin{align*}
\cM_1(\bcZ) &= \mL\cM_1(\bcS)(\mV\otimes\mR)^{T}, \\ \cM_2(\bcZ) &= \mR\cM_2(\bcS)(\mV\otimes\mL)^{T}, \\ \cM_3(\bcZ) &= \mV\cM_3(\bcS)(\mR\otimes\mL)^{T}. 
\end{align*}
{Note that for a complex tensor, its mode-$i$ matricization is not $\cM_1(\bcZ) = \mL\cM_1(\bcS)(\mV\otimes\mR)^{H}$ but still $\cM_1(\bcZ) =\mL\cM_1(\bcS)(\mV\otimes\mR)^{T}$, taking $i=1$ for example.} 

\section{Model formulation and algorithm} 
\subsection{Problem formulation}
\label{sec:pb_formul}
We aim to recover the desired multiple signals through Hankel tensor completion. First, we construct a Hankel tensor 
$\H(\mX_\star)\in\C^{n_1\times n_2\times s}$ as illustrated in Fig. \ref{fig:diagram_mulHT}
where $\H:\C^{s\times n} \rightarrow\C^{n_{1}\times n_{2} \times s}$ is the Hankel tensor lifting operator.
The lifted tensor $\H(\mX_\star)$ has the low-rank Tucker decomposition 
$$
\H(\mX_\star)=(\mE_L,\mE_R,\mB)\bcdot\bcD,
$$
as introduced earlier. {The exact multilinear rank of  $\H(\mX_\star)$ is proved in the following Lemma. {This lemma demonstrates that the 
previous decomposition is equivalent to the exact multilinear rank $(r,r,r)$ under some conditions. Thus $\H(\mX_\star)$ must have a low-rank Tucker decomposition in \eqref{eq:grdtruth-oth-tucker} (which is essentially HOSVD \cite{DeLathauwer2000}).}
\begin{lemma}[The multilinear  rank of $\H(\mX_\star)$] \label{lem:mulrank_HT}
When $\rank(\mB)=r$,  
$p_k=e^{(\imath2\pi f_k-\tau_k)}$ are distinct for ${k=0,1,\cdots,r-1}$ 
 and $r\ll\min\{s,n\}$, the multilinear rank of the Hankel tensor  
 $\H(\mX_\star)
 {=}(\mE_L,\mE_R,\mB)\bcdot\bcD$ 
 satisfies  
 \begin{align*}
     \mulrank{(\H(\mX_\star))}=(r,r,r). 
 \end{align*}
\end{lemma}
\begin{proof}
    See Appendix~\ref{pf:mulrank_HT}.
\end{proof}
}




 
{Following the route in Hankel-lift approaches \cite{Cai2018,Cai2023} in signals reconstruction, we construct the loss function in the lifted tensor domain to recover $\mX_\star$, which is equivalent to the following rank constraint {\em weighted least square} problem} \footnote{ {We construct the loss in the lifted tensor domain as: 
\begin{align*}
&\|\H{\P}_\Omega\mX-\H{\P}_\Omega\mX_\star\|_F^2=\sum_{(k,a)\in\Omega}\sum_{i+j=a}\([\H\mX](i,j,k)-[\H\mX_\star](i,j,k)\)^2
          \\&=\sum_{(k,a)\in\Omega}w_a (\mX(k,a)-\mX_\star(k,a))^2=\langle\P_{\Omega}(\D(\mX-\mX_\star)),\D(\mX-\mX_\star)\rangle,
      \end{align*}where $w_a$ is the length of the $k$-th skew-diagonal of an $n_1\times n_2$ matrix  defined in notations.  $\D^2=\H^*\H :\C^{s \times n} \rightarrow\C^{s\times n} $ is a linear {reweighting} operator such that $[\D(\mM)](:,a)=\sqrt{w_a}\mM(:,a)$ for $a\in[n]$.}}.  
\begin{align*}
\min_{\mX\in\C^{s\times n}}&\la \P_{\Omega}\(\D(\mX-\mX_\star)\),\D(\mX-\mX_\star)\ra
\\&\mbox{s.t.}~\mulrank(\H(\mX))=\vr,\numberthis\label{eq:low_rank_H}
\end{align*}
where $\vr=(r,r,r)$, $\D$ is a linear reweighting operator.  {This modeling admits a well-defined convergence analysis in tensor space, which is a generalization of  \cite{Cai2018,Cai2023}.}  The problem \eqref{eq:low_rank_H} can be reformulated as the following problem, making the substitutions that $\mY_\star= \D(\mX_\star)$ and $\mZ= \D(\mX)$
\begin{align*}
\min_{\mZ\in\C^{s\times n}}\la \P_{\Omega}\(\mZ-\mY_\star\),\mZ-\mY_\star\ra~\mbox{s.t.}~\mulrank(\G(\mZ))=\vr,\numberthis\label{eq:low_rank_G}
\end{align*}
where $\G=\H\D^{-1}$, $\D$ is invertible, and it is evident that  $\G^*\G=\I$. 
Let $\bcZ = \G(\mZ)$, and we can utilize the low-rank Tucker decomposition to remove the multilinear rank constraint $$\bcZ=(\mL,\mR,\mV)\bcdot\bcS,$$ where $\mL \in\C^{n_1\times r}$, {$\mR \in\C^{n_2\times r}$}, $\mV \in\C^{s\times r}$ and $\bcS \in\C^{r\times r \times r}$.
Besides, the Hankel structure of the tensor $\bcZ=\G(\mZ)$ is enforced by the following constraint:
\begin{align*}
(\I-\G\G^*)(\bcZ)=\bm{0},
\end{align*}
where $\G\G^*$ is a projector that maps a tensor to a Hankel tensor. Define the factor quadruple as $\mF=(\mL,\mR,\mV,\bcS)$, and  \eqref{eq:low_rank_G} can be rewritten as
\begin{align*}
&\min_{\mF=(\mL,\mR,\mV,\bcS)}\|\P_{\Omega}\(\G^*\((\mL,\mR,\mV)\bcdot\bcS\)-\mY_\star\)\|_F^2 \\
&\quad\mbox{s.t.}\quad(\I-\G\G^*)((\mL,\mR,\mV)\bcdot\bcS) = \bm{0},\numberthis\label{eq:constrained}\end{align*}
where we insert the facts $\G(\mZ)=(\mL,\mR,\mV)\bcdot\bcS$ and $\mZ=\G^*\(\G(\mZ)\)=\G^*\((\mL,\mR,\mV)\bcdot\bcS\)$ 
into \eqref{eq:low_rank_G}. 
Last, we consider a penalized version of \eqref{eq:constrained}: 
\begin{align*}
\min_{\vF}~f(\vF)\coloneqq &\frac{1}{2p}\|\P_{\Omega}\(\G^*\((\mL,\mR,\mV)\bcdot\bcS\)-\mY_\star\)\|_F^2\\
&+\frac{1}{2}\ln(\I-\G\G^*)((\mL,\mR,\mV)\bcdot\bcS)\rn_F^2, \numberthis\label{eq:penalized}
\end{align*}
 where $p=\frac{m}{sn}$ is the observation ratio. We interpret \eqref{eq:penalized} as that
 one uses a low-rank tensor $\bcZ=(\mL,\mR,\mV)\bcdot\bcS$ with Hankel structure penalty to estimate the Hankel tensor $\bcZ_\star=\H(\mX_\star)$.
 {
\begin{remark}
    From Lemma~\ref{lem:mulrank_HT}, we know $\mulrank(\bcZ_\star)=(r,r,r)$.
Thus we apply the  Tucker decomposition $\bcZ=(\mL,\mR,\mV)\bcdot\bcS$ where 
$\bcS\in\C^{r\times r \times r}$. 
In the practical scenario, the latent dimensions may differ across modes, and it is better to apply a general formulation $\bcS\in\C^{r_1\times r_2 \times r_3}$ to accommodate asymmetric structures. 
\end{remark}
}
{
\begin{remark}[Determining the rank]
Following single/multi-measurement spectral compressed sensing via Hankel-lift approach in  \cite{Cai2018,Cai2019,Zhang2018,Cai2023}, we assume the rank is known a priori; however, determining an appropriate rank remains a crucial challenge in practical scenarios. 
We may adopt the "rank increment" strategy proposed in  \cite{Yokota2018, Yamamoto2022}. 
In \cite{Yokota2018, Yamamoto2022}, the authors demonstrated how to select the incremental mode $m'$  and adjust the corresponding rank  $r_{m'} $. Integrating this strategy with the gradient descent algorithm under low-rank decomposition could further enhance its practical applicability, which we aim to explore in future work.
\end{remark}
}
\vspace{-6mm}
\subsection{Algorithm: ScalHT} \label{sec:def-algorithm}
We introduce a scaled (projected) gradient descent algorithm \cite{Tong2022} to address the Hankel tensor completion problem, which we name ScalHT, as detailed in Algorithm \ref{alg:ScalHT}. {Scaled gradient descent \cite{Tong2022} is preferred for its condition number $\kappa$-independent convergence, even under moderate ill-conditioning, while vanilla gradient descent's performance degrades with $\kappa$. This is critical here, as the lifted Hankel tensor $\mathcal{H}(\mathbf{X}_\star)$ from multi-measurement spectral sparse signals typically exhibits moderately large $\kappa$ in our preliminary simulations. Besides, a practical example of an ill-conditioned Hankel matrix in DOA is provided in \cite{Cai2025}, when the spatial frequencies are close to each other. Such ill-conditioned cases in Hankel matrices can be safely generalized to Hankel tensors.}
 Specifically, the update rules are outlined below to minimize the loss function \eqref{eq:penalized}:
\begin{align}
\begin{split}
\mL_{+} &= \mL - \eta\nabla_{\mL}f(\mF )\big(\breve{\mL}^{H} \breve{\mL}  \big)^{-1}, \\
\mR _{+} &= \mR  - \eta\nabla_{\mR}f(\bF)\big(\breve{\mR}^{H} \breve{\mR}  \big)^{-1}, \\
\mV _{+}& = \mV  - \eta\nabla_{\mV}f(\bF)\big(\breve{\mV}^{H} \breve{\mV}  \big)^{-1}, \\
\bcS _{+}& 
= \bcS  - \eta\big((\mL ^{H}\mL )^{-1},(\mR ^{H}\mR )^{-1},(\mV ^{H}\mV )^{-1}\big)\bcdot\nabla_{\bcS}f(\bF),
\end{split}\label{eq:ScalHT}
\end{align}
where $\breve{\mL}\in \C^{s n_2 \times r}$, $\breve{\mR}\in \C^{s n_1 \times r}$ and  $\breve{\mV}\in \C^{n_1 n_2\times r}$  are defined as:
\begin{align} \label{eq:breve_lrv}
\begin{split}
\breve{\mL} &\coloneqq(\overline{\mV} \otimes\overline{\mR} )\cM_{1}(\bcS )^{H}, \\
\breve{\mR} &\coloneqq(\overline{\mV} \otimes\overline{\mL} )\cM_{2}(\bcS )^{H},\\
\breve{\mV} &\coloneqq(\overline{\mR} \otimes\overline{\mL} )\cM_{3}(\bcS )^{H}.
\end{split}
\end{align}
In \eqref{eq:ScalHT}, the derivatives of $f(\mF)$ are:
\begin{align}
\begin{split}
\nabla_{\mL}f(\mF)&=\cM_{1}\({p}^{-1}\G\P_{\Omega}(\G^{*}\bcZ-\mY_\star)+(\I-\G\G^{*})(\bcZ)\)\breve{\mL},\\
\nabla_{\mR}f(\mF)&=\cM_{2}\({p}^{-1}\G\P_{\Omega}(\G^{*}\bcZ-\mY_\star)+(\I-\G\G^{*})(\bcZ)\)\breve{\mR}, \\
\nabla_{\mV}f(\mF)&=\cM_{3}\({p}^{-1}\G\P_{\Omega}(\G^{*}\bcZ-\mY_\star)+(\I-\G\G^{*})(\bcZ)\)\breve{\mV}, \\
\nabla_{\bcS}f(\mF)&=(\mL^H,\mR^H,\mV^H)\bcdot\big({p}^{-1}\G\P_{\Omega}(\G^{*}\bcZ-\mY_\star)\\
&\quad+(\I-\G\G^{*})(\bcZ)\big), 
\end{split}\label{eq:loss_derivative}
\end{align}
where we denote $\bcZ=(\mL,\mR,\mV)\bcdot\bcS$ for simplicity. Besides, their fast computations are presented in Algorithm \ref{alg:fast_comput}. 


\begin{algorithm}[t]
\caption{Low-rank Hankel Tensor Completion via Scaled  Gradient Descent (ScalHT)}
\label{alg:ScalHT}
\begin{algorithmic} 
 \Statex  {Partition} ${\Omega}$ into disjoint sets ${\Omega}_0,\cdots,{\Omega}_K$ of equal size $\hat{m}$, and let $\hat{p}=\frac{\hat{m}}{sn}$. And set $\bcZ^0= \hat{p}^{-1}\G\P_{{\Omega}_{0}}(\mY_\star)$.
\Statex \textbf{Initialization:} \Statex Initialize $\mF^0=(\mL^0,\mR^0,\mV^0,\bcS^0)$ sequentially via Algorithm \ref{alg:init}. 





 
\For{$k=0,1,\cdots,K$}
\Statex $f^k(\mF)$ is shown in \eqref{eq:penalized} where the set $\Omega$ is replaced with $\Omega_k$. 
The derivatives of $f^k(\mF)$ are computed via Algorithm \ref{alg:fast_comput}. 
\begin{align*}
   1.~{{\mL'}^{k+1}}=&\mL^k-\eta \nabla_{\mL} f^k(\mF^k)\big((\breve{\mL}^{k})^{H} \breve{\mL}^{k}\big)^{-1},
\\{{\mR'}^{k+1}}=&\mR^k-\eta \nabla_{\mR} f^k(\mF^k)\big((\breve{\mR}^{k})^{H} \breve{\mR}^{k}\big)^{-1},
\\{\mV}^{k+1}=&\mV^k-\eta \nabla_{\mV} f^k(\mF^k)\big((\breve{\mV}^{k})^{H} \breve{\mV}^{k}\big)^{-1},
\\\bcS^{k+1}=&\bcS^{k}-\eta\big( \big((\mL^k)^H\mL^k\big)^{-1}, \((\mR^k)^H\mR^k\)^{-1},
    \\&
     \((\mV^k)^H\mV^k\)^{-1}\big)\bcdot\nabla_{\bcS}f^k(\mF^k).
\end{align*}
\\ \quad~~$2.~
(\mL^{k+1},\mR^{k+1})=\PB{{\mL'}^{k+1},{\mR'}^{k+1}}.$ 
\\

\EndFor
\Statex \textbf{Output:} 
$\mX^K = \D^{-1}\G^*\big((\mL^K,\mR^K,\mV^K)\bcdot\bcS^K\big)$. 
\end{algorithmic}
\end{algorithm}

\begin{algorithm}[t]
\caption{Sequential Spectral Initialization}
\label{alg:init}
\begin{algorithmic} 
\Statex 
$\mathrm{SVD}_r(\cdot)$ returns the top-$r$ left singular vectors of a matrix. 
\begin{align*}
    &1.~{\mL'}^0=\mathrm{SVD}_{r}(\cM_1(\bcZ^0)),~{\mR'}^0=\mathrm{SVD}_{r}(\cM_2(\bcZ^0)).
    \\& 2.~\mV^0=\mathrm{SVD}_r(\cM_3(\bcZ^0\times_1({\mL'}^0)^H)), \mbox{based on ${\mL'}^0$}. 
    \\& 3.~\bcS^0=\big(({\mL'}^0)^H,({\mR'}^0)^H,(\mV^0)^H\big)\bcdot \bcZ^0.
    \\& 4.~(\mL^0,\mR^0)=\PB{{\mL'}^0,{\mR'}^0}.
\end{align*}
\Statex \textbf{Output:} 
$\mF^0=(\mL^0,\mR^0,\mV^0,\bcS^0)$.
\end{algorithmic}
\end{algorithm}



In Algorithm \ref{alg:ScalHT}, the observation set $\Omega$ is divided into $K+1$ separate groups, each containing the same number of elements, $\hat{m}$. This method of dividing the dataset is frequently utilized in the study of matrix completion \cite{Cherapanamjeri2017, Jain2013}, as well as in the context of Hankel matrix completion \cite{Cai2019, Zhang2019,Li2024a}. By employing this sample-splitting strategy, the current observation set remains independent of previous iterations, which simplifies the theoretical analysis.

To ensure the Hankel tensor completion can be recovered, it is crucial to maintain the incoherence property of the factors $\mL$ and $\mR$ {as shown in Definition \ref{def:incoh}} throughout the iterations. Inspired by \cite{Tong2022}, we apply the scaled projection operator as follows: 
\begin{align*}
    (\mL,\mR)=\P_{B}(\mL',\mR'), \numberthis \label{eq:incoh_proj}
\end{align*}
where $\mL',\mR'$ are some iterates during the trajectories, and 
\begin{align*}
    \mL(i_1,:)&=\min\{1,\frac{B}{\sqrt{n}\|\mL'(i_1,:)\breve{\mL}'\|_2}\}\mL'(i_1,:),
    \\ \mR(i_2,:)&=\min\{1,\frac{B}{\sqrt{n}\|\mR'(i_2,:)\breve{\mR}'\|_2}\}\mR'(i_2,:),
\end{align*}
 for $i_1\in[n_1]$, $i_2\in[n_2]$, $B> 0$ is the projection radius, and {$\breve{\mL}',\breve{\mR}'$} are defined similarly as \eqref{eq:breve_lrv} from $(\mL',\mR',\mV,\bcS)$. We emphasize that the recovery guarantees for this problem do not require the incoherence of $\mV_\star$, which is defined in \eqref{eq:grdtruth-oth-tucker}, and thus there is no need to project $\mV$. In contrast, in ScaledGD \cite{Tong2022} for tensor completion, all factors $\mL$, $\mR$, and $\mV$ must be projected onto the incoherence set.


Next, we introduce the initialization of ScalHT. The previous off-diagonal spectral initialization methods used in {tensor completion} \cite{Tong2022,Cai2022b} are not applicable to the Hankel tensor completion problem. The theoretical guarantees of off-diagonal spectral initialization rely on the sampling independence between each dimension, whereas the Hankel sampling introduces statistical dependence between the dimensions of the Hankel structure.  
Instead, towards Hankel tensor completion, we propose a sequential spectral initialization strategy 
as shown in Algorithm \ref{alg:init}. 

Let $\bcZ^0= p^{-1}\G\P_{{\Omega}_{0}}(\mY_\star)$. First, we obtain ${\mL'}^0,{\mR'}^0$ using the top-$r$ left singular vectors of $\cM_i(\bcZ^0)$ where $i=1,2$. And let $(\mL^0,\mR^0)=\PB{{\mL'}^0,{\mR'}^0}$ to maintain the incoherence of $\mL^0$ and $\mR^0$. However, we don't initialize $\mV^0$ from $\cM_3(\bcZ^0)$. Through our analysis, 
$\|\cM_3(\bcZ^0-\bcZ_\star)\|$ is large because $\|\cM_3(\bcH_{k,j})\|=1$, as shown in Definition \ref{def:HankelT}. 
Instead, we initialize $\mV^0$ via the top-$r$ left singular vectors of one intermediary quantity $\cM_3(\bcZ^0\times_1({\mL'}^0)^H)$, which depends on the previous estimate ${\mL'}^0$. 
Thus, we name this method as sequential initialization strategy. The guarantees of our sequential spectral initialization are shown in Lemma~\ref{lem:init} without using $\|\cM_3(\bcH_{k,j})\|=1$. 

Finally, we enter into the stage that iterative updating on the factors with projection, and see Alg. \ref{alg:ScalHT} for details. 

\begin{table*}[!t]	
\renewcommand\arraystretch{1.1}
		\begin{center}			
   \caption{ 
   Fast computation of main terms in ScalHT. ($r$ and $\log(n)$ are seen as constants)} 
     \begin{tabular}{c|c|c|c|c}
   \hline
				\textbf{Main terms}  & \textbf{Fast computation} & \textbf{ Previous complexity}  & \textbf{Current complexity}& \textbf{Location}

    \\
	\hline	\hline		
  	  {$\mZ=\G^{*}\((\mL,\mR,\mV)\bcdot\bcS\)$ }  & $\mZ=\mV\mB^H$,  Lemma~\ref{lem:HankelT_algebra}.a & $O(sn^2)$ 
& $O(n)$ 
& All the derivatives of $f(\mF)$
\\ \hline 
           $\G(\mZ)\times_3 \mV^H$
           & $\G((\mV^H\mV)\mB^H)$, Lemma~\ref{lem:HankelT_algebra}.b & $O(sn^2)$  & $O(s+n)$ & $\nabla_{\mL}f(\mF),\nabla_{\mR}f(\mF),\nabla_{\bcS}f(\mF)$
           \\ \hline
           $\G(\mM)\times_3 \mV^H$ 
            
            & $\G(\mV^H\mM)$,  Lemma~\ref{lem:HankelT_algebra}.b
            & $O(sn^2)$ & $O(m)$ & $\nabla_{\mL}f(\mF),\nabla_{\mR}f(\mF),\nabla_{\bcS}f(\mF)$
		\\ \hline
		
   $\G(\hat{\mE})\times_1\mL^H\times_2\mR^H$
    
    & $\overline{\bcW}\times_3\hat{\mE}$,  Lemma~\ref{lem:HankelT_algebra}.c
    & $O(n^2)$  
   & $O(n)$
   & $\nabla_{\bcS}f(\mF)$
  \\ \hline
  $\G({\mZ})\times_1\mL^H\times_2\mR^H$
  & $\overline{\bcW}\times_3 \mB^H \times_3 \mV$,  Lemma~\ref{lem:HankelT_algebra}.c & $O(sn^2)$ 
   & $O(s+n)$
   & $\nabla_{\mV}f(\mF)$
    \\ \hline
  $\G({\mM})\times_1\mL^H\times_2\mR^H$
   
   & $\overline{\bcW}\times_3 \mM$,  Lemma~\ref{lem:HankelT_algebra}.c
   & $O(sn^2)$ 
   & $O(m)$
   & $\nabla_{\mV}f(\mF)$
   \\ \hline
  $\G(\hat{\mE})\times_1\mL^H$
  &  
  Fast convolution {(FFT)},  Lemma~\ref{lem:HT_mul_1or2}
  & $O(n^2)$  
  &$O(n)$
  & $\nabla_{\mL}f(\mF),\nabla_{\mR}f(\mF)$ 
  \\ \hline
  $\G(\hat{\mE})\times_2\mR^H$

  &  
  Fast convolution {(FFT)}, Lemma~\ref{lem:HT_mul_1or2}
  & $O(n^2)$  
  &$O(n)$
    & $\nabla_{\mL}f(\mF),\nabla_{\mR}f(\mF)$
   \\ \hline$\breve{\mL}^H\breve{\mL},\breve{\mR}^H\breve{\mR},\breve{\mV}^H\breve{\mV}$ 
 
 & Lemma~\ref{lem:fast_scaleterm}
 & $O(n(s+n))$ 
   & {$O(s+n)$}
   & Scaled terms in (8) 
    \\
    \hline
 			\end{tabular}
			\label{tab:fast_computation}
   \end{center}
	\end{table*}

    \begin{algorithm}[t]
\caption{Fast Computation of the Gradient of $f(\mF)$} 
\label{alg:fast_comput}
\begin{algorithmic} 
\Statex \textbf{Input:} $\mF=(\mL,\mR,\mV,\bcS).$
\Statex 
1. Compute $\bcW$, 
$\mB$, $\mM$ and $\hat{\mE}$ in \eqref{eq:conv_W}, 
\eqref{eq:B}, \eqref{eq:M} and \eqref{eq:E_hat}.
\Statex
2. $\nabla_{\mL}f(\mF)=\cM_1\(\G(\hat{\mE})\times_2\mR^H\)\cM_1(\bcS)^H+\mL(\breve{\mL}^H\breve{\mL})$.
\Statex
3. $\nabla_{\mR}f(\mF)=\cM_2\(\G(\hat{\mE})\times_1\mL^H\)\cM_2(\bcS)^H+\mR(\breve{\mR}^H\breve{\mR})$.
\Statex
4. $\nabla_{\mV}f(\mF)=\mM\mB-\mV\(\mB^H\mB\)+\mV\(\breve{\mV}^H\breve{\mV}\)$.
\Statex
5. $\nabla_{\bcS}f(\mF) =\overline{\bcW}\times_3\hat{\mE}+\(\mL^H\mL,\mR^H\mR,\mV^H\mV\)\bcdot\bcS$.
\Statex \textbf{Output:} 
$\nabla_{\mL}f(\mF),\nabla_{\mR}f(\mF),\nabla_{\mV}f(\mF),\nabla_{\bcS}f(\mF)$.
\end{algorithmic}
\end{algorithm}

   \section{Fast computation}  
In this section, we introduce the fast computation of Alg. \ref{alg:ScalHT} (ScalHT). These fast computation rules deeply leverage the interaction between  the Hankel structure and low-rank (Tucker) decomposition, resulting in a computational complexity per iteration of {$O\(nr^2\log(n)+nr^3+sr^2\)$ when the number of observations $m=O(sr)$, which corresponds to the degree of freedom of this problem as shown in Remark~\ref{rmk:freedom}.} We see $r$ and $\log(n)$ as constants as $r,\log(n)\ll \min\{s,n\}$, and focus on the ambient dimensions $s$ and $n$, thus the computational complexity further simplifies to $O(s+n)$, highlighting the superior efficiency of ScalHT. 
The fast computations of the gradient of $f(\mF)$ 
are detailed in Algorithm \ref{alg:fast_comput}. Additionally, we summarize the main terms in the gradient and their fast computations 
in Table \ref{tab:fast_computation}. 

{
\subsection{Key computational technique}
The key idea to  accelerate the computation in our algorithm can be summarized as follows:


\emph{
Through Lemma~\ref{lem:HankelT_algebra},\ref{lem:HT_mul_1or2}, which leverages the interplay between Hankel structure and low-rank tensor decomposition
, the high-dimensional operations in tensor space are transformed into the low-dimensional operations in low-complexity factor space, resulting in a significant reduction in computational complexity 
by a factor of $O(\min\{s, n\})$.
}

We present low-rank Hankel tensor algebra Lemma~\ref{lem:HankelT_algebra}, which helps transform high-dimensional operations into low-dimensional 
operations on factors.  
}

\begin{lemma}[Low-rank Hankel tensor algebra] \label{lem:HankelT_algebra}
    Let $\mL\in\C^{n_1\times r}, \mR\in\C^{n_2\times r}, \mV\in\C^{s\times r}$, and $\bcS\in\C^{r\times r\times r}$. 
 \begin{itemize}
  \item [a)] {Let $\mZ=\G^*((\mL,\mR,\mV)\bcdot\bcS)\in\C^{s\times n}$}, 
  we have
     \begin{align*}
         \mZ=\G^*((\mL,\mR,\mV)\bcdot\bcS)=\mV\mB^H, \numberthis \label{eq:fast_HTproj_tt}
     \end{align*}
     {where 
      \begin{align*}
   \mB=\cM_3(\overline{\bcW})\cM_3(\bcS)^H \in\C^{n\times r}
   \numberthis \label{eq:B}, 
   \end{align*}
     and $\bcW\in\C^{r\times r\times n}$, for $j_1,j_2\in[r]$, $a\in[n]$,
\begin{align*}
    \bcW(j_1,j_2,a)=\frac{1}{\sqrt{w_{a}}}[\mL(:,j_1)\ast\mR(:,j_2)](a). \numberthis  \label{eq:conv_W}
\end{align*}
     }
 \item [b)] For  $\mE\in\C^{s\times n}$ and $\tilde{\bcZ}\in\C^{n_1\times n_2\times r}$:
\begin{align*}
\G(\mE)\times_{3}\mV^H&=\G(\mV^H\mE),\label{eq:Hankellift_mul3} \numberthis 
    \\\G^{*}(\tilde{\bcZ}\times_{3}\mV)&=\mV\G^{*}(\tilde{\bcZ}). \label{eq:DeHankel_mul3} \numberthis 
\end{align*}
    
\item [c)] For $\tilde{\mE}\in\C^{k\times n}$, we have
\begin{align*}
\G(\tilde{\mE})\times_{1}\mL^H\times_{2}\mR^H=\overline{\bcW}\times_{3}\tilde{\mE}\in\C^{r\times r \times k},\numberthis\label{lem:HT_mul_1and2}
\end{align*}
where $\bcW$ is defined in \eqref{eq:conv_W}.
 \end{itemize}
\end{lemma}
\begin{proof}
    See Appendix \ref{pf:HankelT_algebra}.
\end{proof}
{
\begin{remark}
 In Lemma~\ref{lem:HankelT_algebra}.a,   $\bcW\in\C^{r\times r\times n}$ can be computed via $r^2$ fast convolution {by FFT} with $O(nr^2\log(n))$ flops, and $\mB=\cM_3(\overline{\bcW})\cM_3(\bcS)^H \in\C^{n\times r}$ costs $O(nr^3)$ flops. Besides, $\mZ=\mV\mB^H$ is not computed explicitly during gradient computation. 
\end{remark}
Then we explain Lemma~\ref{lem:HankelT_algebra} in more detail:
\begin{itemize}
    \item [$\bullet$] {Lemma~\ref{lem:HankelT_algebra}.a} tells us that the low-rank tensor decomposition retains a low-complexity representation even after applying the Hankel adjoint mapping. 
     \item [$\bullet$]  {Lemma~\ref{lem:HankelT_algebra}.b} tells us that high-dimensional multiplication associated with the multi-measurement vector dimension (the third dimension) can be efficiently implemented through direct multiplication with the low-complexity factors. 
      \item [$\bullet$]  {Lemma~\ref{lem:HankelT_algebra}.c} tells us that high-dimensional multiplication associated with two Hankel-structured dimensions jointly can be efficiently implemented using convolutions of the low-rank factors. 
\end{itemize}
}
{
Besides, we introduce  Lemma~\ref{lem:HT_mul_1or2}:
\begin{itemize}
    \item [$\bullet$] Lemma~\ref{lem:HT_mul_1or2} tells us that high-dimensional multiplication associated with a single Hankel-structured dimension can be implemented using fast convolutions by FFT.
\end{itemize}

}
\begin{lemma}[Multiplication involving single dimension of Hankel mode
]\label{lem:HT_mul_1or2}
     Let $\mL\in\C^{n_1\times r}$, $\mR\in\C^{n_2\times r}$, and $\mE\in\C^{r\times n}$. The computation of the following terms 
     $$\G(\mE)\times_{1}\mL^H~\mbox{and}~ \G(\mE)\times_{2}\mR^H$$
     can be realized by $r^2$ fast convolution via FFT, which cost $O(nr^2\log(n))$ flops. 
\end{lemma}
\begin{proof}
     We take the computation of $\G(\mE)\times_{1}\mL^H \in \C^{r\times n_2\times r}$ for example. For $j_1,i_3\in[r]$, $i_2\in[n_2]$
    \begin{align*}
        [\G(\mE)\times_{1}\mL^H](j_1,i_2,i_3)&{=}\sum_{i_1=0}^{{n_1}-1}\frac{1}{\sqrt{w_{i_1{+}i_2}}}\mE(i_3,i_1{+}i_2)\overline{\mL}(i_1,j_1)
        \\&=[\tilde{\mE}(i_3,:)\ast\overline{\mL'}(:,j_1)](n_1+i_2),
    \end{align*}
    where $\tilde{\mE}=\D^{-1}(\mE)$ is the weighted version of $\mE$ and $\mL'(:,j_1)$ is the vector that reverses the order of $\mL(:,j_1)$. The previous computation can be realized via $r^2$ fast convolution, which costs $O(nr^2\log(n))$ flops.   
\end{proof}

{
We take some main terms in the gradient computation, for example, to show how to accelerate the computation through our Lemma~\ref{lem:HankelT_algebra},\ref{lem:HT_mul_1or2}:
\begin{itemize}
    \item [$\bullet$] By Lemma~\ref{lem:HankelT_algebra}.a, $\mZ=\G^*((\mL,\mR,\mV)\bcdot\bcS)=\mV\mB^H$. The direct computation of it 
costs 
$O(sn^2r)$ flops. As $\mZ$ is not computed explicitly, we only need to compute $\mB$, which costs $O(nr^2(r+\log(n)))$ flops.   
    \item [$\bullet$]   By Lemma~\ref{lem:HankelT_algebra}.b, $\G(\mZ)\times_3 \mV^H= \G(\mV\mB^H)\times_3 \mV^H=\G((\mV^H\mV)\mB^H)$.  The direct computation of it 
costs $O(sn^2r)$ flops, which is reduced to $O((s+n)r^2)$ flops.
    \item [$\bullet$]  By Lemma~\ref{lem:HankelT_algebra}.c,  $\G(\mZ)\times_{1}\mL^H\times_{2}\mR^H=\overline{\bcW}\times_3\mZ=\overline{\bcW}\times_3\mB^H\times_3\mV$. The direct computation of it 
costs $O(sn^2r)$ flops, which is reduced to $O(nr^2\log(n)+(s+n)r^3)$ flops.
    \item [$\bullet$] By Lemma~\ref{lem:HT_mul_1or2},   $\G(\hat{\mE})\times_1\mL^H$ can be implemented via fast convolution (FFT), where $\hat{\mE}\in\C^{r\times n}$ is defined later. The direct computation of it costs $O(n^2r^2)$ flops, which is reduced to $O(nr^2\log(n))$ flops.  
\end{itemize}
These main terms and other similar terms as well as the scaled terms $\breve{\mL}^H\breve{\mL},\breve{\mR}^H\breve{\mR},\breve{\mV}^H\breve{\mV}$'s fast computations are summarized in Table \ref{tab:fast_computation}. Last, we emphasize the role of observation sparsity: 
\begin{remark}[Observations’ sparsity]
     The sparsity of the observations $m=O(sr)$ is also leveraged. 
     We decouple the per-iteration computational complexity to $O(s+n)$ (seeing $r$ as a constant) by integrating the Hankel structure, low-rank tensor
 decomposition, and observation sparsity.
\end{remark}
}

\subsection{Fast computation of the gradient}
We introduce the fast computation rules of the gradient of $f(\mF)$ in this subsection and summarize them in Algorithm \ref{alg:fast_comput}. Also, we provide the computational complexity analysis. 
{
First, we introduce some intermediary notations. 
Denote a $m$-sparse matrix $\mM\in\C^{s\times n}$ as 
\begin{align*}
    \mM= p^{-1}\P_{\Omega}(\mZ-\mY_\star), \numberthis \label{eq:M}
\end{align*}
where $\P_{\Omega}(\mZ)=\P_{\Omega}(\mV\mB^H)$ costs $O(mr)$ flops as only $m$ entries of $\mV\mB^H$ need to be explicitly computed. Besides, denote $\hat{\mE}\in\C^{r\times n}$ as 
\begin{align*}
    \hat{\mE}=\mV^H(\mM-\mZ)=\mV^H\mM-(\mV^H\mV)\mB^H, \numberthis \label{eq:E_hat}
\end{align*}
where $\mZ=\mV\mB^H$ is defined in \eqref{eq:fast_HTproj_tt}.  In \eqref{eq:E_hat}, $\mV^H\mM$ costs $O(mr)$ flops as $\mM$ is $m$-sparse, and $(\mV^H\mV)\mB^H$ costs $O((s+n)r^2)$ flops. Now we begin the derivation of the fast computation of the gradient. 
} 

\vspace{4.5mm}
\noindent\emph{1) Fast computation of { $\nabla_{\mL}f(\mF)$}}
\vspace{2mm}

Recalling the definitions of $\mM$, $\mZ$ and $\bcZ=(\mL,\mR,\mV)\bcdot\bcS$, we reformulate $\nabla_{\mL}f(\mF)$ as: 
 \begin{align*}
     &\nabla_{\mL}f(\mF)=\cM_{1}\(\G\big({p}^{-1}\P_{\Omega}(\G^{*}\bcZ-\mY_\star)\big)+(\I-\G\G^{*})(\bcZ)\)\breve{\mL}
     \\&=\cM_1(\G(\mM-\mZ))\breve{\mL}+\cM_1(\bcZ)\breve{\mL}
          \\&=\cM_1\(\G\big((\mM-\mZ)\big)\times_2\mR^H\times_3\mV^H\)\cM_1(\bcS)^H+\mL(\breve{\mL}^H\breve{\mL})
     \\&=\cM_1\(\G(\hat{\mE})\times_2\mR^H\)\cM_1(\bcS)^H+\mL(\breve{\mL}^H\breve{\mL}), \numberthis \label{eq:dfl_split}
 \end{align*}
 where  
 the third equality results from $\cM_1(\bcZ)
 =\mL\breve{\mL}^H$ and the fact \eqref{eq:tensor_properties_l} that $\cM_1(\bcX)\breve{\mL}
=\cM_1(\bcX\times_2 \mR^H\times_{3}\mV^H)\cM_1(\bcS)^H$. The last equality results from Lemma~\ref{lem:HankelT_algebra}.b, 

The computation of $\nabla_{\mL}f(\mF)$ costs $O(nr^2\log(n)+nr^3+sr^2+mr)$ flops in total. 
Computing $\hat{\mE}$ costs $O(mr+(s+n)r^2)$ flops. 
 In \eqref{eq:dfl_split},  $\G(\hat{\mE})\times_2\mR^H$ is computed via $r^2$ fast convolution with  $O(nr^2\log(n))$ flops from Lemma~\ref{lem:HT_mul_1or2}, $\mL(\breve{\mL}^H\breve{\mL})$ costs $O((s+n)r^2+r^4)$ flops from Lemma~\ref{lem:fast_scaleterm}. Let  $\hat{\bcW}=\G(\hat{\mE})\times_{2}\mR^H\in\C^{n_1\times r\times r}$, and $\cM_1(\hat{\bcW})\cM_1(\bcS)^H$ costs $O(nr^3)$ flops.  

Similar results hold for $\nabla_{\mR}f(\mF)$, and we omit this for simplicity. 
 

\vspace{4.5mm}
\noindent\emph{2) Fast computation of  $\nabla_{\mV}f(\mF)$}
\vspace{2mm}

The fast computation of  $\nabla_{\mV}f(\mF)$ exhibits some differences. 
Following the third equality in \eqref{eq:dfl_split}, 
 \begin{align*}
     &\nabla_{\mV}f(\mF)
     \\&{=}\cM_3\(\G\big((\mM-\mZ)\big)\times_1\mL^H\times_2\mR^H\)\cM_3(\bcS)^H{+}\mV(\breve{\mV}^H\breve{\mV})
     \\&=\cM_3(\overline{\bcW}\times_{3}(\mM-\mZ))\cM_3(\bcS)^H+\mV(\breve{\mV}^H\breve{\mV})\\&
     =(\mM-\mZ)\cM_3(\overline{\bcW})\cM_3(\bcS)^H+\mV(\breve{\mV}^H\breve{\mV})
 \\&
 =\mM\mB-\mV(\mB^H\mB)+\mV(\breve{\mV}^H\breve{\mV}), 
 \end{align*}
 where the second equality results from Lemma~\ref{lem:HankelT_algebra}.c, and in the last equality we recall that $\mB=\cM_3(\overline{\bcW})\cM_3(\bcS)^H$ in \eqref{eq:B} and $\mZ=\mV\mB^H$ in \eqref{eq:fast_HTproj_tt}.  

The computational complexity of $\nabla_{\mV} f(\mF)$  is 
$O(nr^2\log(n)+nr^3+sr^2+mr)$ flops in total. 
In $\nabla_{\mV}f(\mF)$, $\mM\mB$ costs $O(mr)$ flops as $\mM$ is $m$-sparse, $\mV(\mB^H\mB)$ costs $O((s+n)r^2)$ flops, $\mB$ costs $O(nr^2(r+\log(n)))$ flops, and $\mV(\breve{\mV}^H\breve{\mV})$ costs $O((s+n)r^2+r^4)$ flops from Lemma~\ref{lem:fast_scaleterm}.


 
\vspace{4.5mm}
\noindent\emph{3) Fast computation of  $\nabla_{\bcS}f(\mF)$}
\vspace{2mm}


 We rewrite $\nabla_{\bcS}f(\mF)$ as
 \begin{align*}
    &\nabla_{\bcS}f(\mF) =(\mL^H,\mR^H,\mV^H)\bcdot\(\G(\mM-\mZ)+\bcZ\)
    \\&{=}\G(\mV^H(\mM{-}\mZ)))\times_1\mL^H\times_2\mR^H{+}(\mL^H\mL,\mR^H\mR,\mV^H\mV)\bcdot\bcS
    \\&{=}\overline{\bcW}\times_3\hat{\mE}+(\mL^H\mL,\mR^H\mR,\mV^H\mV)\bcdot\bcS, 
 \end{align*}
where the second equality results from Lemma~\ref{lem:HankelT_algebra}.b and  $(\mL^H,\mR^H,\mV^H)\bcdot\bcZ=(\mL^H\mL,\mR^H\mR,\mV^H\mV)\bcdot\bcS$, the third equality results from  Lemma~\ref{lem:HankelT_algebra}.c and $\hat{\mE}=\mV^H(\mM-\mZ)$.

The computational complexity of $\nabla_{\bcS} f(\mF)$ is $O(nr^2\log(n)+nr^3+sr^2+mr)$ flops in total. 
In $\nabla_{\bcS} f(\mF)$,  $\overline{\bcW}\times_3\hat{\mE}$ costs $O(nr^3)$ flops, $(\mL^H\mL,\mR^H\mR,\mV^H\mV)\bcdot\bcS$ costs $O((s+n)r^2+r^4)$ flops and $\hat{\mE}$ in \eqref{eq:E_hat} costs $O((s+n)r^2+mr)$ flops. 

\section{Theoretical results} \label{sec:theoretical-results}
In this section, we present the theoretical results for our algorithm ScalHT. We first introduce the definitions required in our analysis and then present the recovery guarantee and the linear convergence result 
of ScalHT. 
\subsection{Definitions}

We first introduce the Hankel matrix basis and the Hankel tensor basis. Here $\ve(j)$ denotes the $j$-th element of vector $\ve$. 
\begin{definition}[Hankel matrix basis] \label{def:Hankelm}
 For $k\in[n]$, define the $k$-th orthogonal basis of Hankel matrix $\mH_k\in\R^{n_1\times n_2}$ as: 
    \begin{align*}
        \mH_k(i_1,i_2)=\frac{1}{\sqrt{w_k}}\ve_k(i_1+i_2),
    \end{align*}
     where $i_1\in[n_1]$, $i_2\in[n_2]$, $\ve_k$ is the $k$-th canonical orthogonal basis of $\R^{n}$ (  $n=n_1+n_2-1$) and {$w_k$ is the length of the $k$-th skew-diagonal of an $n_1\times n_2$ matrix  defined in notations.}
\end{definition}

\begin{definition}[Hankel tensor basis] \label{def:HankelT}
For $k\in[n], j\in[s]$, define the $(k,j)$-th orthogonal basis $\bcH_{k,j}
\in\C^{n_1\times n_2\times s }$ of Hankel tensors as 
\begin{align*}
   [\bcH_{k,j}](i_1,i_2,i_3) = \mH_{k}(i_1,i_2)\ve_{j}(i_3),
\end{align*}
where $i_1\in[n_1]$, $i_2\in[n_2]$, $i_3\in[s]$, $\mH_{k}\in\R^{n_1\times n_2}$ is the $k$-th Hankel matrix basis, and $\ve_{j}$ is the $j$-th canonical basis of $\R^s$. The spectral norms of different matricizations of $\bcH_{k,j}$ are:
\begin{align*}
   \|\cM_{1}(\bcH_{k,j})\|=\|\cM_{2}(\bcH_{k,j})\|=\frac{1}{\sqrt{w_k}},~
   \|\cM_{3}(\bcH_{k,j})\|=1. \numberthis \label{eq:HT_basis_bd}
\end{align*}
\end{definition}
We define the  condition number of $\bcZ_{\star}=\H(\mX_\star)$ as:  
\begin{definition}[Condition number]\label{def:kappa} The condition number of {$\bcZ_{\star}=\H(\mX_\star)$} is defined as
\begin{align} \label{eq:kappa}
\kappa \coloneqq \frac{\sigma_{\max}(\bcZ_{\star})}{\sigma_{\min}(\bcZ_{\star})},
\end{align}
where $\sigma_{\max}(\bcZ_{\star})$ and $\sigma_{\min}(\bcZ_{\star})$ are defined previously in notations.
\end{definition}

{If the conditions in Lemma~\ref{lem:mulrank_HT} hold, we have $\mulrank(\bcZ_\star)=\vr$ where $\bcZ_\star=\H(\mX_\star)=(\mE_L,\mE_R,\mB)\bcdot\bcD$. 
 From \cite{DeLathauwer2000,Tong2022}, when  $\mulrank(\bcZ_\star)=\vr$, 
 $\bcZ_\star$ admits the High Order Singular Vector Decomposition (HOSVD) that}  
 \begin{align}
     \bcZ_\star=(\mL_{\star},\mR_{\star},\mV_{\star})\bcdot\bcS_{\star}, \label{eq:grdtruth-oth-tucker}
 \end{align}
 where $\mL_\star\in\C^{n_1\times r}$, $\mR_\star\in\C^{n_2\times r}$, $\mV_\star\in\C^{s\times r}$ and $\bcS_\star\in\C^{r\times r \times r}$. 
 {The factors $(\mL_{\star},\mR_{\star},\mV_{\star})$ 
are column-orthonormal.} 
Besides, 
the core tensor $\bcS_\star$ satisfies 
\begin{align*}
\cM_{k}(\bcS_{\star})\cM_{k}(\bcS_{\star})^{T} = \bSigma_{\star,k}^2, \qquad k=1,2,3,
\end{align*}
where $\bSigma_{\star,k} \coloneqq \diag[\sigma_{1}(\cM_{k}(\bcZ_{\star})),\dots,\sigma_{r}(\cM_{k}(\bcZ_{\star}))]$. {Note that $\bcS_\star$ is not assured to exhibit the diagonal structure as $\bcD$ in \eqref{eq:Ground_truth_Tuckerdcp} as HOSVD is not unique.} {We apply this type of decomposition \eqref{eq:grdtruth-oth-tucker} for ease of convergence analysis.}  
 Also, we define a factor quadruple as $\mF_\star=(\mL_\star,\mR_\star,\mV_\star,\bcS_\star)$. 

We introduce the incoherence property of $\bcZ_\star$, which is pivotal in governing the well-posedness of low-rank Hankel tensor completion.
\begin{definition}[Incoherence]\label{def:incoh}
Let the Tucker decomposition of $\bcZ_\star=\H(\mX_\star)$ with multilinear rank $\vr=(r,r,r)$  be $\bcZ_\star=(\mL_{\star},\mR_{\star},\mV_{\star})\bcdot\bcS_{\star}$. {The $\mu_0$-incoherence property of $\bcZ_\star$ is defined as:} 
\begin{align*}
\ln\vL_\star\rn_{2,\infty} \leq &\sqrt{\frac{\mu_0 c_{\mathrm{s}} r}{n}},~\ln\vR_\star\rn_{2,\infty} \leq \sqrt{\frac{\mu_0 c_{\mathrm{s}} r}{n}},
\end{align*}
where $c_{\mathrm{s}}=\max\{n/n_1,n/n_2\}$ {can measure the symmetry between $n_1$ and $n_2$ as $n=n_1+n_2-1$ is fixed}.
\end{definition}{The performance of Hankel-lift approaches depends on the choice of $n_1$ and $n_2$ \cite{Cai2018,Cai2019,Chen2014,Cai2023} ($n=n_1+n_2-1$). 
   In our problem and \cite{Cai2018,Cai2019,Chen2014,Cai2023}, the sample complexity required for faithful recovery is an increasing function of $c_\mathrm{s}$. 
   Therefore, it is advisable to reduce $c_\mathrm{s}$ to $O(1)$ via $n_1 = O(n)$ and $n_2 = O(n)$ \cite{Chen2014,Cai2023}. }
\begin{remark}[Incoherence from 
frequency separation]
   Following the routes in \cite{Cai2018} and \cite[Thm. 2]{Liao2016}, it can be proven that $\bcZ_\star=\H(\mX_\star)$ is $\mu_0$-incoherent as long as the minimum wrap-around distance between the frequencies is greater than about $2/n$, and the damping factor $\tau_k=0$ for $k\in[r]$.
\end{remark}
\begin{remark}[No incoherence assumption of $\mV_\star$] \label{remark:incoh_V}
It doesn't impact our recovery guarantees whether the incoherence assumption of $\mV_\star$ exists or not.  Our guarantees depends on the dominant part 
only associated with  $\mL_\star$, $\mR_\star$, and the result of $\bcH_{k,j}$ that  $\|\cM_3(\bcH_{k,j})\|=1$ in Definition \ref{def:HankelT}.  The incoherence of $\mV_\star$ doesn't influence this dominant part. 

\end{remark}

\subsection{Theoretical guarantees}\label{subsec:thm}
In our analysis, we apply the sampling with replacement model as in \cite{Li2024a,Cai2018,Zhang2018}, which differs from the Bernoulli sampling model in \cite{Tong2022}. Novel concentration results under Hankel tensor sampling are presented in Lemma~\ref{lem:PT_HTbasis_bd}, \ref{lem:PTconc}, and \ref{lem:init_conc_m1}. The guarantees of our sequential spectral initialization in Alg.\ref{alg:init} are shown in Lemma~\ref{lem:init}. Now, we present the theoretical guarantees of ScalHT as follows. 
\begin{theorem}[Recovery guarantee] \label{thm:recovery}
     Suppose $\bcZ_{\star}$ is  incoherent in Definition \ref{def:incoh}, the step size  $0<\eta\leq0.4$ and the projection radius in \eqref{eq:incoh_proj} is   $B=C_{B}\sqrt{\mu_0 c_{\mathrm{s}}  r}\sigma_{\max}(\bcZ_\star)$ for $C_B\geq (1+\varepsilon_0)^3$ where  $\varepsilon_0>0$ is a small constant. With probability at least $1-O((sn)^{-2})$, the iterate in Algorithm \ref{alg:ScalHT} satisfies 
\begin{align*}
    \|\mX^k-\mX_\star\|_F\leq 3\varepsilon_0(1-0.5\eta)^k\sigma_{\min}(\bcZ_\star)
\end{align*}
 provided $m\gtrsim O(\varepsilon_0^{-2}\mu_0 c_{\mathrm{s}} sr^3\kappa^2\log(sn))$.  
\end{theorem}
\begin{proof}
{Recall 
$\mF=(\mL,\mR,\mV,\bcS)$ and $\mF_\star=(\mL_\star,\mR_\star,\mV_\star,\bcS_\star)$, we need to introduce the distance metric $\dist{\mF}{\mF_\star}$ in Appendix~A which measures the distance between two decompositions $\bcZ=(\mL,\mR,\mV)\bcdot\bcS$ and $\bcZ_\star=(\mL_\star,\mR_\star,\mV_\star)\bcdot\bcS_\star$.  }
  If the following inequality holds,
  \begin{align}
    \dist{\mF^k}{\mF_\star}\leq \varepsilon_0(1-0.5\eta)^k\sigma_{\min}(\bcZ_\star), \label{eq:dist_conveg_thm}
  \end{align}
 we can establish that 
   \begin{align*}
     \| \mX^k-\mX_\star\|_F&\leq  \|\D^{-1}\|\|\G^{*}\|\|(\mL^k,\mR^k,\mV^k)\bcdot\bcS^k-\bcZ_\star\|_F
     \\&\leq3\varepsilon_0(1-0.5\eta)^k\sigma_{\min}(\bcZ_\star), 
  \end{align*}
  where $\mX^k-\mX_\star=\D^{-1}\G^{*}((\mL^k,\mR^k,\mV^k)\bcdot\bcS^k-\bcZ_\star)$, $\|\D^{-1}\|\leq 1$ and $\|\G^{*}\|\leq 1$. $\|(\mL^k,\mR^k,\mV^k)\bcdot\bcS^k-\bcZ_\star\|_F\leq3\dist{\mF^k}{\mF_\star}$ is 
  in Lemma~\ref{lemma:perturb_bounds} of   supplementary material.  
  
  Next, we prove \eqref{eq:dist_conveg_thm} {via an inductive way}. For $k=0$, by Lemma~\ref{lem:init}, \eqref{eq:dist_conveg_thm} holds with high probability when  $\hat{m}\gtrsim O(\varepsilon_0^{-2}\mu_0 c_{\mathrm{s}} sr^3\kappa^2\log(sn))$, and $\mL^0,\mR^0$ satisfy the incoherence condition \eqref{eq:incoh_proj_F0}. Supposing \eqref{eq:dist_conveg_thm} and the incoherence for  $\mL^k,\mR^k$ in \eqref{eq:incoh_iter_full} hold for the $k$-th step, 
    we invoke  Lemma~\ref{lem:linconverge} to obtain  
  \begin{align*}
      \dist{{\mF'}^{k+1}}{\mF_\star}
      \leq\varepsilon_0(1-0.5\eta)^{k+1}\sigma_{\min}(\bcZ_\star),
  \end{align*}
  provide $\hat{m}\gtrsim O(\varepsilon_0^{-2}\mu_0 c_{\mathrm{s}} sr\kappa^2\log(sn))$ where ${\mF'}^{k+1}=({\mL'}^{k+1},{\mR'}^{k+1},\mV^{k+1},\bcS^{k+1})$ is shown in Algorithm~\ref{alg:ScalHT}. 
  As $(\mL^{k+1},\mR^{k+1})=\P_{B}({\mL'}^{k+1},{\mR'}^{k+1})$, we invoke Lemma~\ref{lem:proj} which shows the properties after projection to establish 
  \begin{align*}
      \dist{{\mF}^{k+1}}{\mF_\star}&\leq\dist{{\mF'}^{k+1}}{\mF_\star}
      \\ &\leq \varepsilon_0(1-0.5\eta)^{k+1}\sigma_{\min}(\bcZ_\star),
  \end{align*}
  and the incoherence condition \eqref{eq:incoh_iter_full} for $\mL^{k+1},\mR^{k+1}$.  
  Therefore, we prove \eqref{eq:dist_conveg_thm} via an induction way and conclude that $m=(k+1)\hat{m}\gtrsim O(\varepsilon_0^{-2}\mu_0 c_{\mathrm{s}} sr^3\kappa^2\log(sn))$. 
\end{proof}
{
\begin{remark}[Degrees of freedom] \label{rmk:freedom}
  The degree of freedom of the multi-measurement spectral compressed sensing problem is $O(sr)$. We aim to reconstruct $\mX_\star=\sum_{k = 0}^{r - 1}\vb_k\va(p_{k})^T\in\C^{s\times n}$, with $\va(p_k)=[1,e^{(\imath2\pi p_k)},\cdots,e^{\imath2\pi (n - 1)p_k}]^T$ and $\vb_k\in\mathbb{C}^{s}$.  The unknown variables are $\{\vb_k\}_{k=0}^{r-1}$ and $\{p_k\}_{k=0}^{r-1}$, and the number of free parameters in them is $(s+1)r$.
\end{remark}}
\begin{remark}[Sample complexity]
    Our sample complexity $O(sr^3)$ is optimal with respect to the 
    dimension $s$ {($\kappa$ and $\log(sn)$ are seen as constants)}. {In comparison to directly recovering $\mX_\star$ using the standard matrix completion approach \cite{Tong2021,Candes2009}—which has a sample complexity of $O(\max\{s,n\}r)$—our sample complexity can be smaller when $r\ll s<n$. 

   Besides, the sample complexity here refers to the number of observations of the matrix $\mX_\star$, rather than the lifted tensor  $\bcZ_\star$. As a result, direct comparison of this sample complexity with that in the low-rank tensor completion problem \cite{Chen2022,Tong2022} is not feasible. It is also inappropriate to apply the number of observations of the lifted tensor  $\H(\mX_\star)$ in this problem, as the statistical analysis depends on the Hankel tensor sampling basis $\bcH_{k,j}$ for $k\in[n],j\in[s]$ in Definition 2 which differs much from tensor sampling basis in low-rank tensor completion  \cite{Cai2022b,Tong2022}.  
   } 
\end{remark}
\begin{remark}[Iteration complexity]
    To achieve the $\varepsilon$ recovery accuracy that $\|\mX^k-\mX_{\star}\|_F\leq\varepsilon \sigma_{\min}(\bcZ_\star)$, the iteration complexity is $O(\log(1/\varepsilon))$ for our algorithm ScalHT. 
\end{remark}
{
Next, we provide the recovery guarantee in a noisy environment. Denoting $\mE\in\C^{s\times n}$ as the noise matrix, note that $\mX_\star$  in \eqref{eq:low_rank_H} should be replaced with $\mX_\star+\mE$, $\mY_\star$ in \eqref{eq:penalized} should be replaced with $\mY_\star+\D(\mE)$, and $\bcZ^0$ in Algorithm~\ref{alg:ScalHT} should be replaced with $\bcZ^0_{e}=\hat{p}^{-1}\P_{\Omega_0}(\mY_\star+\D(\mE))$. 
\begin{corollary}[Recovery guarantee with noise] \label{cor:recovery_noise}
      Suppose the conditions in Theorem 1 hold, the noise matrix $\mE\in\C^{s\times n}$ has independent sub-Gaussian entries with parameter $\sigma$ \cite{Vershynin2018},  and ${\sigma}\leq c_0\frac{\sigma_{\max}(\bcZ_\star)}{\kappa\sqrt{n^2\max\{s,n\}}}$ for some sufficiently small constant $c_0$.  
      With probability at least $1-O((sn)^{-2})$, the iterate in Algorithm 1 satisfies 
      \begin{align*}
    \|\mX^k-\mX_\star\|_F&\leq\|(\mL^k,\mR^k,\mV^k)\bcdot\bcS^k-\bcZ_\star\|_F \\
    &\leq\varepsilon_0(1{-}0.3\eta)^k\sigma_{\min}(\bcZ_\star){+}C_0\sigma\sqrt{{n^2\max\{s,n\}}}, \numberthis \label{eq:noise_recovery_bd}
\end{align*}
provided ${m}\geq O(\varepsilon_0^{-2}\mu_0 c_{\mathrm{s}} sr^3\kappa^2\log^2(sn))$, where $C_0$ is some constant. 
\end{corollary}
\begin{proof}
    See Appendix~\ref{apd:pf_noise} of the Supplementary Material. 
\end{proof}
\begin{remark}
    The dependence on noise part $\sigma\sqrt{n^2\max\{s,n\}}$ 
   is comparable to $\sigma\sqrt{n^2s}$ in AM-FIHT \cite{Zhang2018}, which arises from the fact that the number of noise elements in the lifted Hankel tensor domain is $O(n^2 s)$. The dependence on $\max\{s,n\}$ is an artifact of our proof technique. 
\end{remark}}

\begin{figure*}[!t]
\centering
	\subfloat[ ]{\includegraphics[width = 0.325\textwidth]{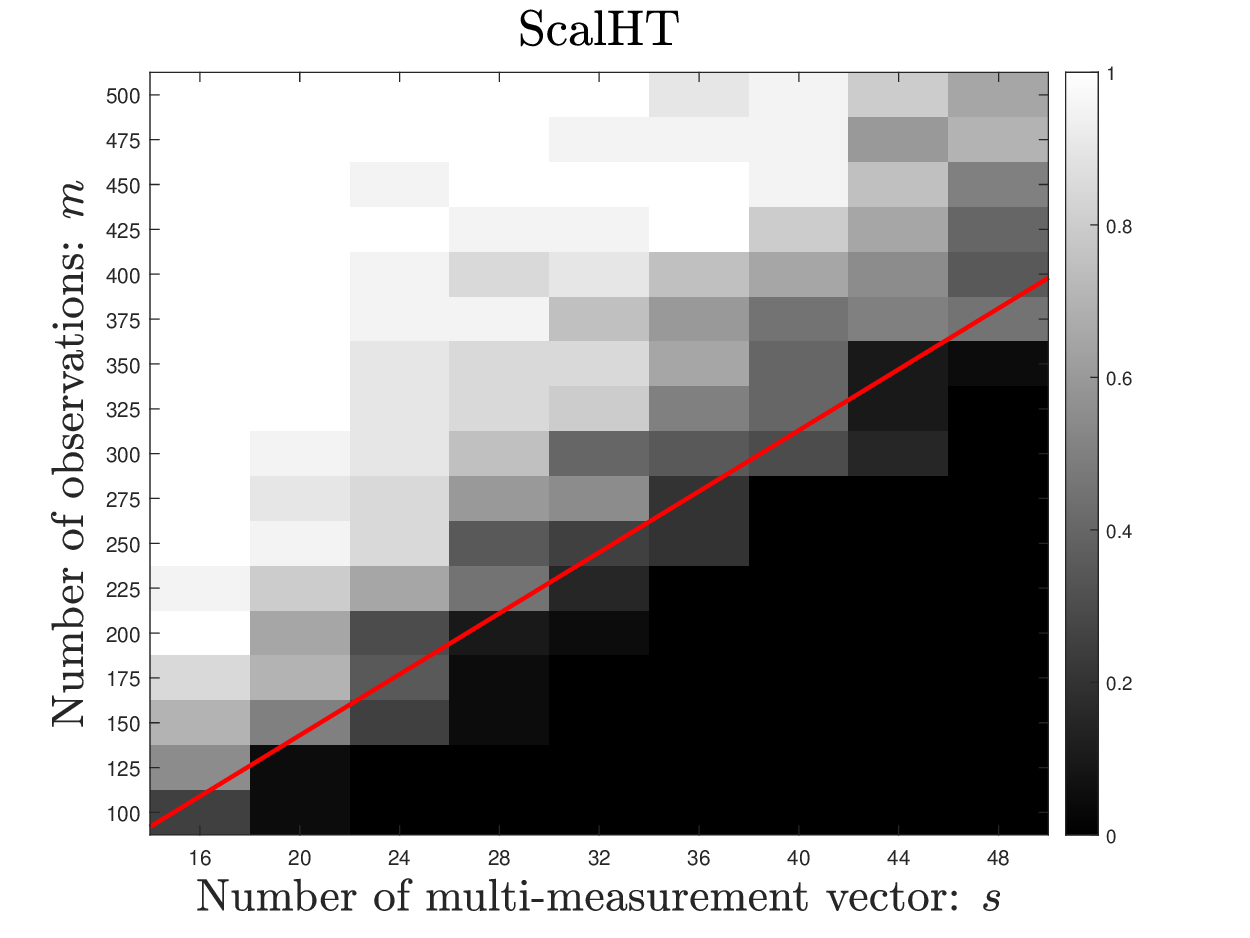}}
	\hfill
	\subfloat[]{\includegraphics[width = 0.325\textwidth]{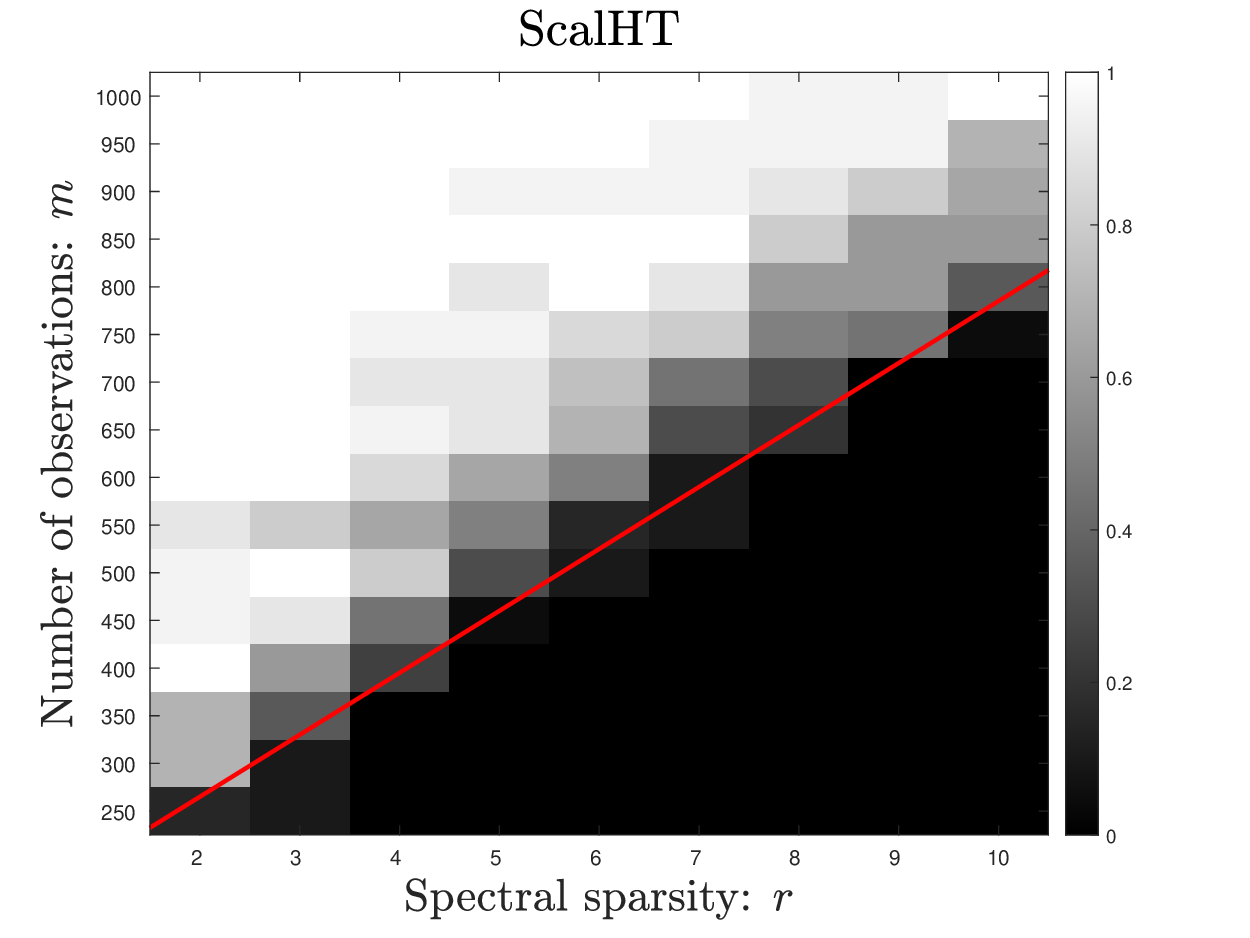}}
	\hfill
	\subfloat[]{\includegraphics[width = 0.325\textwidth]{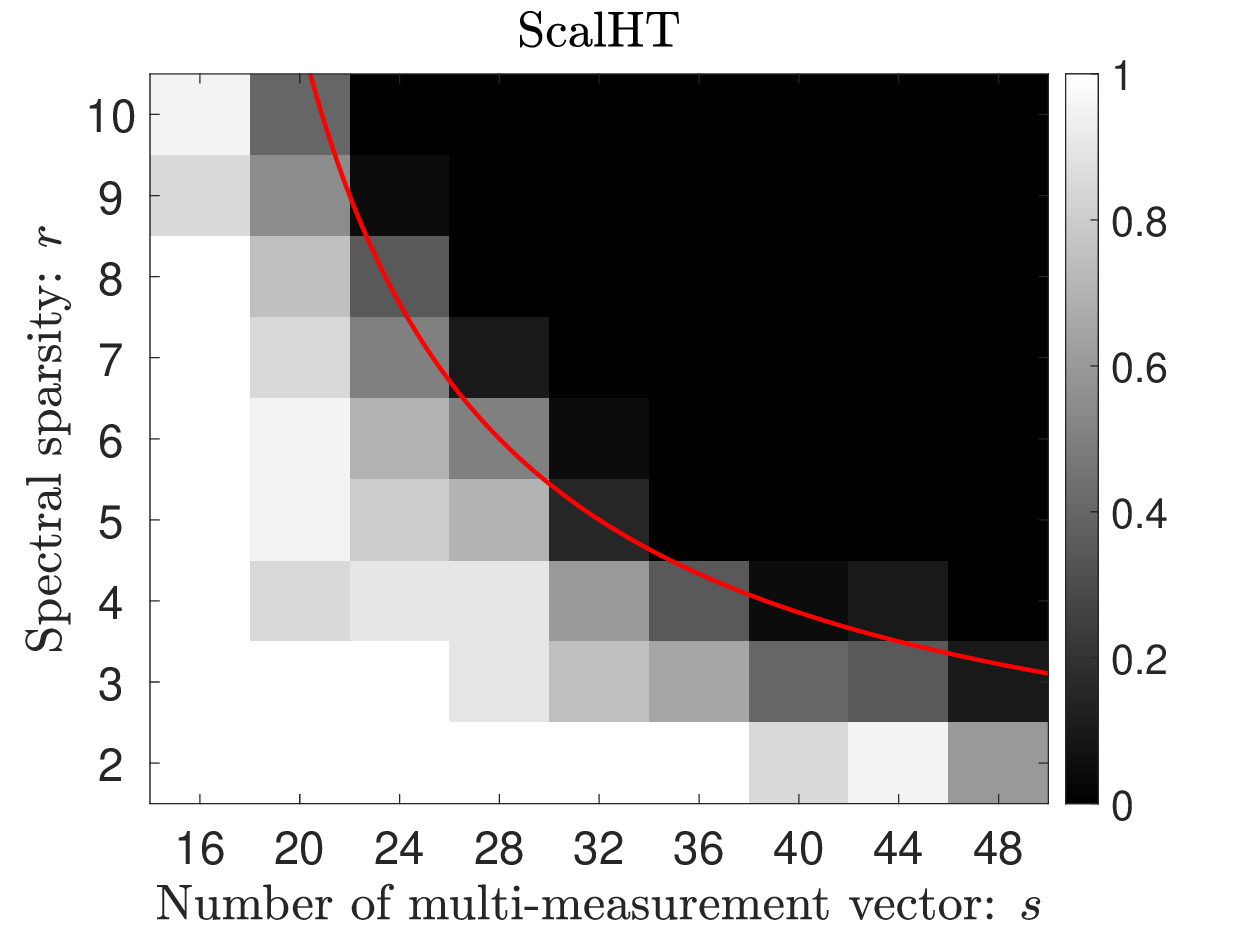}} 
\caption{{The phase transition performance of ScalHT. (a) Performance for varying $m$ and $s$ when $r=2$, and the red line plots $m=8s$. (b) Performance for varying $m$ and $r$ when $s=32$, and the red line plots $m=80r$. (c) Performance for varying $r$ and $s$ when $m=500$, and the red curve plots $sr=160$.}} \label{fig:sample_complexity}
\end{figure*}

{
\begin{remark} \label{rmk:noise_level_mag}
    When $\kappa=O(1)$, $r=O(1)$ 
   and the entries in $\bcZ_\star$ share the same order of magnitude, 
   the noise conditions reformulated  as  $\sigma\leq O(\frac{\sigma_{\max}(\bcZ_\star)}{\sqrt{n^2\max\{s,n\}}})= O(\sqrt{\frac{s}{\max\{s,n\}}}\|\mX_\star\|_{\infty})$, which implies that the noise level $\sigma$ can be as large as 
   a constant fraction of $\|\mX_\star\|_{\infty}$. There we use the facts $\sigma_{\max}(\bcZ_\star)= O( \|\bcZ_\star\|_F)=O(\sqrt{n^2s}\|\bcZ_\star\|_{\infty})$ and $\|\bcZ_\star\|_{\infty}=\|\H(\mX_\star)\|_{\infty}=\|\mX_\star\|_{\infty}$.

   Besides, as $\sigma_{\min}(\bcZ_\star)=O(\sigma_{\max}(\bcZ_\star))=O(\sqrt{n^2s}\|\bcZ_\star\|_{\infty})$, we conclude that $\sigma_{\min}(\bcZ_\star)$   in \eqref{eq:noise_recovery_bd} 
  exhibits an order-of-magnitude consistency with the noise part  
  $\sqrt{n^2\max\{s,n\}}\cdot \sigma$ when $s=n$. 
\end{remark}
}

\section{Numerical Simulations} \label{sec:numerical}

In this section, we conduct extensive simulations to showcase the performance of ScalHT \footnote{Our code is available at \url{https://github.com/Jinshengg/ScalHT}.}.  As in \cite{Li2024a,Cai2019,Zhang2019,Cherapanamjeri2017}, we employ the entire observation set rather than disjoint subsets in our simulations. The simulations are performed using MATLAB R2019b on a 64-bit Windows system equipped with a multi-core Intel i9-10850K CPU running at 3.60 GHz and 16GB of RAM. We begin by showing the recovery performance of ScalHT and compare it with ScaledGD, ANM, and AM-FIHT. Next, we examine the convergence performance in terms of recovery error for ScalHT, AM-FIHT, and ScaledGD, as shown in \ref{sim:converge}. 
Following this, we report the average 
runtime needed for our algorithm to reach a fixed accuracy in different problem settings, as detailed in \ref{sim:runtime_sn}. {Besides, we compare the performance of ScalHT with the existing Hankel tensor completion methods: STH-LRTC and the Fast Tucker method.}  Finally, we apply ScalHT to direction-of-arrival (DOA) estimation using a sparse linear array, as described in \ref{sim:doa}.

\subsection{Recovery performance} \label{sim:pst}
 In this subsection, we show {the recovery performance of  ScalHT to validate the sample complexity $m\gtrsim O(sr^3)$ established in Theorem~\ref{thm:recovery}. Also, we compare the recovery performance between ScalHT and atomic norm minimization (ANM) \cite{Li2016,Yang2016a}, AM-FIHT \cite{Zhang2018}, and ScaledGD \cite{Tong2022}.}
 
 The ground truth matrix $\mX_\star$ is constructed as $\mX_\star = \sum_{k=0}^{r-1} \vb_k \va(p_k)^T$ where $p_k=e^{(\imath2\pi f_k)}$. Here, the coefficient vectors $\{\vb_k\}_{k=0}^{r-1}$ are drawn from a standard Gaussian distribution and then normalized. 
 The frequencies $\{f_k\}_{k=0}^{r-1}$ are randomly selected from $[0,1)$ without any separation constraints. The step size for both ScalHT is set to $\eta = 0.25$. The termination condition for ScalHT is met when $\|\bcS^{k+1} - \bcS^k\|_F / \|\bcS^k\|_F \leq 10^{-7}$ or when the maximum number of iterations is reached.  A test is considered successful if $\|\mX^k - \mX_\star\|_F / \|\mX_\star\|_F \leq 10^{-3}$. We run 30 random trials for each parameter configuration and record the success rate.

{In the first experiment, we test the phase transition performance of ScalHT for three cases: varying the number of observations $m$ and the number of multi-measurement vectors $s$ with the spectral sparsity $r$ fixed, varying $m$ and $r$ with $s$ fixed, and varying $r$ and $s$ with $m$ fixed. We first set $m$ to range from $100$ to $500$ in increments of $25$, $s$ to range from $16$ to $48$ in increments of $4$, with the spectral sparsity $r=2$ fixed, and test the phase transition performance of ScalHT. Next, we set $m$ to range from $250$ to $1000$ in increments of $50$, $r$ to range from $2$ to $10$ in increments of $1$, and $s=32$ as fixed. Last, we set  $s$ to range from $16$ to $48$ in increments of $4$, $r$ to range from $2$ to $10$ in increments of $1$, and $m=500$ as fixed.

We can observe that in Figure~\ref{fig:sample_complexity}.(a), ScalHT achieves successful recovery when $m\gtrsim 8s$, indicating that the required sample complexity exhibits a linear dependence on $s$.  Figure~\ref{fig:sample_complexity}.(b) tells us that ScalHT achieves successful recovery when $m\gtrsim 80r$, indicating that the required sample complexity exhibits a linear dependence on $r$. Figure~\ref{fig:sample_complexity}.(c) tells us that ScalHT achieves successful recovery when $m\gtrsim 3.125 sr$ \footnote{{The red curve is $sr=160$, thus when $sr\leq 160=\frac{m}{3.125}$, we can achieve successful recovery where $m=500$.} },  
indicating that the required sample complexity exhibits a linear dependence on $sr$ jointly. Therefore, our current sample complexity  $m \gtrsim O(sr^3)$ exhibits an optimal dependence on $s$, and a suboptimal dependence on $r$, which we aim to improve in future work.
}
 
 
 {In the second experiment, we compare the recovery performance between ScalHT and ANM \cite{Li2016,Yang2016a}, AM-FIHT \cite{Zhang2018}, and ScaledGD \cite{Tong2022}.}  {The stepsize and the termination condition of ScaledGD are the same as ScalHT.} For AM-FIHT \cite{Zhang2018}, we use the $\beta = 0$ version, such that the heavy ball acceleration step does not introduce additional effects. The termination condition for AM-FIHT is $\|\mX^{k+1} - \mX^k\|_F / \|\mX^k\|_F \leq 10^{-7}$ or when the maximum number of iterations is reached. ANM is implemented using CVX. We set the signal length as $n = 63$, the number of multiple measurement vectors as $s = 32$, and the spectral sparsity as $r = 8$, with observation ratios $p$ ranging from $0.05$ to $0.95$ in $19$ increments.  In Fig. \ref{fig:pst}, we observe that ScalHT is more stable than ANM. ANM performs better at lower observation ratios. Additionally, ScalHT shows similar performance to ScaledGD, and both of them outperform AM-FIHT.
\begin{figure}[!t]
		\centering
	  \includegraphics[width=0.85\linewidth]{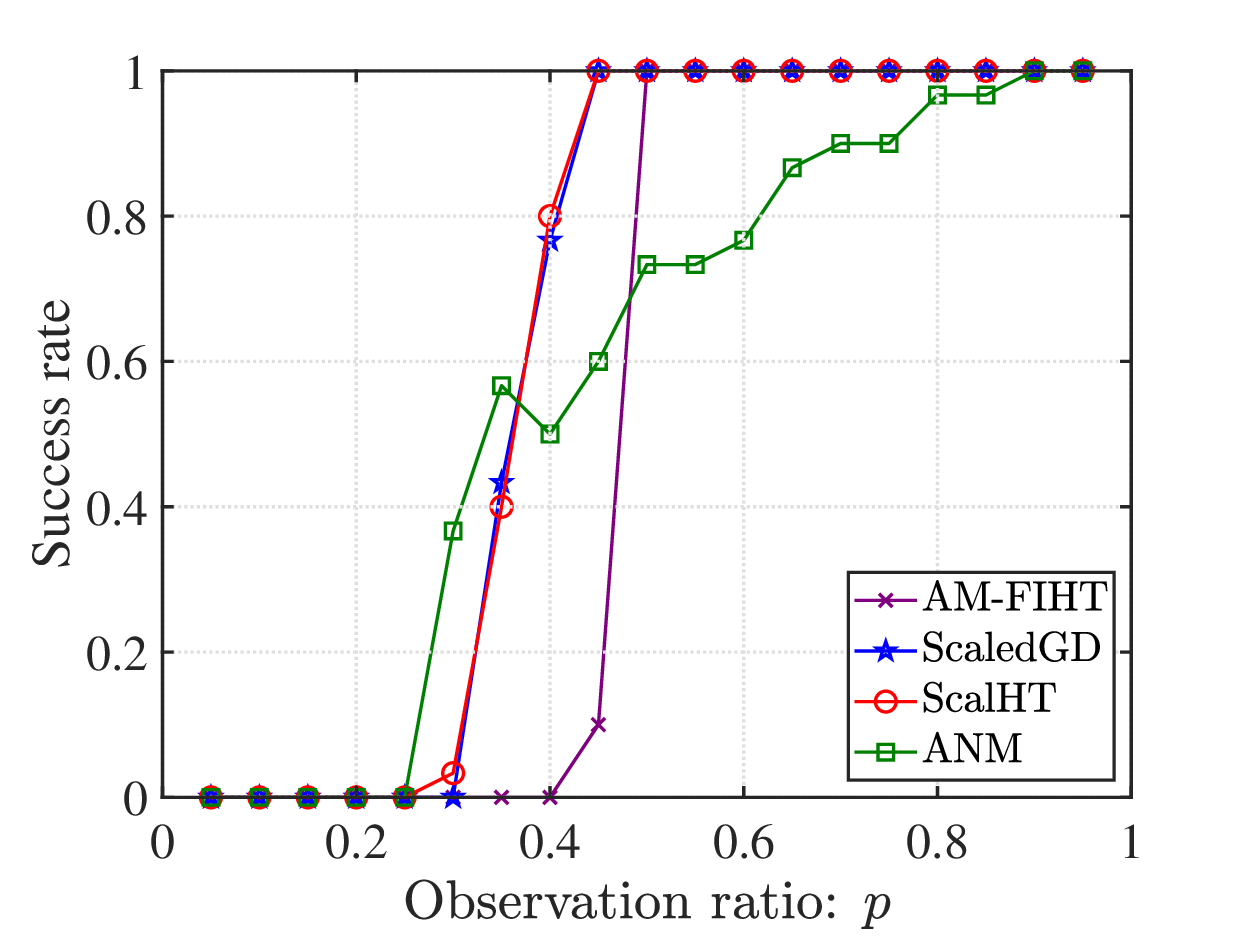}	
  \caption{The {success rate curve} of ScalHT, AM-FIHT, ScaledGD, and ANM when $n=63$, $s=32$ and $r=8$. 
  }
		\label{fig:pst}
\end{figure}

\subsection{Convergence performance} \label{sim:converge}

This subsection presents the convergence performance of ScalHT, AM-FIHT, and ScaledGD in terms of relative error, defined as $\|\mX^k - \mX_\star\|_F/\|\mX_\star\|_F$. We set the parameters as $n = 511$, $s = 512$, and $r = 6$, and use observation ratios $p$ of $0.17$ and $0.22$. The step size $\eta$ is set to $0.4$ for both ScalHT and ScaledGD. Other parameters are consistent with those in \ref{sim:pst}. Each case is run 20 times, and the number of iterations and runtime are recorded. We 
present the average iteration and runtime for achieving four relative errors ${10^{-2}, 10^{-4}, 10^{-6}, 10^{-8}}$ in Table \ref{tab:time_versus_relerr}.


From Table \ref{tab:time_versus_relerr}, 
we conclude that ScalHT, AM-FIHT, and ScaledGD all exhibit linear convergence. When the observation ratio is low at $p = 0.17$, ScalHT and ScaledGD converge faster than AM-FIHT in terms of iterations. However, when the observation ratio is higher at $p = 0.22$, AM-FIHT converges more quickly than both ScalHT and ScaledGD. Additionally, the convergence behavior of ScalHT is similar to that of ScaledGD. Next, we examine the average time required for these algorithms to achieve different relative errors. 
ScalHT converges the fastest in terms of runtime among the three algorithms. The computational efficiency of ScalHT is approximately 10 times higher than that of AM-FIHT and 100 times higher than that of ScaledGD, highlighting the superior efficiency of ScalHT.

\begin{table}[h!]	
		\begin{center}
			
			\caption{The average iterations and time (in seconds) versus different relative errors.} 
            \setlength{\tabcolsep}{1.6mm}{
			\begin{tabular}{c|c|c|c|c|c|c|c|c}
   \hline
				\textbf{Rel. err} & \multicolumn{2}{c|}{$10^{-2}$} &  \multicolumn{2}{c|}{$10^{-4}$}  &  \multicolumn{2}{c|}{$10^{-6}$} &  \multicolumn{2}{c}{$10^{-8}$}  \\
	\hline			
				 &Iter. & Time &Iter. & Time &Iter. & Time  &Iter. & Time 
    \\ 
    \hline
     & \multicolumn{8}{c}{$p=0.17$}
    \\ 	\hline 	
				
				ScalHT   &$12.1$ & $\bm{0.14}$  & ${33.8}$  & $\bm{0.4}$ & ${59.5}$  & $\bm{0.7}$ & ${88.7}$  & $\bm{1.0}$  \\
				AM-FIHT    &${11.3}$  & $1.7$  & $49.3$  & $7.2$  & $92.9$  & $13.5$ &$141.4$  & $20.4$  \\
				ScaledGD    &$12.5$  & $19.0$  & $34.2$  & $52.5$  & $59.8$  & $91.6$ &$88.8$  & $136.4$  \\

    \hline 
     & \multicolumn{8}{c}{$p=0.22$}
    \\ \hline

				ScalHT   &$10.3$ & $\bm{0.13}$  & $29.1$  & $\bm{0.37}$ & $50.9$  & $\bm{0.6}$ & $75.5$  & $\bm{0.95}$  \\
				AM-FIHT    &$6.5$  & $1.0$  & $21.2$  & $3.1$  & $39.3$  & $5.6$ &$60.3$  & $8.7$  \\
				ScaledGD    &$10.7$  & $16.3$  & $29.3$  & $44.7$  & $51.0$  & $78.0$ &$75.6$  & $115.7$  \\
    \hline
			\end{tabular}}
			\label{tab:time_versus_relerr}
		\end{center}		
	\end{table}

    \begin{table*}[h!]	
		\begin{center}
			
			\caption{The average iterations and time (in seconds) to achieve the relative error $10^{-3}$. }
			\begin{tabular}{c|c|c|c|c|c|c|c|c|c|c|c|c}
   \hline
				{${n}$} & \multicolumn{6}{c|}{$511$} &  \multicolumn{6}{c}{$1023$}  
                \\
	\hline			
    {${s}$} & \multicolumn{2}{c|}{$768$} &  \multicolumn{2}{c|}{$1024$}  &  \multicolumn{2}{c|}{$1280$} &  \multicolumn{2}{c|}{$768$} &  \multicolumn{2}{c|}{$1024$}  &  \multicolumn{2}{c}{$1280$}  \\
	\hline			
				 &Iter. & Time &Iter. & Time &Iter. & Time  &Iter. & Time &Iter. & Time  &Iter. & Time 

    \\ 	\hline 	
				
				ScalHT   &${18.6}$ & $\bm{0.38}$  & $19.1$  & $\bm{0.43}$ & $19.4$  & $\bm{0.50}$ & $19.1$  & $\bm{0.39}$ & $19.3$  & $\bm{0.46}$ & $19.2$  & $\bm{0.53}$  \\
				AM-FIHT    &$20.2$  & $4.4$  & $14.2$  & $4.2$  & $14.4$  & $5.5$ &$26.2$  & $27$  & $19.4$  & $12$ &$18.1$  & $14$ \\
				ScaledGD    &$19.2$  & $40$  & $20.5$  & $62$  & $20.1$  & $72$ &$20.2$  & $4.3$e$2$  & $20.8$  & $8.6$e$2$ &$21.1$  & $1.4$e$3$ 
    \\ \hline 	

			\end{tabular}
			\label{tab:timeiter_fixedacc}
		\end{center}		
	\end{table*}
    
\subsection{Runtime comparisons under different problem scales} \label{sim:runtime_sn}

In this subsection, we evaluate the runtime of our algorithm under different problem scales.  The spectral sparsity is set as $r=6$.  
The step size $\eta$ is set to $0.4$ for both ScalHT and ScaledGD.  For each problem scale, we run 10 random tests and record the average number of iterations and runtime required to achieve a fixed error of $\|\mX^k - \mX_\star\| / \|\mX_\star\|_F = 10^{-3}$. In our first simulation, we evaluate the runtime performance of our algorithm across varying problem dimensions $n$, keeping $s=256$ constant. Specifically, we consider $n=2^j-1$ for $j\in\{8,9,\cdots,12\}$, ranging from a minimum of $n=255$ to a maximum of $n=4095$. To maintain well-conditioned tasks, the number of observations is determined as $m=\lfloor 2.1sr \log(n) \rfloor$. 
 Figure \ref{fig:runtime_vs_n} reveals that our algorithm ScalHT consistently demonstrates significantly reduced computational expenses in comparison to AM-FIHT and ScaledGD across different problem dimensions $n$. Moreover, this enhancement in computational efficiency becomes more pronounced as the value of $n$ increases.

Second, we test our algorithm's runtime performance under various high dimensional $(s, n)$ scenarios. 
We set $n = 511$ and $n = 1023$. For each value of $n$, we set $s = 768, 1024$, and $1280$. 
The number of observations is set as $m = \lfloor 2.7sr \log(n) \rfloor$ in these cases. 
\begin{figure}[!t]
		\centering
	  \includegraphics[width=0.85\linewidth]{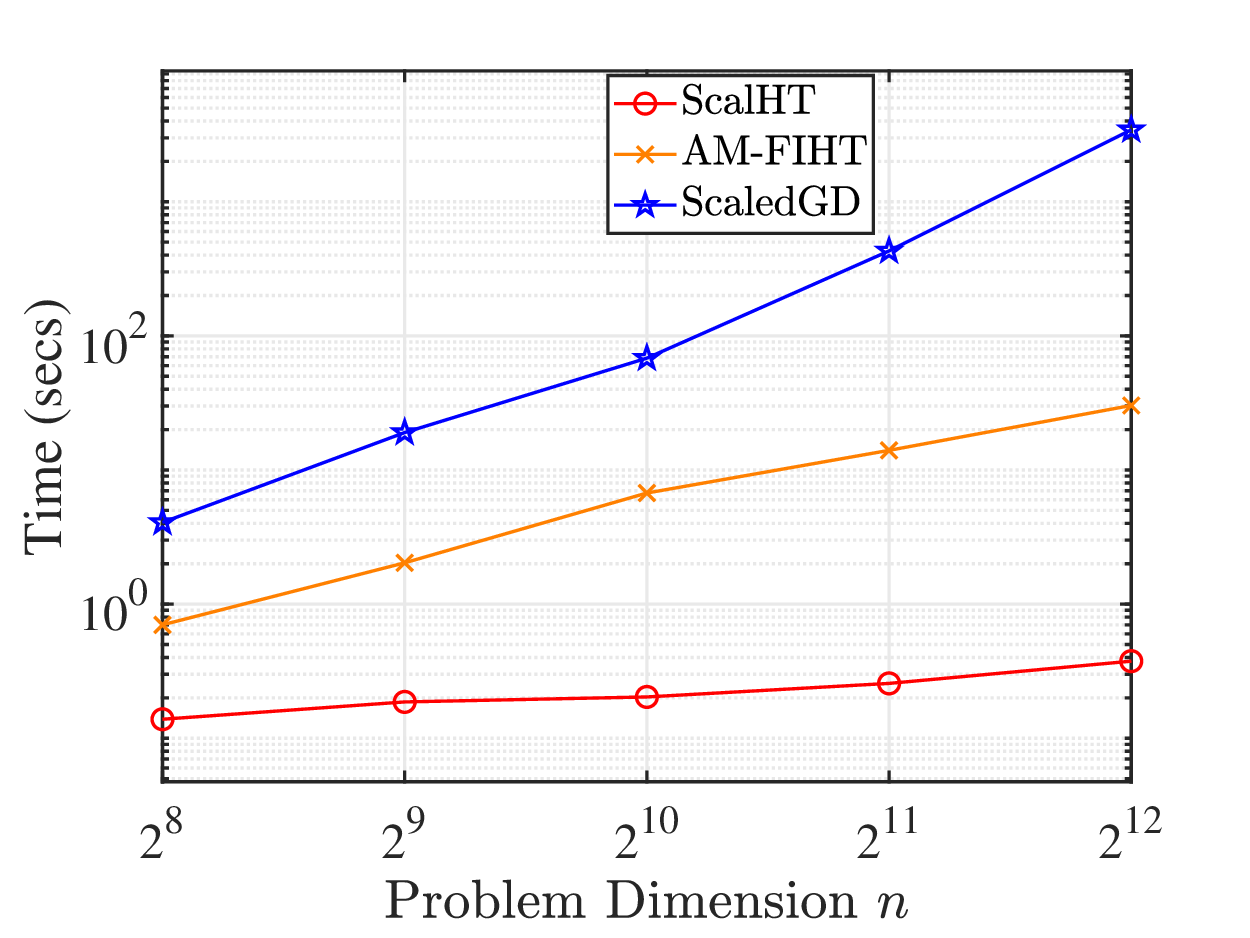}	
  \caption{The average run time comparisons versus different problem dimensions $n$ to achieve relative error $10^{-3}$.  
  }
		\label{fig:runtime_vs_n}
\end{figure}
From the results in Table \ref{tab:timeiter_fixedacc}, it is apparent that ScalHT exhibits markedly superior computational efficiency relative to both AM-FIHT and ScaledGD within high-dimensional contexts. Notably, for scenarios where $s = 1024$ and $n = 1023$, ScalHT showcases a computational efficiency surpassing that of AM-FIHT at least by a factor of 20. Furthermore, in instances where $s = 1280$ and $n = 1023$, ScalHT's computational efficiency exceeds that of ScaledGD by a minimum of 2500 times.

{
\subsection{Comparisons with Existing Hankel Tensor Methods}
In this subsection, we compare the performance of ScalHT with the existing Hankel tensor completion methods in the synthetic data and a real dataset, traffic 40 in \cite{Shi2020}. The related baseline methods are spatiotemporal Hankel low rank tensor completion (STH-LRTC) in  \cite{Wang2023a} and the fast Tucker method in delay-embedding (Hankelization) space in \cite{Yamamoto2022}. Note that we only Hankelize the second dimension (the dimension of the spectral sparse signal) of the input data matrix $\mX_\star$. Consequently, the window size is set as $\vtau=(1,\lfloor n/2\rfloor)$ in the fast Tucker method \cite{Yamamoto2022} and STH-LRTC \cite{Wang2023a}.

In the first experiment, we compare the recovery performance of ScalHT, STH-LRTC, and the fast Tucker method in the synthetic data. We choose the real-valued signal $\mX_\star=\sum_{k=0}^{r-1}\vb_k\va(p_{k})^T$ where $\vb_k\in\R^{n}$, $p_k\in\R$ and $p_k\in[0,1/20]$. 
This choice arises from fast Tucker method is designed for real-valued data
. We set the number of multi-measurement vectors $s=768$ as fixed, and the length of the signal $n=511, 767, 1023$ respectively. The spectral sparsity is $r=6$, which is assumed to be known in ScalHT.  
The number of observations is set as $m=\lfloor 6.5  sr\log(n)\rfloor$. The number of Monte Carlo trials is $20$.  We record the average runtime and relative error $\|\mX^k-\mX_\star\|_F/\|\mX_\star\|$. Other settings of ScalHT are the same as in subsection~\ref{sim:pst}.

As shown in Table~\ref{tab:simdata}, we conclude that ScalHT exhibits a higher computational efficiency compared to the existing low-rank Hankel tensor methods, especially when $n$ is large.  ScalHT achieves lower recovery error compared to the Fast Tucker method and STH-LRTC when $n=511$ and $n=767$.
}
\begin{table}[h!]	
		\begin{center}
			{
			\caption{{The average relative error and time (in seconds) for reconstruction.}} 
            \setlength{\tabcolsep}{1.6mm}{
			\begin{tabular}{c|c|c|c|c|c|c}
   \hline
				\textbf{$\bm{n}$} & \multicolumn{2}{c|}{$511$} & \multicolumn{2}{c|}{$767$} &  \multicolumn{2}{c}{$1023$} 
                \\
	\hline			
				 &Err. & Time &Err. & Time  &Err. & Time 
    \\ 	\hline 	
				
				ScalHT   &\textbf{5.3e-3}
 & \textbf{1.7} & \textbf{2.5e-3
}  & \textbf{2.2} & {3.5e-3
}  & \textbf{2.3} 
                \\
				\multirow{1}{*}
                {Fast Tucker} &\multirow{1}{*}{{2.4e-1}}  & \multirow{1}{*}{20}  & \multirow{1}{*}{2.2e-1}  & \multirow{1}{*}{38}   & \multirow{1}{*}{3.2e-1}  & \multirow{1}{*}{89}  
               \\
				\multirow{1}{*}
                {STH-LRTC} &\multirow{1}{*}{{2.7e-2}}  & \multirow{1}{*}{4.1e2}  & \multirow{1}{*}{4.2e-3}  & \multirow{1}{*}{5.5e2}   & \multirow{1}{*}{\textbf{3.1e-3}}  & \multirow{1}{*}{7.0e2}  
                \\
    \hline 
			\end{tabular}}
			\label{tab:simdata}
            }
		\end{center}		
	\end{table}
    
{
In the second experiment, we compare the performance of these algorithms in the traffic 40 dataset \cite{Shi2020}, which is a $228\times 40$ matrix. We study the recovery performance under two observation ratios $p=0.4$ and $p=0.6$. For the traffic 40 dataset, we know $s=228$ and $n=40$. We prescribe $r=5$ 
in our algorithm, ScalHT. The termination condition for ScalHT is $\|\bcS^{k+1} - \bcS^k\|_F / \|\bcS^k\|_F \leq 10^{-3}$ or maximum number of iterations is reached. 
We compare ScalHT with the fast Tucker method \cite{Yamamoto2022} and STH-LRTC \cite{Wang2023a}. The window size is set as $\vtau=(1,\lfloor n/2\rfloor)$ as before. 
We record the average relative error $\|\mX^k-\mX_\star\|_F/\|\mX_\star\|$ and runtime. 

From Table~\ref{tab:realdata}, we conclude that ScalHT exhibits a lower recovery error and higher computational efficiency compared to the fast Tucker method and STH-LRTC under the sample ratio $p=0.4$ and $p=0.6$. This verifies the competitiveness of  ScalHT in real data.}
\begin{table}[h!]	
		\begin{center}
			{
			\caption{{The average relative error and time (in seconds) to reconstruct the traffic 40 dataset.}} 
            \setlength{\tabcolsep}{2mm}{
			\begin{tabular}{c|c|c|c|c}
   \hline
				\textbf{Observation ratio} & \multicolumn{2}{c|}{$p=0.4$}  &  \multicolumn{2}{c}{$p=0.6$} 
                \\
	\hline			
				 &Err. & Time &Err. & Time  
    \\ 	\hline 	
				
				ScalHT   
 & \textbf{1.1e-2
}  & \textbf{7.1e-1} & \textbf{1.3e-3}  & \textbf{3.3e-1} 
                \\
                {Fast Tucker } 
                & \multirow{1}{*}{2.0e-2}  & \multirow{1}{*}{8.2e-1}  & \multirow{1}{*}{1.6e-2} & \multirow{1}{*}{5.7e-1} 
                  \\
                {STH-LRTC } 
                & \multirow{1}{*}{1.9e-1}  & \multirow{1}{*}{2.7}  & \multirow{1}{*}{7.9e-2} & \multirow{1}{*}{1.9} 
               \\
    \hline 
			\end{tabular}}
			\label{tab:realdata}}
		\end{center}		
	\end{table}
    \vspace{-8mm}
\subsection{Application to direction-of-arrival (DOA) estimation} \label{sim:doa}
\begin{figure}[!t]
		\centering
	  \includegraphics[width=0.85\linewidth]{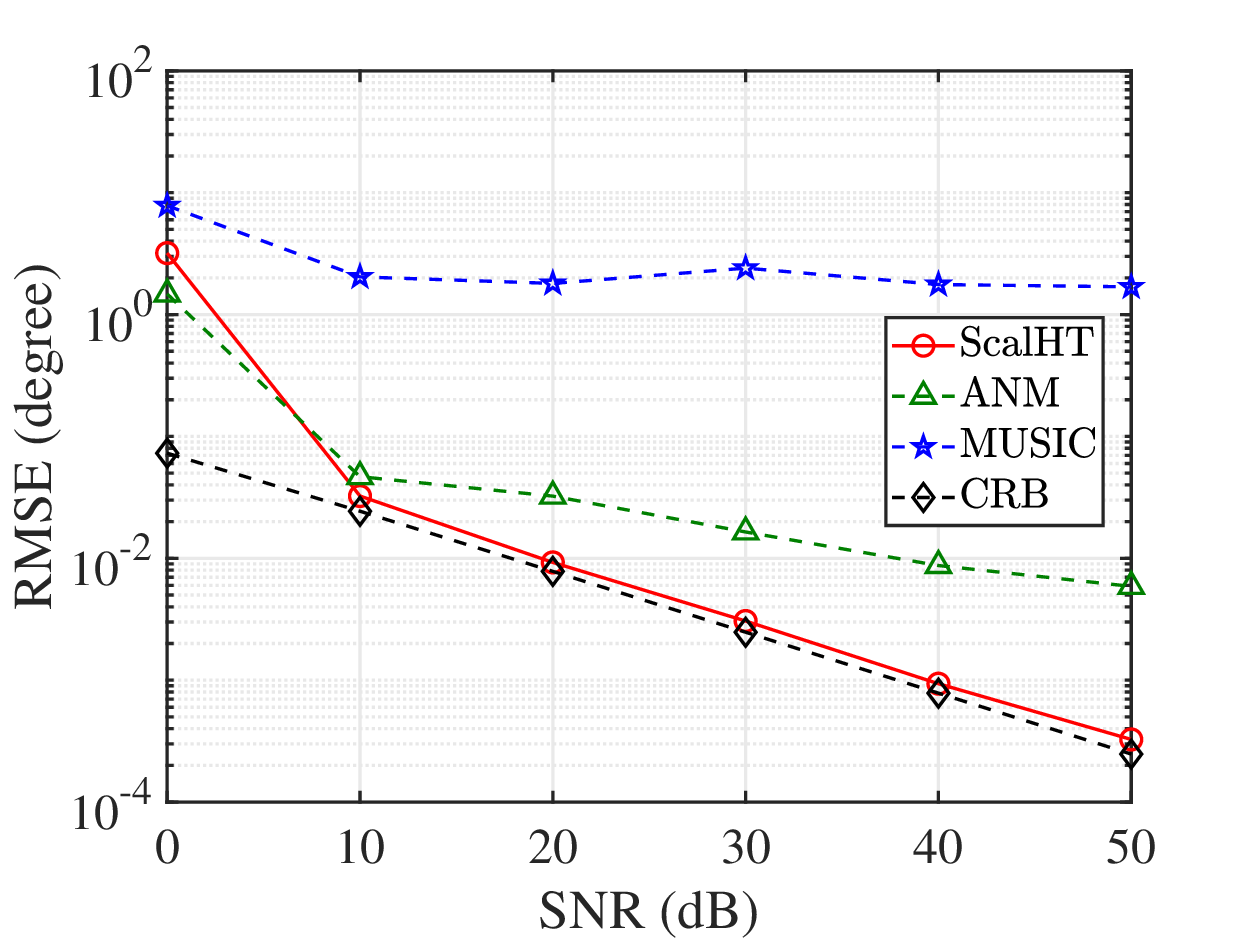}	
  \caption{RMSE in degree versus SNR for 
  ScalHT from $32$ snapshots.  
  }
		\label{fig:RMSE}
\end{figure}
Last, we consider the direction-of-arrival (DOA) problem for a sparse linear array (SLA) under electromagnetic attacks in the time domain. Without loss of generality, we suppose the minimum intersensor spacing is taken as half the wavelength. The multiple signals received from the virtual full
array  with $n$ elements and $s$ snapshots are given as 
\begin{align*}
    \mX=\sum_{k=0}^{r-1} \vb_k \va(\theta_k)^T+\mE=\mX_\star+\mE, 
\end{align*}
where $r$ is the number of far-field narrowband sources, $\vb_k\in\C^{s\times 1}$ denotes the $s$ snapshots of the $k$-the source signal in time domain, $\va(\theta_k)=[1,p_k,\cdots,p_k^{n-1}]^T$ where $p_k=e^{\imath\pi\sin(\theta_k)}$ for $\theta_k\in[-90^\circ,90^\circ)$,  $\mE\in\C^{s\times n}$ denotes the noise matrix, and $\mX_\star=\sum_{k=0}^{r-1} \vb_k \va(\theta_k)^T$. Denote the sparse sensor index set as $\Omega_{\mathrm{SLA}}\in\{0,1,\cdots,n-1\}$. The multiple signals received from a sparse linear array with random electromagnetic attacks in the time domain are given as 
\begin{align*}
    \mX_{\Omega}=\P_{\Omega}(\mX_\star)+\P_{\Omega}(\mE),
\end{align*}
where $\Omega\subseteq\{0,1,\cdots,s-1\}\times\Omega_{\mathrm{SLA}}$. Specifically, we consider the scenario where several snapshots from one or more sensors are missing due to electromagnetic attacks. In our simulation, we set 10 percent of locations of  $\{0,1,\cdots,s-1\}\times\Omega_{\mathrm{SLA}}$ that are missing as a result of random electromagnetic attacks. The corresponding observation ratio is given by $p=0.9\cdot|\Omega_{\mathrm{SLA}}|/n$ relative to the full data matrix $\mX$. The number of virtual full array's elements is $n=64$. The sparse linear array consists of $17$ sensors, whose index set is $\Omega_{\mathrm{SLA}}=\{2,4,6,10,20,21,23,30,33,38,39,40,49,50,51,56,62\}$.  We set four sources whose DOAs are $\{1^\circ, 2^\circ, 4^\circ, 6^\circ\}$. 
The number of snapshots is $s=32$.  The signal-to-noise ratio (SNR) refers to the ratio of the power of the source signal to the power of noise  ranging from $0$ dB to $50$ dB. The root mean square error (RMSE) is defined as $\sqrt{\frac{1}{P}\sum_{i=1}^{P}\|\hat{\bm{\theta}}_i-\bm{\theta}\|_2^2}$ where $\bm{\theta}=[\theta_0,\cdots,\theta_{r-1}]^T$, $\hat{\bm{\theta}}_i$ is the estimation result in the $i$-th trial, and $P$ is the number of trials. 

We first apply MUSIC \cite{Schmidt1986} directly to the covariance matrix $\frac{1}{s}{\mX_{\Omega}^T}({\mX_\Omega^T})^H$ to estimate the DOAs. 
Next, we use ScalHT with $\eta=0.25$ and ANM \cite{Li2016,Yang2016a} to reconstruct $\mX_{\Omega}$, yielding  the estimated full data matrix $\hat{\mX}$. MUSIC is then applied to $\hat{\mX}$ to obtain the DOAs. For each SNR, we run 100 random trials and record the RMSE. Besides, the CRB \cite{Wang2017} for the sparse array 
is plotted. From Fig. \ref{fig:RMSE}, we observe that our algorithm ScalHT achieves a lower RMSE compared to both ANM and applying MUSIC directly when  the SNR exceeds 10 dB. Furthermore, the RMSE of ScalHT approaches the CRB for the sparse array, demonstrating its superior estimation performance.


\section{Conclusions} 
In this study, we propose a novel low-rank Hankel tensor completion approach to solve multi-measurement spectral compressed sensing. 
Building on low-rank Tucker decomposition and the Hankel tensor structure, we introduce a fast, non-convex scaled gradient descent method named ScalHT for solving the Hankel tensor completion problem. 
We present innovative fast computation formulations for ScalHT, achieving $O(\min\{s,n\})$-fold improvement in storage and computation efficiency over the previous algorithms. 
Furthermore, we provide recovery and linear convergence guarantees for ScalHT.  
Numerical experiments demonstrate that our method outperforms existing algorithms, exhibiting substantially lower computational and storage costs while achieving superior recovery performance.
\appendices
\section{Additional preliminaries}

\textbf{Useful facts from tensor algebra:}
\begin{align} 
\cM_1(\bcX)\breve{\mL}
=\cM_1(\bcX\times_2 \mR^H&\times_{3}\mV^H)\cM_1(\bcS)^H,\label{eq:tensor_properties_l}
\\\cM_2(\bcX)\breve{\mR}
=\cM_2(\bcX\times_1 \mL^H&\times_{3}\mV^H)\cM_2(\bcS)^H,\label{eq:tensor_properties_r}
\\\cM_3(\bcX){\breve{\mV}}
=\cM_3(\bcX\times_1 \mL^H&\times_{2}\mR^H)\cM_3(\bcS)^H.\label{eq:tensor_properties_v}
\end{align}
For $\bQ_k \in \C^{r\times r}$, $k=1,2,3$:
\begin{align}
(\mL,\mR,\mV)\bcdot \big((\bQ_{1},\bQ_{2},\bQ_{3})\bcdot\bcS \big) &=(\mL\bQ_{1},\mR\bQ_{2},\mV\bQ_{3})\bcdot\bcS, \label{eq:tensor_properties_1} \\
\left\|(\bQ_{1},\bQ_{2},\bQ_{3})\bcdot\bcS\right\|_{F} &\le \|\bQ_{1}\|\|\bQ_{2}\|\|\bQ_{3}\|\|\bcS\|_{F} \label{eq:tensor_properties_3}. 
\end{align}

\textbf{Distance metric:}\label{apd:prel_dist}
As the Tucker decomposition of $\bcZ_\star$ is not uniquely specified, it is necessary to define the scaled distance metric \cite{Tong2022}  between factor quadruples $\bF=(\mL,\mR,\mV,\bcS)$ and $\bF_{\star}=(\mL_{\star},\mR_{\star},\mV_{\star},\bcS_{\star})$ as:
\begin{align*}
&\distsq{\bF}{\bF_{\star}} \coloneqq \\&\inf_{\bQ_{k}\in\GL(r)}\; \left\|(\mL\bQ_{1}-\mL_{\star})\bSigma_{\star,1}\right\|_{F}^{2}+\left\|(\mR\bQ_{2}-\mR_{\star})\bSigma_{\star,2}\right\|_{F}^{2}\\&+\left\|(\mV\bQ_{3}-\mV_{\star})\bSigma_{\star,3}\right\|_{F}^{2} \nonumber+\left\|(\bQ_{1}^{-1},\bQ_{2}^{-1},\bQ_{3}^{-1})\bcdot\bcS-\bcS_{\star}\right\|_{F}^2. 
\end{align*}
Here, $\GL(r)$ denotes the set of invertible matrices in $\C^{r\times r}$.  For $k=1,2,3$, the existence of $\bQ_{k}$ is shown in \cite[Lemma~12]{Tong2022}. When optimal $\bQ_{k}$ to  $\distsq{\bF}{\bF_{\star}}$ is an identity matrix, we call the factor quadruple $\mF$ is aligned with $\mF_\star$. 

\begin{lemma}[Useful incoherence results]\label{lem:incoh_useful}
    Suppose $\bcZ_\star$ satisfies the incoherence property in Definition \ref{def:incoh}, then 
     \begin{align*}
    \|\bcZ_\star\|_{\infty}
    &\leq {\frac{\mu_0 c_{\mathrm{s}} r}{n}}\sigma_{\max}(\bcZ_\star). \numberthis \label{eq:prel_incohT}
 \\  \max\{\|\M_i(\bcZ_{\star})\|_{2,\infty},\|&\M_i(\bcZ_{\star})^H\|_{2,\infty}\}{\leq} \sqrt{\frac{\mu_0 c_{\mathrm{s}} r}{n}}\sigma_{\max}(\bcZ_\star), \numberthis \label{eq:prel_incoh12}
   \end{align*}  
where $i=1,2$.
\end{lemma}

\begin{proof}
Take $i=1$ for example:
    \begin{align*}
&\|\M_1(\bcZ_{\star})\|_{2,\infty}=\|\mL_\star\M_1(\bcS_{\star})(\mV_{\star}\otimes\mR_{\star})^T\|_{2,\infty}
    \\ &\leq\|\mL_\star\|_{2,\infty}\|\M_3(\bcS_{\star})(\mV_{\star}\otimes\mR_{\star})^T\|\leq \|\mL_\star\|_{2,\infty}\sigma_{\max}(\bcZ_\star),
    \end{align*}
    where we invoke $\|\mV_\star\|\leq 1$ and $\|\mR_\star\|\leq 1$. From Definition \ref{def:incoh}, $\|\mL_\star\|_{2,\infty}\leq \sqrt{\mu_0 c_{\mathrm{s}} r/n}$. Similarly, 
\begin{small}
    \begin{align*}
    \|\cM_1(\bcZ_{\star})^H\|_{2,\infty}&{\leq} \|(\mV_{\star}\otimes\mR_{\star})\cM_3(\bcS_\star)^T\|_{2,\infty}\|\mL_\star\|
    \\&
    {\leq}\|\mR_{\star}\|_{2,\infty}\sigma_{\max}(\bcZ_\star),
\end{align*}
\end{small}where we invoke $\|\mV_\star\|_{2,\infty}\leq\|\mV_\star\|\leq 1$ and $\|\mL_\star\|\leq 1$. 

From $\|\bcZ_\star\|_{\infty}=\|\cM_1(\bcZ_{\star})\|_{\infty}$ and then 
\begin{align*}
    \|\cM_1(\bcZ_{\star})\|_{\infty}&\leq \|\mL_\star\|_{2,\infty} \|(\mV_{\star}\otimes\mR_{\star})\M_3(\bcS_{\star})^T\|_{2,\infty} 
    \\&\leq\|\mL_\star\|_{2,\infty}\|\mR_\star\|_{2,\infty}\|\mV_{\star}\|_{2,\infty}\sigma_{\max}(\bcZ_\star)
    \\&\leq\|\mL_\star\|_{2,\infty}\|\mR_\star\|_{2,\infty}\sigma_{\max}(\bcZ_\star),
\end{align*}
where $\|\mL_\star\|_{2,\infty}, \|\mR_\star\|_{2,\infty}$ are bounded in Definition \ref{def:incoh}. 
\end{proof}
The projection operator $\P_T(\cdot)$, interacted with the Hankel tensor sampling basis $\bcH_{k,j}$ in Definition \ref{def:HankelT}, has the following incoherence property: 
\begin{lemma} \label{lem:PT_HTbasis_bd}
   Denote a tensor $\bcZ\in\C^{n_1\times n_2 \times s}$, and  matrices $\mA\in\C^{n_1\times r}$, $\mB\in\C^{n_2\times r}$. 
   Define a self-adjoint projection operator as
   $\P_{T}(\bcZ)=(\mP_{A},\mP_{B},\mI_s)\bcdot \bcZ,$ where $\mP_{A}=\mA(\mA^H \mA)^{-1}\mA^H$ and $\mP_{B}=\mB(\mB^H \mB)^{-1}\mB^H$ 
   are two projection matrices. Then one has 
    \begin{align*}
        \max_{k,j}\|\P_{T}\bcH_{k,j}\|_F\leq \min\{\|\mA\|_{2,\infty}/\sigma_r(\mA),\|\mB\|_{2,\infty}/\sigma_r(\mB)\}. 
    \end{align*}
\end{lemma}
\begin{proof}
We first list the following bounds:
    \begin{align*}
        &\|\mP_A(:,i_1)\|_2=\|[\mA(\mA^H\mA)^{-1}\mA^H](:,i_1)\|_2
        \\&{=}\|\mA(\mA^H\mA)^{-1}\cdot\mA^H(:,i_1)\|_2\\&
        {\leq}\|\mA(\mA^H\mA)^{-1}\|\|\mA\|_{2,\infty}{=}\|\mA\|_{2,\infty}/\sigma_r(\mA),
    \end{align*}
    and similarly, we can obtain $\|\mP_B(:,i_2)\|_2\leq \|\mB\|_{2,\infty}/\sigma_r(\mB)$.

    From the definition of $\P_T$, we have
  \begin{align*}
        &\|\P_{T}\bcH_{k,j}\|_F=\|(\mP_{A},\mP_{B},\mI_s)\bcdot \bcH_{k,j}\|_F
        \\&=\|\mP_{A}\cM_1(\bcH_{k,j})(\mP_{B}\otimes\mI_s)^T\|_F
        \leq \|\mP_{A}\cM_1(\bcH_{k,j})\|_F
        \\&=\|\mP_{A}\mH_k\|_F
        \leq\max_{i_1}\|\mP_{A}(:,i_1)\|_2
        \leq\|\mA\|_{2,\infty}/\sigma_r(\mA),
    \end{align*}
    where $\mH_k$ is the $k$-th Hankel matrix basis in Definition \ref{def:Hankelm}. Similarly, we can establish that 
$$\|\P_{T}\bcH_{k,j}\|_F\leq \|\mB\|_{2,\infty}/\sigma_r(\mB).$$
\end{proof}
\vspace{-15pt}
\subsection{Proof~of~Lemma~\ref{lem:mulrank_HT}} \label{pf:mulrank_HT}
{
\begin{fact} \label{lem:rank_inq}
    Let $\mA\in\C^{m\times r}, \mB\in\C^{n\times r}$ where $m>r$, $n>r$. 
        \begin{itemize}  
            \item [(1)] When $\rank(\mA)=r$, $\rank(\mB)=r$, we have $\rank(\mA\mB^H)=r$. 
            \item [(2)]  
            When $\rank(\mA)\geq 1$ and $\rank(\mB)\geq 1$, then  $\rank(\mA\odot\mB)\geq\min(\rank(\mA)+\rank(\mB)-1,r)$. 
        \end{itemize}
    \end{fact}
    \begin{proof}
    \begin{align*}
      r = \rank(\mA)&=\rank\(\mA\mB^H\mB(\mB^H\mB)^{-1}\) \\&\leq \rank(\mA\mB^H)
     \leq \rank(\mA)=r, 
    \end{align*}
    and thus $\rank(\mA\mB^H)=r$. 
        (2) can be found in 
        \cite[Lemma~1]{Sidiropoulos2000}.  
    \end{proof}
   Now we provide the formal proof of Lemma~\ref{lem:mulrank_HT}. 
As $\H(\mX_\star){=}(\mE_L,\mE_R,\mB)\bcdot\bcD=\sum_{k=0}^{r-1} \va_{n_1}(p_k)\circ\va_{n_2}(p_k)\circ\vb_k$,  it is obvious 
$
    \cM_1(\H(\mX_\star))=\mE_L(\mB\odot\mE_R)^T,
    \cM_2(\H(\mX_\star))=\mE_R(\mB\odot\mE_L)^T, 
    \cM_3(\H(\mX_\star))=\mB(\mE_R\odot\mE_L)^T, 
$
where $\odot$ denotes Khatri-Rao product, $\mE_L\in\C^{n_1\times r}$, $\mE_R\in\C^{n_2\times r}$ and $\mB\in\C^{s\times r}$. 

We prove the rank of $  \cM_1(\H(\mX_\star))=\mE_L(\mB\odot\mE_R)^T$ first. 
 Construct a Vandermonde matrix as $\mE_r=\begin{bmatrix}
    \va_{r}(p_0),\cdots,\va_{r}(p_{r-1})
\end{bmatrix}\in\C^{r\times r}$ where $\va_{r}(p_k)=[1,p_k,\cdots,p_k^{r-1}]^T$ by choosing the first $r$ rows of $\mE_L$. It is well known that the rank of the Vandermonte matrix is $\rank(\mE_r)=r$ when $\{p_k\}_{k=0}^{r-1}$ are distinct. As $\mE_r$ is a submatrix of $\mE_L$, then  $r=\rank(\mE_r)\leq\rank(\mE_L)\leq\min\{n_1,r\}=r$ from $n_1=O(n)$, $n_2=O(n)$. Thus we conclude $\rank(\mE_L)=r$. 

Invoking (2) of Fact~\ref{lem:rank_inq}, we know $\rank(\mB\odot\mE_R)\geq r$ from the facts $\rank(\mB)=r$ and $\rank(\mE_R)=r$. As $\mB\odot\mE_R\in\C^{sn_2\times r}$, we conclude that $\rank(\mB\odot\mE_R)=r$. 
Then invoking (1) of Fact~\ref{lem:rank_inq}, we obtain that $\rank(\cM_1(\H(\mX_\star)))=r$. Similarly, we can prove $\rank(\cM_2(\H(\mX_\star)))=r$, $\rank(\cM_3(\H(\mX_\star)))=r$. 

}

    \vspace{-10pt}
\section{Fast computation} \label{apd:fast_comput}
\subsection{Proof of Lemma~\ref{lem:HankelT_algebra}.}
\label{pf:HankelT_algebra}
  1)Lemma~\ref{lem:HankelT_algebra}.a. Let $\tilde{\bcZ}=\bcS\times_1\mL\times_2\mR\in\C^{n_1\times n_2 \times r}$, then
\begin{align*}
    \mZ=\G^{*}(\tilde{\bcZ}\times_3\mV)=\mV\G^{*}(\tilde{\bcZ}),
\end{align*}
where we invoke \eqref{eq:DeHankel_mul3} proved later. 
For 
$k\in[s]$ and $a\in[n]$
\begin{align*}
   & [\G^{*}(\tilde{\bcZ})](k,a)=\frac{1}{\sqrt{w_{a}}}\sum_{i_1+i_2=a}[\bcS\times_1\mL\times_2\mR](i_1,i_2,k)
   \\&=\frac{1}{\sqrt{w_{a}}}\sum_{i_1+i_2=a}\sum_{j_2=0}^{r-1}\sum_{j_1=0}^{r-1}\bcS(j_1,j_2,k)\mL(i_1,j_1)\mR(i_2,j_2)
   \\&=\sum_{j_1,j_2}\bcS(j_1,j_2,k)\frac{1}{\sqrt{w_{a}}}\sum_{i_1+i_2=a}\mL(i_1,j_1)\mR(i_2,j_2)
   \\&=\sum_{j_1,j_2}\bcS(j_1,j_2,k)\bcW(j_1,j_2,a)=\cM_3(\bcS)\cM_3(\overline{\bcW})^H{(k,a)}.
\end{align*}
2) Proof of Lemma~\ref{lem:HankelT_algebra}.b. We first prove \eqref{eq:Hankellift_mul3}. For $i_1\in[n_1]$, $i_2\in[n_2]$, and $j_3\in[r]$:
    \begin{align*}
        &[\G(\mE)\times_{3}\mV^H](i_1,i_2,j_3)=\sum_{i_3=0}^{s-1}[\G(\mE)](i_1,i_2,i_3)\overline{\mV}(i_3,j_3)
        \\&=\frac{1}{\sqrt{w_{i_1+i_2}}}\sum_{i_3=0}^{s-1}\mE(i_3,i_1+i_2)\overline{\mV}(i_3,j_3)
        \\&=\frac{1}{\sqrt{w_{i_1+i_2}}}[\mV^H\mE](j_3,i_1+i_2)
        =[\G(\mV^H\mE)](i_1,i_2,j_3).
    \end{align*}
    
Then we prove \eqref{eq:DeHankel_mul3}. Let $\mZ=\G^{*}(\tilde{\bcZ}\times_{3}\mV)$, and we rewrite that for $k\in[s]$ and $a\in[n]$,
     \begin{align*}
        &\mZ(k,a)=\frac{1}{\sqrt{w_a}}\sum_{i_1+i_2=a}\sum_{i_3=0}^{r-1}\tilde{\bcZ}(i_1,i_2,i_3)\mV(k,i_3)
        \\&=\sum_{i_3=0}^{r-1}\(\frac{1}{\sqrt{w_a}}\sum_{i_1+i_2=a}\tilde{\bcZ}(i_1,i_2,i_3)\)\mV(k,i_3)=\mV\G^{*}(\tilde{\bcZ}){(k,a)}.
    \end{align*}


3) Proof of Lemma~\ref{lem:HankelT_algebra}.c. We rewrite that
 \begin{align*}
        &[\G(\tilde{\mE})\times_{1}\mL^H\times_{2}\mR^H](j_1,j_2,i_3)\\
        &=\sum_{i_1,i_2}[\G(\tilde{\mE})](i_1,i_2,i_3)\overline{\mL}(i_1,j_1)\overline{\mR}(i_2,j_2)
        \\&=\sum_{i_1,i_2}\frac{1}{\sqrt{w_{i_1+i_2}}}\tilde{\mE}(i_3,i_1+i_2)\overline{\mL}(i_1,j_1)\overline{\mR}(i_2,j_2)
         \\&= \sum_{a=0}^{n-1}\sum_{i_1+i_2=a}\frac{1}{\sqrt{w_{a}}}\tilde{\mE}(i_3,a)\overline{\mL}(i_1,j_1)\overline{\mR}(i_2,j_2)
        \\&= \sum_{a=0}^{n-1}\frac{1}{\sqrt{w_{a}}}\tilde{\mE}(i_3,a)\sum_{i_1+i_2=a}\overline{\mL}(i_1,j_1)\overline{\mR}(i_2,j_2)
        \\&= \sum_{a=0}^{n-1}\tilde{\mE}(i_3,a)\overline{\bcW}(j_1,j_2,a)=\overline{\bcW}\times_{3}\tilde{\mE}{(j_1,j_2,i_3)},
    \end{align*}
    where $i_1\in[n_1]$, $i_2\in[n_2]$ in the second line.

\begin{lemma}[Fast computation of the scaled terms] \label{lem:fast_scaleterm}
    We give the fast computation of $ \breve{\mL}^H\breve{\mL}, \breve{\mR}^H\breve{\mR},   \breve{\mV}^H\breve{\mV}$ as follows, where $\breve{\mL},\breve{\mR},\breve{\mV}$ are defined in \eqref{eq:breve_lrv}. 
    \begin{align*}
        \breve{\mL}^H\breve{\mL}&=\cM_{1}\((\mI_r,\mR^H\mR,\mV^H\mV)\bcdot\bcS\)\cM_1(\bcS)^H,
       \\ \breve{\mR}^H\breve{\mR}&=\cM_{2}\((\mR^H\mR,\mI_r,\mV^H\mV)\bcdot\bcS\)\cM_2(\bcS)^H,
       \\ \breve{\mV}^H\breve{\mV}&=\cM_{3}\((\mL^H\mL,\mR^H\mR,\mI_r)\bcdot\bcS\)\cM_3(\bcS)^H,
    \end{align*}
    and this costs $O((s+n)r^2+r^4)$ flops in total. 
\end{lemma} 
\begin{proof}
    Taking $\breve{\mL}^H\breve{\mL}$ for example:
    \begin{align*}
        \breve{\mL}^H\breve{\mL}&=\cM_{1}(\bcS ) (\overline{\mV} \otimes\overline{\mR})^H(\overline{\mV} \otimes\overline{\mR} )\cM_{1}(\bcS )^{H}
        \\&=\cM_{1}(\bcS )\(\mV^H\mV\otimes\mR^H\mR\)^T\cM_{1}(\bcS )^{H}
        \\&=\cM_{1}\((\mI_r,\mR^H\mR,\mV^H\mV)\bcdot\bcS\)\cM_1(\bcS)^H
    \end{align*}
    Note that we should compute $\breve{\mL}^H\breve{\mL}$ from the last equality, which is realized by mode-$i$ tensor product sequentially. 
\end{proof}

\vspace{-15pt}
\section{Lemmas for Initialization}
\begin{lemma}[Initialization] \label{lem:init}
  Suppose $\bcZ_{\star}\in\C^{n_1\times n_2\times s}$ is  incoherent in Definition \ref{def:incoh}, the projection radius in \eqref{eq:incoh_proj} is $B=C_{B}\sqrt{\mu_0 c_{\mathrm{s}}  r}\sigma_{\max}(\bcZ_\star)$ for $C_B\geq (1+\varepsilon_0)^3$, and $\varepsilon_0>0$ is a small constant.  For $\mF_0=(\mL^0,\mR^0,\mV^0,\bcS^0)$ in Algorithm \ref{alg:init}, 
  with probability at least $1-O\((sn)^{-2}\)$
\begin{align*}
    \dist{\mF^0}{\mF_{\star}} \le \varepsilon_0\sigma_{\min}(\bcZ_\star),
\end{align*}
holds provided $\hat{m}\geq O(\varepsilon_0^{-2}\mu_0 c_{\mathrm{s}} sr^3\kappa^2\log(sn))$. Besides, 
  \begin{align*}
      \max\{\|\mL^0(\breve{\mL}^0)^H\|_{2,\infty},\|\mR^0(\breve{\mR}^0)^H\|_{2,\infty}\}\leq B/\sqrt{n}.  \numberthis \label{eq:incoh_proj_F0}
  \end{align*}
\end{lemma}


\begin{proof}
{See Appendix~\ref{pf_leminit} of the Supplementary Material.}
\end{proof}
\begin{lemma} \label{lem:init_V0}
Let $\mathrm{SVD}_r(\cdot)$ return the top-$r$ left singular vectors of a matrix. For a tensor $\bcZ\in\C^{n_1\times n_2 \times s}$ and a orthogonal matrix $\mL\in\C^{n_1\times r}$ where $\mL^H\mL=\mI_r$, we have 
\begin{align*}
  \mathrm{SVD}_{r}(\cM_3(\bcZ\times_1 \mL^H))=  \mathrm{SVD}_{r}(\cM_3(\bcZ\times_1 \mL\mL^H)). 
\end{align*}
\end{lemma}
\begin{proof}
    {See Appendix~\ref{pf_init_V0} of the Supplementary Material.}
\end{proof}
The following concentration inequality for the Hankel tensor sampling is a key hammer for the sequential spectral initialization in our work.
\begin{lemma}\label{lem:init_conc_m1}
For $\bcZ_\star\in\C^{n_1\times n_2\times s}$ which is incoherent in Definition \ref{def:incoh}, $\tilde{\Omega}$ 
is any index set with $\tilde{m}$ samples 
and let $\tilde{p}=\frac{\tilde{m}}{sn}$, then 
  \begin{align*}
      \|\cM_i\((\tilde{p}^{-1}\G\P_{\tilde{\Omega}}\G^*-\I)(\bcZ_\star)\)\|{\lesssim} \sqrt{\frac{\mu_0 c_{\mathrm{s}} sr{\log(sn)}}{\tilde{m}}}\sigma_{\max}(\bcZ_{\star})
  \end{align*}
  holds with probability at least $1-(sn)^{-2}$ when $\tilde{m}\geq \mu_0 c_{\mathrm{s}} sr \log(sn)$, where $i=1,2$.
\end{lemma}
\begin{proof}
  {See Appendix~\ref{pf_lem_initconc_m1} of the Supplementary Material.}
\end{proof}

\begin{lemma}[Properties of projection] \label{lem:proj}
Suppose $\bcZ_\star$ is { $\mu_0$-incoherent} in Definition \ref{def:incoh}, $\dist{\mF'}{\mF_\star}\leq \varepsilon_0\sigma_{\min}(\bcZ_\star)$ where $\varepsilon_0<1$ and $\mF'=(\mL',\mR',{\mV},\bcS)$. The projection radius in \eqref{eq:incoh_proj} is set as $B=C_{B}\sqrt{\mu_0 c_{\mathrm{s}} r}\sigma_{\max}(\bcZ_\star)$ for $C_{B}\geq(1+\varepsilon_0)^3$, and $({\mL},{\mR})=\PB{\mL',\mR'}$. Then for $\mF=({\mL},{\mR},{\mV},\bcS)$, we have 
\begin{align*}
    \dist{\mF}{\mF_\star}\leq\dist{\mF'}{\mF_\star},
\end{align*}
and the incoherence condition:
  \begin{align*}
      \max\{\|\mL\breve{\mL}^H\|_{2,\infty},\|\mR\breve{\mR}^H\|_{2,\infty}\}\leq B/\sqrt{n}, \numberthis \label{eq:incoh_proj_full}
  \end{align*}
\end{lemma}
\begin{proof}
    This is adapted from \cite[ Lemma~7]{Tong2022}. 
\end{proof}
\section{Lemmas for Local linear convergence}
\begin{lemma}[Linear convergence] \label{lem:linconverge}
  Suppose $\bcZ_{\star}\in\C^{n_1\times n_2\times s}$ is incoherent in Definition \ref{def:incoh}. The current iterate $\mF=(\mL,\mR,\mV,\bcS)$ satisfies  $\dist{\mF}{\mF_{\star}}\leq \varepsilon\sigma_{\min}(\bcZ_\star)$ for some small constant $\varepsilon$, and the following incoherence condition:
  \begin{align*}
      \max\{\|\mL\breve{\mL}^H\|_{2,\infty},\|\mR\breve{\mR}^H\|_{2,\infty}\}\leq B/\sqrt{n}, \numberthis \label{eq:incoh_iter_full}
  \end{align*}
  where $B=C_B\sqrt{\mu_0 c_{\mathrm{s}} r}\sigma_{\max}(\bcZ_\star)$ for $C_B\geq (1+\varepsilon)^3$. If $\mF$
is independent of the {current sampling set $\tilde{\Omega}$ with $\tilde{m}$ samples},
 with probability at least $1-O\((sn)^{-2}\)$ we have
\begin{align*}
    \dist{\mF_{+}}{\mF_{\star}} \leq (1-0.5\eta) \dist{\mF}{\mF_{\star}},
\end{align*}
provided {$\tilde{m}$}
$\geq O(\varepsilon^{-2}\mu_0 c_{\mathrm{s}} sr\kappa^2\log(sn))$ and $\eta\leq 0.4$, where the next iterate $\mF_{+}=(\mL_{+},\mR_{+},\mV_{+},\bcS_{+})$ is updated in \eqref{eq:ScalHT}.
\end{lemma}
\begin{proof}
 {See Appendix~\ref{pf_lem_lin_cvg} of the Supplementary Material.} 
\end{proof}

The following concentration inequality for the Hankel tensor sampling is a key hammer in our problem, greatly different from the hammers Lemma~18, 20 in \cite{Tong2022}.  
\begin{lemma}\label{lem:PTconc}
    For $i=1,2$, denote $\P_{T_i}:\C^{n_1\times n_2\times s}\rightarrow\C^{n_1\times n_2\times s}$ as two self-adjoint projection operators that are independent of {the current sampling set $\tilde{\Omega}$ with $\tilde{m}$ samples, $\tilde{p}=\frac{\tilde{m}}{sn}$}, and $q_i = \max_{k,j}\|\P_{T_i}\bcH_{k,j}\|_F$ where $\bcH_{k,j}$ is the Hankel tensor basis defined in Definition \ref{def:HankelT}. Then with probability at least $1-(sn)^{-2}$, one has:
    \begin{align*}
        \|\P_{T_2}\G(\tilde{p}^{-1}\P_{\tilde{\Omega}}-\I)\G^{*}\P_{T_1}\|\lesssim
        \sqrt{\frac{\max\{q_1^2,q_2^2\}sn\log(sn)}{\tilde{m}}},
    \end{align*}
    provided $\tilde{m}\geq O(\max\{q_1^2,q_2^2\}sn\log(sn))$.
  
\end{lemma}
\begin{proof}
{See Appendix~\ref{pf_lem_PTconc} of the Supplementary Material.}
\end{proof}

\clearpage
\begin{center}
  \Large\textbf{Supplementary Material}
\end{center}
\vspace{-15pt}
\renewcommand{\baselinestretch}{0.84}
\section{ }
\subsection{Proofs of Lemma~\ref{lem:init}} \label{pf_leminit}
The main task is to prove the following inequality 
\begin{align*}
    \dist{\mF'^0}{\mF_{\star}} &\leq(\sqrt{2}+1)^{3/2}\|(\mL'^{0},\mR'^{0},{\mV}^{0})\bcdot\bcS^{0}-\bcZ_{\star}\|_F
    \\&\leq \varepsilon_0 \sigma_{\min}(\bcZ_\star),
\end{align*}
where $\mF'^0=(\mL'^0,\mR'^0,{\mV}^0,\bcS^0)$ is shown in Algorithm \ref{alg:init}, and the first inequality results from Lemma~\ref{lem:dist_fullerr_upbd}. When the previous inequality hold, invoking Lemma~\ref{lem:proj} we can establish
\begin{align*}
    \dist{\mF^0}{\mF_{\star}} &\leq \dist{\mF'^0}{\mF_{\star}} \leq\varepsilon_0 \sigma_{\min}(\bcZ_\star),
\end{align*}
and the incoherence condition \eqref{eq:incoh_proj_F0}. 


   Next, the main task is to bound $\|(\mL'^{0},\mR'^{0},{\mV}^{0})\bcdot\bcS^{0}-\bcZ_{\star}\|_F$. Let $\bP_{L}\coloneqq\mL'^{0}(\mL'^{0})^{H}$,  $\bP_{R}\coloneqq\mR'^{0}(\mR'^{0})^{H}$ and $\bP_{V}\coloneqq{\mV}^{0}({\mV}^{0})^{H}$ be the projection matrix onto the column space of $\mL'^{0}$, $\mR'^{0}$ and ${\mV}^{0}$. Denote $\bP_{L_\perp}$, $\bP_{R_\perp}$ and $\bP_{V_\perp}$ as their orthogonal complement.  We have the decomposition
\begin{align*}
\bcZ_{\star} &= (\bP_{L},\bP_{R},\bP_{V})\bcdot\bcZ_{\star} + (\bP_{L_\perp},\bP_{R},\bP_{V})\bcdot\bcZ_{\star} \\
&\quad+ (\bI_{n_1}, \bP_{R_\perp},\bP_{V})\bcdot\bcZ_{\star} + (\bI_{n_1}, \bI_{n_2},\bP_{V_\perp})\bcdot\bcZ_{\star}.
\end{align*}
Equivalently denote  $\bcZ^0=\hat{p}^{-1}\G\P_{\Omega_0}\G^*\bcZ_\star,$
we have 
\begin{align}
& \left\|(\mL'^{0},\mR'^{0},{\mV}^{0})\bcdot\bcS^{0}-\bcZ_{\star}\right\|_{F} = \left\|\bcZ_{\star}-(\bP_{L},\bP_{R},\bP_{V})\bcdot\bcZ^0\right\|_{F} \nonumber\\
&= \big\|(\bP_{L},\bP_{R},\bP_{V})\bcdot(\bcZ_{\star}-\bcZ^0)-(\bP_{L},\bP_{R},\bP_{V})\bcdot\bcZ_{\star}+\bcZ_\star\big\|_F
\nonumber 
\\& \leq  \left\|\bP_{L_\perp}\cM_{1}(\bcZ_{\star})\right\|_F {+} \left\|\bP_{R_\perp}\cM_{2}(\bcZ_{\star})\right\|_F {+} \left\|\bP_{V_\perp}\cM_{3}(\bcZ_{\star})\right\|_F \nonumber\\
&\quad+\left\|(\bP_{L},\bP_{R},\bP_{V})\bcdot(\bcZ^0-\bcZ_{\star})\right\|_F,
\label{eq:TC_init_expand}
\end{align}
where we invoke $\bcS^0=\big(({\mL'}^0)^H,({\mR'}^0)^H,(\mV^0)^H\big)\bcdot \bcZ^0$ and the previous decomposition of $\bcZ_\star$ in the inequality. 
We bound the previous four terms separately as follows. 

\paragraph{Bounding $\left\|\bP_{i_\perp}\cM_{i}(\bcZ_{\star})\right\|_F$, where $i=1,2$}~
\begin{align*}
    &\left\|\bP_{L_\perp}\cM_{1}(\bcZ_{\star})\right\|_F\leq \sqrt{r}\|\bP_{L_\perp}\cM_{1}(\bcZ_{\star})\|
    \\&\leq\sqrt{r}\(\|\bP_{L_\perp}\cM_{1}(\bcZ^0-\bcZ_\star)\|+\|\bP_{L_\perp}\cM_{1}(\bcZ^0)\|\)
    \\&\leq\sqrt{r}\(\|\cM_{1}(\bcZ^0-\bcZ_\star)\|+\sigma_{r+1}\(\cM_1(\bcZ^0)\)\)
    \\&\leq2\sqrt{r}\|\cM_{1}(\bcZ^0-\bcZ_\star)\|,\numberthis\label{eq:init_part_main}
\end{align*}
where the first inequality results from $\cM_1(\bcZ_\star)$ is a rank-$r$ matrix, in the third line we use the fact that $\|\bP_{L_\perp}\cM_{1}(\bcZ^0)\|=\sigma_{r+1}(\cM_1(\bcZ^0))$ and the last line results from $\sigma_{r+1}(\cM_{1}(\bcZ^0))-\sigma_{r+1}(\cM_1(\bcZ_\star))\leq\|\cM_{1}(\bcZ^0-\bcZ_\star)\|$ and $\sigma_{r+1}(\cM_1(\bcZ_\star))=0$.
Invoking Lemma~\ref{lem:init_conc_m1} to obtain that 
\begin{align*}
    \sqrt{r}\|\cM_{i}(\bcZ^0-\bcZ_\star)\|\lesssim \varepsilon_0 \sigma_{\min}(\bcZ_\star)
\end{align*}
for $i=1,2$, provided $\hat{m}\geq O(\varepsilon_0^{-2}\mu_0 c_{\mathrm{s}} sr^2\kappa^2 \log(sn)) $. 
\paragraph{Bounding $\left\|\bP_{V_\perp}\cM_{3}(\bcZ_{\star})\right\|_F$} 
~

If we  initialize $\mV^0$ from $\cM_3(\bcZ^0)$, 
we need to bound $\|\cM_{3}(\bcZ^0-\bcZ_\star)\|$. However, for $i=1,2$ $\|\cM_i(\bcZ^0-\bcZ_\star)\|$ can be sharply bounded as  $\|\cM_i(\bcH_{k,j})\|=1/\sqrt{w_k}$, while $\|\cM_3(\bcZ^0-\bcZ_\star)\|$ is larger as $\|\cM_3(\bcH_{k,j})\|=1$ in Definition \ref{def:HankelT}. 
Therefore, we don't initialize $\mV^0$ from $\cM_3(\bcZ^0)$. 

We propose to initialize $\mV^0$ via the top-$r$ left singular vectors of the intermediary quantity $\cM_3(\bcZ^0\times_1({\mL'}^0)^H)$ and name this method as sequential spectral initialization. From Lemma~\ref{lem:init_V0}, this is equivalent to initialize $\mV^0$  via the top-$r$ left singular vectors of $\hat{\bcZ}^0=\bcZ^0\times_{1} \mL'^0(\mL'^0)^H$. 
Then we follow the route in \eqref{eq:init_part_main} to establish that 
\begin{align*}
    &\left\|\bP_{V_\perp}\cM_{3}(\bcZ_{\star})\right\|_F\leq 2 \sqrt{r} \|\cM_{3}(\hat{\bcZ}^0-\bcZ_\star)\|
    \\&\leq 2\sqrt{r}\|\cM_{3}(\hat{\bcZ}^0-\bcZ_\star)\|_F=2\sqrt{r}\|\M_{1}(\hat{\bcZ}^0-\bcZ_\star)\|_F.
\end{align*}
For $\|\M_{1}(\hat{\bcZ}^0-\bcZ_\star)\|_F$, we can prove that it is sharply small: 
\begin{align*}
 &\|\M_{1}(\hat{\bcZ}^0-\bcZ_\star)\|_F
\overset{(a)}{\leq} \sqrt{2r}\|\M_{1}(\hat{\bcZ}^0-\bcZ_\star)\|
    \\&\leq \sqrt{2r}(\|\cM_1(\hat{\bcZ}^0-{\bcZ}^0)\|{+}\|\cM_1({\bcZ}^0-\bcZ_\star)\|)
    \\&\overset{(b)}{\leq} 2\sqrt{2r}\|\cM_1({\bcZ}^0-\bcZ_\star)\|\overset{(c)}{\lesssim} \varepsilon_0 \sigma_{\min}(\bcZ_\star)/\sqrt{r},
\end{align*}
 where $(a)$ results from $\cM_1(\hat{\bcZ}^0)$ and $\cM_1(\bcZ_\star)$ are rank-$r$, and $(b)$ results from $\cM_1(\hat{\bcZ}^0)=\mL'^0(\mL'^0)^H\cM_1(\bcZ^0)$, which is rank-$r$ approximation of $\cM_1(\bcZ^0)$, thus $\|\cM_1(\hat{\bcZ}^0-{\bcZ}^0)\|\leq\|\cM_1({\bcZ}^0-\bcZ_\star)\|$. $(c)$ results from  Lemma~\ref{lem:init_conc_m1}, provided $\hat{m}\geq O(\varepsilon_0^{-2}\mu_0 c_{\mathrm{s}} sr^3\kappa^2\log(sn))$. Consequently, we obtain that
 $$\left\|\bP_{V_\perp}\cM_{3}(\bcZ_{\star})\right\|_F\lesssim \varepsilon_0 \sigma_{\min}(\bcZ_\star).
 $$

\paragraph{Bounding $\left\|(\bP_{L},\bP_{R},\bP_{V})\bcdot(\bcZ^0{-}\bcZ_{\star})\right\|_F$}~
By \eqref{eq:tensor_properties_3}, 
\begin{align*}
    &\left\|(\bP_{L},\bP_{R},\bP_{V})\bcdot(\bcZ^0{-}\bcZ_{\star})\right\|_F
    \leq \|(\bP_{L},\mI_{n_2},\mI_{s})\bcdot(\bcZ^0{-}\bcZ_{\star})\|_F
    \\&= \|\mP_L\cM_1(\bcZ^0-\bcZ_\star)\|_F\leq \sqrt{r}\|\cM_1(\bcZ^0-\bcZ_\star)\|.
\end{align*}
Invoking Lemma~\ref{lem:init_conc_m1}, we have 
\begin{align*}
    \sqrt{r}\|\cM_1(\bcZ^0-\bcZ_\star)\|\lesssim \varepsilon_0 \sigma_{\min}(\bcZ_\star),
\end{align*}
provided $\hat{m}\geq O(\varepsilon_0^{-2}\mu_0 c_{\mathrm{s}} sr^2\kappa^2 \log(sn))$.
\subsection{Proofs of Lemma~\ref{lem:init_V0}} \label{pf_init_V0}
  Let $\hat{\mV}$ consist of the top-$r$ left singular vectors of $\cM_3(\bcZ\times_1 \mL^H)$,  which can be equivalently reformulated as: 
    \begin{align*}
        \hat{\mV}\coloneqq&\argmin_{\mV^H\mV=\mI_r,\mV\in\C^{s\times r}}\|\mV\mV^H\cM_3(\bcZ\times_1 \mL^H)\|_F^2
        \\\coloneqq&\argmin_{\mV^H\mV=\mI_r,\mV\in\C^{s\times r}}\|\mV\mV^H\cM_3(\bcZ\times_1 \mL\mL^H)\|_F^2,
    \end{align*}
    where we invoke
    \begin{align*}
       & \|\mV\mV^H\cM_3(\bcZ\times_1 \mL^H)\|_F^2=\|\bcZ\times_1 \mL^H\times_3 \mV\mV^H\|_F^2\\
       &=\|\mL^H\cM_1(\bcZ\times_3 \mV\mV^H)\|_F^2=\|\mL\mL^H\cM_1(\bcZ\times_3 \mV\mV^H)\|_F^2\\
       &=\|\bcZ\times_1 \mL\mL^H\times_3 \mV\mV^H\|_F^2=\|\mV\mV^H\cM_3(\bcZ\times_1 \mL\mL^H)\|_F^2.
    \end{align*}
\subsection{Proofs of Lemma~\ref{lem:init_conc_m1}} \label{pf_lem_initconc_m1}
 We prove the lemma for the case $i=1$, and the same holds for the case $i=2$. 
 Denote $\mZ_{\star}=\cM_1(\bcZ_\star)$ and $\mH_{k,j}=\cM_1(\bcH_{k,j})$. 
   We reformulate the target term as 
    \begin{align*}
      &\cM_1\((\tilde{p}^{-1}\G\P_{\tilde{\Omega}}\G^*-\I)(\bcZ_\star)\)
      \\&=\sum_{i=1}^{\tilde{m}}\tilde{p}^{-1}\langle\mH_{a_i,b_i},\mZ_{\star}\rangle\mH_{a_i,b_i}-\frac{1}{\tilde{m}}\mZ_{\star}=\sum_{i=1}^{\tilde{m}}\mZ_{a_i,b_i},
    \end{align*}
where each pair $(a_i,b_i)$ is drawn uniformly from $[n]\times[s]$. We list some  results that will be repeatedly used:
\begin{align*}
&\|\langle\mH_{k,j},\mZ_{\star}\rangle\mH_{k,j}\|\leq\|\bcZ_\star\|_{\infty}, \numberthis \label{eq:init_conc1_bd}   \\ &\|\mZ_{\star}\|\leq\|\mZ_{\star}\|_F\leq\sqrt{sn}\min\{\|\mZ_{\star}\|_{2,\infty},\|\mZ_{\star}^H\|_{2,\infty}\}.\numberthis \label{eq:init_conc1_bd2}
\end{align*}
For \eqref{eq:init_conc1_bd},
\begin{align*}
    \|\langle\mH_{k,j},\mZ_{\star}\rangle\mH_{k,j}\|&=\frac{1}{\sqrt{w_k}}\big|\sum_{i_1+i_2=k} \bcZ_\star(i_1,i_2,j)\big|\|\mH_{k,j}\|
    \\&\leq \sqrt{w_k}\|\bcZ_\star\|_{\infty}\|\mH_{k,j}\|=\|\bcZ_\star\|_{\infty},
\end{align*}
where we invoke $\big|\sum_{i_1+i_2=k} \bcZ_\star(i_1,i_2,j)\big|\leq\sum_{i_1+i_2=k}|\bcZ_\star(i_1,i_2,j)|\leq w_k \|\bcZ_\star\|_{\infty}$ and   $\|\mH_{k,j}\|=\|\cM_1(\bcH_{k,j})\|=1/\sqrt{w_k}$ in Definition \ref{def:HankelT}. 

From \eqref{eq:init_conc1_bd}, $\|\mZ_\star\|\leq \|\mZ_\star\|_F\leq\sqrt{sn^2}\|\mZ_\star\|_{\infty}$, \eqref{eq:prel_incohT} that $ \|\bcZ_\star\|_{\infty}
    \leq {\frac{\mu_0 c_{\mathrm{s}} r}{n}}\sigma_{\max}(\bcZ_\star)$, one has 
\begin{align*}
   R= \|\mZ_{a_i,b_i}\|&\leq \frac{1}{\tilde{m}}\|\mZ_{\star}\|+\frac{sn}{\tilde{m}}\|\bcZ_\star\|_{\infty}\leq \frac{2\mu_0 c_{\mathrm{s}} sr}{\tilde{m}}\sigma_{\max}(\bcZ_\star), 
\end{align*}
Besides, 
\begin{align*}
    &\E(\mZ_{a_i,b_i}\mZ_{a_i,b_i}^H)=\frac{sn}{\tilde{m}^2}\sum_{k,j}|\langle\mH_{k,j},\mZ_{\star}\rangle|^2\mH_{k,j}\mH_{k,j}^H\\&-\frac{1}{\tilde{m}^2}\mZ_{\star}\mZ_{\star}^H=\frac{sn}{\tilde{m}^2}\mC-\frac{1}{\tilde{m}^2}\mZ_{\star}\mZ_{\star}^H,
\end{align*}
where we set $\mC=\sum_{k,j}|\langle\mH_{k,j},\mZ_{\star}\rangle|^2\mH_{k,j}\mH_{k,j}^H$. From $\mZ_{\star}=\sum_{k,j}\langle\mH_{k,j},\mZ_{\star}\rangle\mH_{k,j}$, one can verify that $\mC=\diag(\mZ_{\star}\mZ_{\star}^H)$. Therefore  $\|\mC\|\leq\|\mZ_{\star}\|_{2,\infty}^2$. Combining \eqref{eq:init_conc1_bd2}, we have
\begin{align*}
    \|\E(\sum_{i=1}^{\tilde{m}}\mZ_{a_i,b_i}\mZ_{a_i,b_i}^H)\|&{\leq}\max
    \{\frac{sn}{\tilde{m}}\|\mC\|,\frac{1}{\tilde{m}}\|\mZ_{\star}\mZ_{\star}^H\|
    \}
    {\leq}\frac{sn}{\tilde{m}}\|\mZ_{\star}\|_{2,\infty}^2,
\end{align*}
where we invoke \eqref{eq:init_conc1_bd2}. Similarly, 
$
    \|\E(\sum_{i=1}^{\tilde{m}}\mZ_{a_i,b_i}^H\mZ_{a_i,b_i})\|\leq\frac{sn}{\tilde{m}}\|\mZ_{\star}^H\|_{2,\infty}^2.
$
Then from \eqref{eq:prel_incoh12}, one can obtain
\begin{align*}
    \sigma^2 = \frac{sn}{\tilde{m}}\max\{\|\mZ_{\star}\|_{2,\infty}^2,\|\mZ_{\star}^H\|_{2,\infty}^2\}\leq \frac{\mu_0 c_{\mathrm{s}} sr}{\tilde{m}}\sigma_{\max}^2(\bcZ_\star).
\end{align*}
By applying Bernstein inequality \cite[Theorem 1.6]{Tropp2012}, with probability $1-(sn)^{-2}$, when $\tilde{m}\geq\mu_0 c_{\mathrm{s}} sr\log(sn)$  one has
\begin{align*}
    \|\sum_{i=1}^{\tilde{m}}\mZ_{a_i,b_i}\|
    &\lesssim\big(\sqrt{\frac{\mu_0 c_{\mathrm{s}} sr{\log(sn)}}{\tilde{m}}}{+}\frac{{\mu_0 c_{\mathrm{s}} sr}\log(sn)}{\tilde{m}}\big)\sigma_{\max}(\bcZ_{\star})
    \\&\lesssim \sqrt{\frac{\mu_0 c_{\mathrm{s}} sr{\log(sn)}}{\tilde{m}}}\sigma_{\max}(\bcZ_{\star}).
\end{align*}
\vspace{-20pt}
\section{}
\subsection{Proofs of Lemma~\ref{lem:linconverge}} \label{pf_lem_lin_cvg} 
From the definition of the distance metric, we have
\begin{align*}
&\distsq{\bF_{+}}{\bF_{\star}} \leq \left\|( {\mL} _{+}\bQ_{1}-\mL_{\star})\bSigma_{\star,1}\right\|_{F}^{2}\\
&+\left\|( {\mR} _{+}\bQ_{2}-\mR_{\star})\bSigma_{\star,2}\right\|_{F}^{2}+\left\|( {\mV} _{+}\bQ_{3}-\mV_{\star})\bSigma_{\star,3}\right\|_{F}^{2} \nonumber\\&+\left\|(\bQ_{1}^{-1},\bQ_{2}^{-1},\bQ_{3}^{-1})\bcdot\bcS _{+}-\bcS_{\star}\right\|_{F}^2, \numberthis\label{eq:dist_updt_splitbd}
\end{align*}
where $\mQ_1$, $\mQ_2$ and $\mQ_3$ are the alignment matrices of $\distsq{\bF}{\bF_{\star}}$, and the existence of them can be checked from in \cite[Lemma~12]{Tong2022}. 
For simplicity, we suppose that the factor quadruple  $\mF=(\mL, \mR , \mV,\bcS)$ is aligned with $\mF_\star$, which can be constructed by $\mL \leftarrow \mL\mQ_1$, $\mR \leftarrow \mR\mQ_2$, $\mV \leftarrow \mV\mQ_3$ and 
$\bcS \leftarrow (\bQ_{1}^{-1},\bQ_{2}^{-1},\bQ_{3}^{-1})\bcdot \bcS$. Besides, we introduce the following notations:
\begin{align}
    &\bDelta_{L}  \coloneqq\mL-\mL_{\star},~\bDelta_{R}\coloneqq\mR-\mR_{\star},\nonumber\\
~&\bDelta_{V}\coloneqq\mV-\mV_{\star},~ \bDelta_{\cS}\coloneqq\bcS-\bcS_{\star},\nonumber 
 \\
 &\bcE_a\coloneqq(\mL,\mR,\mV)\bcdot\bDelta_S,~ \bcE_b\coloneqq(\bDelta_L,\mR,\mV)\bcdot\bcS_\star,\nonumber \\
 &\bcE_c\coloneqq(\mL_\star,\bDelta_R,\mV)\bcdot\bcS_\star, ~\bcE_d\coloneqq(\mL_\star,\mR_\star,\bDelta_V)\bcdot\bcS_\star
 ,\nonumber
 \\
 &~\bcE\coloneqq({\mL},{\mR},{\mV})\bcdot\bcS-\bcZ_{\star} = \bcE_a+\bcE_b+\bcE_c+\bcE_d,\nonumber
 \\ &\bcE_1 \coloneqq \G(\tilde{p}^{-1}\P_{\tilde{\Omega}}-\I)\G^{*}(\bcE),\nonumber
 \\&\mP_L\coloneqq\mL(\mL^H\mL)^{-1}\mL^H, \mP_R\coloneqq\mR(\mR^H\mR)^{-1}\mR^H, \label{eq:short_notations}
\end{align}
where $\tilde{p}=\frac{\tilde{m}}{sn}$. 
Next, we bound the first term in \eqref{eq:dist_updt_splitbd} as follows: 
\begin{align*}
     &\|({\mL} _{+}\bQ_{1}{-}\mL_{\star})\bSigma_{\star,1}\|_F^2
     \\&=\|\big(\bDelta_L{-}\eta(\cM_1(\bcE+\bcE_1))\breve{\mL}(\breve{\mL}^{H}\breve{\mL})^{-1}\big)\bSigma_{\star,1}\|_F^2
     \\&=\| \mI_0-\mI_1\|_F^2\leq (1+t)\|\mI_0\|_F^2+(1+1/t)\|\mI_1\|_F^2,
\end{align*}
where 
\begin{align*}
    \mI_0&=\bDelta_L-\eta\cM_1(\bcE)\breve{\mL}(\breve{\mL}^{H}\breve{\mL})^{-1}\big)\bSigma_{\star,1},
    \\ \mI_1&=\eta\cM_1(\bcE_1)\breve{\mL}(\breve{\mL}^{H}\breve{\mL})^{-1}\bSigma_{\star,1}, 
\end{align*}
 and the last inequality  results from the scalar inequality that $(a-b)^2\leq(1+t)a^2+(1+1/t)b^2$ ({$t>0$}). 
   Similarly,  
\begin{align*}
    & \|({\mR} _{+}\bQ_{2}-\mR_{\star})\bSigma_{\star,2}\|_F^2
     \leq (1+t)\|\mJ_0\|_F^2+(1+1/t)\|\mJ_1\|_F^2
     \\&\|({\mV} _{+}\bQ_{3}-\mV_{\star})\bSigma_{\star,3}\|_F^2
     \leq (1+t)\|\mK_0\|_F^2+(1+1/t)\|\mK_1\|_F^2,
\end{align*}
where $\mJ_0, \mJ_1, \mK_0, \mK_1$ are defined similarly as $\mI_0,\mI_1$. 

For the last term in \eqref{eq:dist_updt_splitbd} mainly associated 
with $\bcS$, we have
\begin{align*}
&\left\|(\bQ_{1}^{-1},\bQ_{2}^{-1},\bQ_{3}^{-1})\bcdot\bcS _{+}-\bcS_{\star}\right\|_{F}^2
\\&=\|\bDelta_S-\eta(\mL^{\dagger},\mR^{\dagger},\mV^{\dagger})\bcdot(\bcE+\bcE_1)\|_F^2
\\&=\|\mP_0-\mP_1\|_F^2\leq (1+t)\|\mP_0\|_F^2+(1+1/t)\|\mP_1\|_F^2,
\end{align*}
where we denote $\mL^{\dagger}=(\mL^H\mL)^{-1}\mL^H$, and $\mR^{\dagger}$, $\mV^{\dagger}$ are similarly defined. Also, we denote 
\begin{align*}
    \mP_0=\bDelta_S-\eta(\mL^{\dagger},\mR^{\dagger},\mV^{\dagger})\bcdot\bcE,~\mP_1=\eta(\mL^{\dagger},\mR^{\dagger},\mV^{\dagger})\bcdot\bcE_1. 
\end{align*}


In previous derivations, $\mI_0, \mJ_0, \mK_0, \mP_0$ are the quantities associated with the tensor Tucker factorization problem, 
which have been bounded in \cite{Tong2022}, seeing (45) and (46) in Appendix B of \cite{Tong2022}. We list the following results from \cite{Tong2022}: 
\begin{align*}
    \|\mI_0\|_F^2+\|\mJ_0\|_F^2+\|\mK_0\|_F^2+\|\mP_0\|_F^2{\leq}(1-0.7\eta)^2 \distsq{\mF}{\mF_{\star}},
\end{align*}
provided $\eta\leq2/5$ and $\varepsilon\leq0.2/C$ where $C$ is some universal constant. 

We concentrate more on $\mI_1$, $\mJ_1$, $\mK_1$, and $\mP_1$, which denote the perturbation from partial Hankel tensor sampling. 
We first give their results directly:
\begin{align*}
    \max\{\|\mI_1\|_F^2,\|\mJ_1\|_F^2,\|\mK_1\|_F^2,\|\mP_1\|_F^2\}\leq 6\eta^2\varepsilon^2 \distsq{\mF}{\mF_{\star}},
\end{align*}
provide $m\geq O(\varepsilon^{-2}\mu_0 c_{\mathrm{s}} sr\kappa^2\log(sn))$ and $\varepsilon\leq 0.2$, with high probability. The analysis for these terms is shown later, which is one part of our contribution, greatly different from \cite{Tong2022}. 

Combining the above pieces, we give an upper bound of $\distsq{\bF_{+}}{\bF_{\star}}$ and we set $t=\eta/10$:
\begin{align*}
    &\distsq{\bF_{+}}{\bF_{\star}} {\leq} (1{+}\frac{\eta}{10})\(\|\mI_0\|_F^2+\|\mJ_0\|_F^2+\|\mK_0\|_F^2+\|\mP_0\|_F^2\)\\&+(1+\frac{10}{\eta})(\|\mI_1\|_F^2+\|\mJ_1\|_F^2+\|\mK_1\|_F^2+\|\mP_1\|_F^2)
    \\&\leq (1-0.5\eta)^2\distsq{\mF}{\mF_{\star}},
\end{align*}
provided $\varepsilon\leq 1/40$ and $\eta\leq 0.4$. 
Next, we provide a detailed analysis of $\|\mI_1\|_F,\|\mJ_1\|_F,\|\mK_1\|_F$ and $\|\mP_1\|_F$.
\paragraph{Bounding $\|\mI_1\|_F,\|\mJ_1\|_F,\|\mK_1\|_F$}
The analysis for $\|\mI_1\|_F,\|\mJ_1\|_F,\|\mK_1\|_F$ is similar, and we take bounding $\|\mI_1\|_F$ for example. We reformulate that
 \begin{align*}
     \|\mI_1/\eta\|_F&=\|\cM_1(\bcE_1)\breve{\mL}(\breve{\mL}^{H}\breve{\mL})^{-1}\bSigma_{\star,1}\|_F\\
     &=\max_{\|\mN\|_F=1}|
    \langle\cM_1(\bcE_1)\breve{\mL}(\breve{\mL}^{H}\breve{\mL})^{-1}\bSigma_{\star,1},\mN\rangle|
 \end{align*} where $\mN\in\C^{n_1\times r}$ and $\|\mN\|_F=1$. Besides, 
     invoking $\breve{\mL}=(\overline{\mV} \otimes\overline{\mR} )\cM_{1}(\bcS )^{H}$, 
     we rewrite the variation form as
\begin{align*}
    &|\langle\cM_1(\bcE_1)\breve{\mL}(\breve{\mL}^{H}\breve{\mL})^{-1}\bSigma_{\star,1},\mN\rangle|
    \\&{=}
   | \langle\cM_1(\bcE_1)(\overline{\mV} \otimes\overline{\mR} )\cM_{1}(\bcS )^{H}(\breve{\mL}^{H}\breve{\mL})^{-1}\bSigma_{\star,1},\mN\rangle|
    \\&{=}|\langle \M_{1}(\bcE_1),\M_{1}((\mN \bSigma_{\star,1}(\breve{\mL}^{H}\breve{\mL})^{-1},\mR,\mV)\bcdot\bcS)\rangle |
    \\&{=}|\langle \G(p^{-1}\P_{\Omega}-\I)\G^{*}(\bcE_a+\bcE_b+\bcE_c+\bcE_d),(\mL_N,\mR,\mV)\bcdot\bcS\rangle |, \numberthis \label{eq:upbd_perturb_L}
\end{align*}
where we denote $\mL_N=\mN \bSigma_{\star,1}(\breve{\mL}^{H}\breve{\mL})^{-1}$, and the last equality results from the inner product between the tensor's matricization is equivalent that between tensors, $\bcE_1=\G(p^{-1}\P_{\Omega}{-}\I)\G^{*}(\bcE)$, and $\bcE = \bcE_a+\bcE_b+\bcE_c+\bcE_d$.   
For \eqref{eq:upbd_perturb_L}, it is obvious that we bound it separately and we bound the first part associated with $\bcE_a$ for example. We define two projection operators such that $\P_{T_1}(\bcE_a)=\bcE_a$ where $\bcE_a=(\mL,\mR,\mV)\bcdot\bDelta_S$ and $\P_{T_2}((\mL_N,\mR,\mV)\bcdot\bcS)=(\mL_N,\mR,\mV)\bcdot\bcS$, which are 
$$\P_{T_1}(\bcZ)=(\mP_L,\mP_R,\mI_s)\bcdot\bcZ,$$ $$\P_{T_2}(\bcZ)=(\mL_N(\mL_N^H\mL_N)^{-1}\mL_N^H,\mP_R,\mI_s)\bcdot\bcZ.$$
Thus the first part in \eqref{eq:upbd_perturb_L} is reformulated as
\begin{align*}
    &|\langle \P_{T_{2}}\G(\tilde{p}^{-1}\P_{\tilde{\Omega}}-\I)\G^{*}\P_{T_1}(\bcE_a),(\mL_N,\mR,\mV)\bcdot\bcS\rangle |
    \\&\leq \|\P_{T_{2}}\G(\tilde{p}^{-1}\P_{\tilde{\Omega}}-\I)\G^{*}\P_{T_1}\|\|\bcE_a\|_F\|(\mL_N,\mR,\mV)\bcdot\bcS\|_F
    \\&\leq 0.1\varepsilon \|\bcE_a\|_F \|\mN\|_F \|\bSigma_{\star,1}(\breve{\mL}^{H}\breve{\mL})^{-1}\cM_1(\bcS)(\mV\otimes\mR)^T\|
    \\&=0.1\varepsilon\|\bcE_a\|_F\|\bSigma_{\star,1}(\breve{\mL}^{H}\breve{\mL})^{-1}\breve{\mL}^H\|_F\leq 0.6\varepsilon \dist{\bF}{\bF_{\star}},
\end{align*}
provided $\varepsilon\leq 0.2$ and $\tilde{m}\geq C\varepsilon^{-2}\max\{q_1^2,q_2^2\}sn\log(sn)$ where $q_i = \max_{k,j}\|\P_{T_i}\bcH_{k,j}\|_F$ for $i=1,2$, and $C$ is some universal constant. 

In the last inequality, the bounds of $\|\bcE_a\|_F\leq 3\dist{\bF}{\bF_{\star}}$, $\|\bSigma_{\star,1}(\breve{\mL}^{H}\breve{\mL})^{-1}\breve{\mL}^H\|_F\leq(1-\varepsilon)^{-3}
$  can be checked in Lemma~\ref{lemma:perturb_bounds}. 
In the third line, we invoke $\|\P_{T_{2}}\G(\tilde{p}^{-1}\P_{\tilde{\Omega}}-\I)\G^{*}\P_{T_1}\|\leq 0.1\varepsilon$ from Lemma~\ref{lem:PTconc}. 
Now we check the bounds for $q_1$ and $q_2$.

For $q_1 = \max_{k,j}\|\P_{T_1}\bcH_{k,j}\|_F$, from Lemma~\ref{lem:PT_HTbasis_bd}, we have
\begin{align*}
    q_1=\max_{k,j}\|\P_{T_1}\bcH_{k,j}\|_F\leq \|\mL\|_{2,\infty}/\sigma_r(\mL)\lesssim \sqrt{\frac{\mu_0 c_{\mathrm{s}} r \kappa^2}{n}},
\end{align*}
where the bounds of $\|\mL\|_{2,\infty}$, $\sigma_r(\mL)$ can be checked in  Lemma~\ref{lemma:perturb_bounds}.
Similarly, we bound $q_2$ as
\begin{align*}
    q_2=\max_{k,j}\|\P_{T_2}\bcH_{k,j}\|_F\leq  \|\mR\|_{2,\infty}/\sigma_r(\mR)\lesssim \sqrt{\frac{\mu_0 c_{\mathrm{s}} r \kappa^2}{n}}.
\end{align*}
 Consequently, the sample condition is summarized as $\tilde{m}\geq O(\varepsilon^{-2}\mu_0 c_{\mathrm{s}} s r \kappa^2\log(sn))$. 

 Following the previous route, we can bound the other three parts in \eqref{eq:upbd_perturb_L} and finally \eqref{eq:upbd_perturb_L} is upper bounded as:
 \begin{align*}
     |\langle\cM_1(\bcE_1)\breve{\mL}(\breve{\mL}^{H}\breve{\mL})^{-1}\bSigma_{\star,1},\mN\rangle|\leq 2.4 \varepsilon \dist{\bF}{\bF_{\star}},
 \end{align*}
 which means that we can prove $$ \|\mI_1\|_F^2\leq 6\eta^2\varepsilon^2 \distsq{\mF}{\mF_{\star}}.$$ 
Similarly, we can establish the bounds for $\|\mJ_1\|_F^2$, $\|\mK_1\|_F^2$.  
 \paragraph{Bounding $\|\mP_1\|_F$}
  For any tensor $\bcZ$, 
  \begin{align*}
&(\mL^{\dagger},\mR^{\dagger},\mV^{\dagger})\bcdot((\mP_L,\mP_R,\mI_s)\bcdot\bcZ)\\&=(\mL^{\dagger}\mP_L,\mR^{\dagger}\mP_R,\mV^{\dagger})\bcdot\bcZ=(\mL^{\dagger},\mR^{\dagger},\mV^{\dagger})\bcdot\bcZ
\end{align*}
 from tensor algebra \eqref{eq:tensor_properties_1} and the fact $\mL^{\dagger}\mP_L=(\mL^H\mL)^{-1}\mL^H\mL(\mL^H\mL)^{-1}\mL^H=(\mL^H\mL)^{-1}\mL^H=\mL^{\dagger}$, and $\mR^{\dagger}\mP_R=\mR^{\dagger}$. Therefore, we reformulate $\|\mP_1/\eta\|_F$ as
 \begin{align*}
     &\|\mP_1/\eta\|_F =\|(\mL^{\dagger},\mR^{\dagger},\mV^{\dagger})\bcdot((\mP_L,\mP_R,\mI_s)\bcdot\bcE_1)\|_F
     \\&{=}\|(\mL^{\dagger},\mR^{\dagger},\mV^{\dagger})\bcdot(\P_{T_2}\bcE_1)\|_F\\
     &{=}\|(\mL^{\dagger},\mR^{\dagger},\mV^{\dagger})\bcdot\P_{T_2}(\G(\tilde{p}^{-1}\P_{\tilde{\Omega}}{-}\I)\G^{*}(\bcE_a{+}\bcE_b{+}\bcE_c{+}\bcE_d))\|_F, \numberthis \label{eq:upbd_perturb_S}
 \end{align*} where we invoke $\bcE_1=\G(\tilde{p}^{-1}\P_{\tilde{\Omega}}{-}\I)\G^{*}(\bcE)$,  $\bcE = \bcE_a+\bcE_b+\bcE_c+\bcE_d$, and define a projection operator
 $$\P_{T_2}(\bcZ)=(\mP_L,\mP_R,\mI_s)\bcdot\bcZ.$$
 We bound the second part associated with $\bcE_b$ in \eqref{eq:upbd_perturb_S} for example, where $\bcE_b=(\bDelta_L,\mR,\mV)\bcdot\bcS_\star$. Define a projection operator $\P_{T_1}$ such that $\P_{T_1}(\bcE_b)=\bcE_b$, which is 
 \begin{align*}
     \P_{T_1}(\bcZ)=(\bDelta_L(\bDelta_L^H\bDelta_L)\bDelta_L^H,\mP_R,\mI_s)\bcdot\bcZ,
 \end{align*}
for any tensor $\bcZ$.
Thus the second part in \eqref{eq:upbd_perturb_S} is rewritten as
\begin{align*}
    &\|(\mL^{\dagger},\mR^{\dagger},\mV^{\dagger})\bcdot\P_{T_2}(\G(\tilde{p}^{-1}\P_{\tilde{\Omega}}{-}\I)\G^{*}\P_{T_1}(\bcE_b))\|_F
    \\&\leq \|\mL^{\dagger}\|\|\mR^{\dagger}\|\|\mV^{\dagger}\| \|\P_{T_2}(\G(\tilde{p}^{-1}\P_{\tilde{\Omega}}{-}\I)\G^{*}\P_{T_1}\|\|\bcE_b\|_F
    \\&\leq 0.6\varepsilon \dist{\bF}{\bF_{\star}},
\end{align*}
provided $\varepsilon\leq 0.2$ and $\tilde{m}\geq C\varepsilon^{-2}\max\{q_1^2,q_2^2\}sn\log(sn)$ where $q_i = \max_{k,j}\|\P_{T_i}\bcH_{k,j}\|_F$ for $i=1,2$, and $C$ is some universal constant. The second line results from tensor algebra \eqref{eq:tensor_properties_3}. The bounds of $\|\mL^{\dagger}\|$, $\|\mR^{\dagger}\|$, $\|\mV^{\dagger}\|$ and $\|\bcE_b\|_F$ can be checked in Lemma~\ref{lemma:perturb_bounds}. In the last line, we invoke $\|\P_{T_{2}}\G(p^{-1}\P_{\tilde{\Omega}}-\I)\G^{*}\P_{T_1}\|\leq 0.1\varepsilon$ from Lemma~\ref{lem:PTconc}.

Similar to previous route for bounding $q_1,q_2$, by combining Lemma~\ref{lem:PT_HTbasis_bd} and Lemma~\ref{lemma:perturb_bounds}, we establish that $\max\{q_1^2,q_2^2\}\lesssim\mu_0 c_{\mathrm{s}} r\kappa^2/n$. Consequently, the sample condition is summarized as $\tilde{m}\geq O(\varepsilon^{-2}\mu_0 c_{\mathrm{s}} s r \kappa^2\log(sn))$. 

Following the previous route, we can bound the other three parts in \eqref{eq:upbd_perturb_S} and finally \eqref{eq:upbd_perturb_S} is upper bounded as:
 \begin{align*}
     \|\mP_1/\eta\|_F\leq 2.4 \varepsilon \dist{\bF}{\bF_{\star}}.
 \end{align*}


\subsection{Proofs of Lemma~\ref{lem:PTconc}}\label{pf_lem_PTconc}
The sampling set is $\Omega=\{\va_i|i=1,\cdots,\tilde{m}\}$ where the indices $\va_i$ is drawn independently and uniformly from $\{0,\cdots,s-1\}\times\{0,\cdots,n-1\}$. Let $\Omega_i= \{\va_i\}$, and then it is obvious $\P_{\tilde{\Omega}}(\mX)=\sum_{i=1}^{\tilde{m}}\P_{\Omega_i}(\mX)$. We rewrite that
\begin{small}
        \begin{align*}
       &\(\P_{T_2}\G(\tilde{p}^{-1}\P_{\tilde{\Omega}}-\I)\G^{*}\P_{T_1}\)(\bcZ)\\&
       =\sum_{i=1}^{\tilde{m}}\big( \tilde{p}^{-1}\P_{T_2}\G \P_{\Omega_i}\G^{*}\P_{T_1}-\frac{1}{\tilde{m}}\P_{T_2}\G\G^{*}\P_{T_1}\big)(\bcZ)
       \\&=\sum_{i=1}^{\tilde{m}} \(\tilde{\cX}_{i}-\overline{\cX}\)(\bcZ)=\sum_{i=1}^{\tilde{m}}\cX_{i}(\bcZ),
    \end{align*}
    \end{small}where we denote $\tilde{\cX}_{i}=\tilde{p}^{-1}\P_{T_2}\G \P_{\Omega_i}\G^{*}\P_{T_1}$, $\overline{\cX}=\frac{1}{\tilde{m}}\P_{T_2}\G\G^{*}\P_{T_1}$ and $\cX_i=\tilde{\cX}-\overline{\cX}$. 
    Also, we rewrite that 

\begin{small}
    \begin{align*}
        \tilde{\cX}_{i}(\bcZ)&{=}\frac{1}{\tilde{p}}\P_{T_2}\langle \bcH_{a_i,b_i},\P_{T_1}\bcZ\rangle\bcH_{a_i,b_i}{=}
        \frac{1}{\tilde{p}}\langle \P_{T_1}\bcH_{a_i,b_i},\bcZ\rangle\P_{T_2}\bcH_{a_i,b_i},
    \end{align*}
    \end{small}
    where each pair $(a_i,b_i)$ is drawn uniformly from $\{0,\cdots,n-1\}\times\{0,\cdots,s-1\}$. And we have
    \begin{align*}
        \overline{\cX}(\bcZ)=\frac{1}{\tilde{m}}\sum_{k,j}\langle \P_{T_1}\bcH_{k,j},\bcZ\rangle\P_{T_2}\bcH_{k,j},
    \end{align*}
 \vspace{-15pt}
 \begin{small}
          \begin{align*}
         \|\overline{\cX}\|&{=}\frac{1}{\tilde{m}}\sup_{\|\bcZ\|_F=1}\|\sum_{k,j}\langle \P_{T_1}\bcH_{k,j},\bcZ\rangle\P_{T_2}\bcH_{k,j}\|_F
         \\ &{\leq} \frac{sn}{\tilde{m}}\max_{k,j}\|\P_{T_1}\bcH_{k,j}\|_F\max_{k,j}\|\P_{T_2}\bcH_{k,j}\|_F=\frac{sn}{\tilde{m}}q_1 q_2. \numberthis \label{eq:PTconc_bd1_mean}
     \end{align*}    
      \end{small}We can bound $\|\tilde{\cX}_{i}\|$ from a similar route. Therefore,
     \begin{align*}
         \|\cX_{i}\|
         &{\leq} \frac{2sn}{\tilde{m}}\max_{k,j}\|\P_{T_1}\bcH_{k,j}\|_F\max_{k,j}\|\P_{T_2}\bcH_{k,j}\|_F{=} \frac{2sn}{\tilde{m}}q_1 q_2.
     \end{align*}
Besides, 
$
    \E(\cX^{*}_i\cX_i)=
    \E(\tilde{\cX}_{i}^{*}\tilde{\cX}_{i})-\overline{\cX}^{*}\overline{\cX}, 
$
and we point out a relation that 
\begin{align*}
    \tilde{\cX}_{i}^{*}\tilde{\cX}_{i}&{=}{\tilde{p}}^{-2}(\P_{T_2}\G\P_{\Omega_i}\G^{*}\P_{T_1})^{*} \P_{T_2}\G\P_{\Omega_i}\G^{*}\P_{T_1}\\
    &{=}\frac{s^2 n^2}{{\tilde{m}}^2}(\P_{\Omega_i}\G^{*}\P_{T_1})^{*}(\P_{T_2}\G\P_{\Omega_i})^{*}(\P_{T_2}\G\P_{\Omega_i})(\P_{\Omega_i}\G^{*}\P_{T_1}), 
\end{align*}
where $\Omega_i\in\{1,\cdots,s\}\times\{1,\cdots,n\}$. Thus 
\begin{small}
\begin{align*}
&\|\E(\tilde{\cX}_{i}^{*}\tilde{\cX}_{i})\|
{\overset{(a)}{\leq}}\frac{s^2 n^2}{\tilde{m}^2}\max_{\Omega_i}\|\P_{\Omega_i}\G^{*}\P_{T_2}\|^2\|\E\((\P_{\Omega_i}\G^{*}\P_{T_1})^{*}(\P_{\Omega_i}\G^{*}\P_{T_1})\)\|
\\&=\frac{s^2 n^2}{\tilde{m}^2}\max_{\Omega_i}\|\P_{\Omega_i}\G^{*}\P_{T_2}\|^2\|\E\(\P_{T_1}\G\P_{\Omega_i}\G^{*}\P_{T_1}\)\|
\\&=\frac{sn}{\tilde{m}^2}\max_{\Omega_i}\|\P_{\Omega_i}\G^{*}\P_{T_2}\|^2\|\P_{T_1}\G\G^{*}\P_{T_1}\|\overset{(b)}{\leq}\frac{sn}{\tilde{m}^2} q_2^2,
\end{align*}
\end{small}where $(a)$ is from 
the matrix version of linear operators:
\begin{align*}
    \|\sum_{i}\mA_i^H\mB_i^H\mB_i\mA_i\|&\leq\max_{i}\|\mB_i^H\mB_i\|\|\sum_{i}\mA_i^H\mA_i\|
    \\&\leq \max_{i}\|\mB_i\|^2\|\sum_{i}\mA_i^H\mA_i\|,
\end{align*}
and in $(b)$ we apply the fact that $\|\P_{T_1}\G\G^{*}\P_{T_1}\|\leq (\|\P_{T_1}\|\|\G\|)^2\leq 1$ and 
\begin{align*}
{\|\P_{\Omega_i}\G^{*}\P_{T_2}\|}&=\max_{\|\bcZ\|_F{=}1}\|{\langle}\bcH_{a_i,b_i},\P_{T_2}\bcZ\rangle \ve_{b_i}\ve_{a_i}^T\|_F\\
&\leq\max_{k,j}\|\P_{T_2}\bcH_{k,j}\|_F=q_2.
\end{align*}
Then we can bound the following quantity:
 \begin{align*}
     \|\E(\sum_{i=1}^{\tilde{m}}\cX^{*}_{i}\cX_{i})\|
     &\leq {\tilde{m}}\|\E(\tilde{\cX}_{i}^{*}\tilde{\cX}_{i})\| +{\tilde{m}}\|\overline{\cX}^{*}\overline{\cX}\|
    \\ &\leq \frac{s n}{{\tilde{m}}}q_2^2+{\tilde{m}}\|\overline{\cX}\|\|\overline{\cX}\|\leq \frac{2sn}{{\tilde{m}}}\max\{q_1^2,q_2^2\},
 \end{align*}
 where in the last inequality we invoke that fact that $\|\overline{\cX}\|=\frac{1}{{\tilde{m}}}\|\P_{T_2}\G\G^{*}\P_{T_1}\|\leq\frac{1}{{\tilde{m}}}$ and $\|\overline{\cX}\|\leq \frac{sn}{{\tilde{m}}}q_1 q_2\leq\frac{sn}{{\tilde{m}}}\max\{q_1^2,q_2^2\}$ from \eqref{eq:PTconc_bd1_mean}. 
 Similarly, we can obtain that 
$
\|\E(\sum_{i=1}^{\tilde{m}}\cX_{i}\cX_{i}^{*})\|\leq\frac{2sn}{{\tilde{m}}}\max\{q_1^2,q_2^2\}.
$
 
 By applying Bernstein inequality \cite[Theorem 1.6]{Tropp2012}, with probability $1-(sn)^{-2}$, one has
 \begin{small} 
 \begin{align*}
    \| \sum_{i=1}^{\tilde{m}} \cX_{i}(\bcZ)\|
    &{\lesssim}\big(\sqrt{\frac{\max\{q_1^2,q_2^2\}sn\log(sn)}{{\tilde{m}}}}{+}\frac{\max\{q_1^2,q_2^2\}sn\log(sn)}{{\tilde{m}}}\big)
    \\&\lesssim \sqrt{\frac{\max\{q_1^2,q_2^2\}sn\log(sn)}{{\tilde{m}}}},
 \end{align*}
  \end{small} provided $\tilde{m}\geq \max\{q_1^2,q_2^2\}sn\log(sn)$.
 \vspace{-5pt}
  \section{} \label{apd:pf_noise}
  
{
  \textbf{Proof of Corollary~\ref{cor:recovery_noise}.}
   Similarly in the noiseless case, we need to bound $ \dist{\mF^k}{\mF_\star}$  and
    apply an inductive way to prove that  the upper bound of $\dist{\mF^k}{\mF_\star}$ is 
\begin{small}
\begin{align*}
    \dist{\mF^k}{\mF_\star}\leq \varepsilon_0\rho^k\sigma_{\min}(\bcZ_\star)/3+C_3\frac{1-\rho^{k+1}}{1-\rho}\sigma\sqrt{{n^2\max\{s,n\}}}, \numberthis \label{eq:dist_noise_bd_induc}
\end{align*}
    \end{small}where $\rho=1-0.3\eta$. 
   For $k=0$, from Lemma~\ref{lem:init_noise}, the previous inequality holds provided $\hat{m}\gtrsim O(\varepsilon_0^{-2}\mu_0 c_{\mathrm{s}} sr^3\kappa^2\log^2(sn))$ and $C_3\geq C_1$. Besides, the incoherence condition \eqref{eq:incoh_proj_F0} holds. 
   
   Next, supposing  \eqref{eq:dist_noise_bd_induc} and \eqref{eq:incoh_iter_full} hold for the $k$-th step, ${\sigma}\leq c_0\frac{\sigma_{\max}(\bcZ_\star)}{\kappa\sqrt{n^2\max\{s,n\}}}$ where $c_0< O(\varepsilon_0)$ is a sufficiently small constant. Then invoke Lemma~\ref{lem:iter_converge_noise} to obtain 
   \begin{align*}
     &  \dist{{\mF'}^{k+1}}{\mF_{\star}} \\
     &\leq \varepsilon_0\rho^{k+1}\sigma_{\min}(\bcZ_\star)/3{+}C_3\sigma\sqrt{n^2\max\{s,n\}}(1+\rho\frac{1-\rho^{k+1}}{1-\rho})
       \\&=  \varepsilon_0\rho^{k+1}\sigma_{\min}(\bcZ_\star)/3{+}C_3\frac{1-\rho^{k+1}}{1-\rho}\sigma\sqrt{{n^2\max\{s,n\}}},
   \end{align*}
   provided $\hat{m}\gtrsim O(\varepsilon_0^{-2}\mu_0 c_{\mathrm{s}} sr\kappa^2\log^2(sn))$ and we invoke that $C_3\geq C_2$. We can conclude that  $ \dist{{\mF'}^{k+1}}{\mF_{\star}}\leq \varepsilon_0 \sigma_{\min}(\bcZ_\star)$ from the condition of $\sigma$. Therefore, as in the noiseless case, we can obtain that 
   \begin{align*}
       \dist{{\mF}^{k+1}}{\mF_{\star}} \leq  \dist{{\mF'}^{k+1}}{\mF_{\star}} 
   \end{align*}
   after the projection step, and the incoherence condition \eqref{eq:incoh_iter_full} for the $k+1$-th step holds. We finish the induction route and conclude that $m=(k+1)\hat{m}\gtrsim O(\varepsilon_0^{-2}\mu_0 c_{\mathrm{s}} sr^3\kappa^2\log^2(sn))$. Last, combining the fact $0.88\leq\rho<1$, we define a  constant $C_0$ such that $C_0/3\geq C_3 \frac{1-\rho^{k+1}}{1-\rho}$. 
   }
   {
   \begin{lemma}[Initialization with noise] \label{lem:init_noise}
  Suppose the conditions in Lemma~\ref{lem:init} hold,   
  the noise matrix $\mE\in\C^{s\times n}$ has independent sub-Gaussian entries with parameter $\sigma$ \cite{Vershynin2018},  and 
  ${\sigma}\leq c_1\frac{\sigma_{\max}(\bcZ_\star)}{\kappa\sqrt{n^2\max\{s,n\}}}$ for some sufficiently small constant $c_1$.  For $\mF_0=(\mL^0,\mR^0,\mV^0,\bcS^0)$ in Algorithm \ref{alg:init}, 
  with probability at least $1-O\((sn)^{-2}\)$
\begin{align*}
    \dist{\mF^0}{\mF_{\star}} \le \varepsilon_0\sigma_{\min}(\bcZ_\star)/3+C_1\sigma\sqrt{{n^2\max\{s,n\}}},
\end{align*}
holds provided $\hat{m}\geq O(\varepsilon_0^{-2}\mu_0 c_{\mathrm{s}} sr^3\kappa^2\log^2(sn))$, $C_1$ is some constant. Besides, 
  \begin{align*}
      \max\{\|\mL^0(\breve{\mL}^0)^H\|_{2,\infty},\|\mR^0(\breve{\mR}^0)^H\|_{2,\infty}\}\leq B/\sqrt{n}, \numberthis \label{eq:incoh_proj_F0_noise}
  \end{align*}
\end{lemma}
\begin{proof}
    See Appedix~\ref{pf:init_noise}. 
\end{proof}

  \begin{lemma}[Local convergence with noise] \label{lem:iter_converge_noise}
  Suppose the conditions in Lemma~\ref{lem:linconverge} hold, and the noise matrix $\mE\in\C^{s\times n}$ has independent sub-Gaussian entries with parameter $\sigma$ \cite{Vershynin2018}.  
  If $\mF$
is independent of the {current sampling set $\tilde{\Omega}$ with $\tilde{m}$ samples},
 with probability at least $1-O\((sn)^{-2}\)$ we have
\begin{align*}
    \dist{\mF_{+}}{\mF_{\star}} \leq (1-0.3\eta) \dist{\mF}{\mF_{\star}}+C_{2}\sigma\sqrt{{n^2\max\{s,n\}}},
\end{align*}
provided {$\tilde{m}$}
$\geq O(\varepsilon^{-2}\mu_0 c_{\mathrm{s}} sr\kappa^2\log^2(sn))$ and $\eta\leq 0.4$, where the next iterate $\mF_{+}=(\mL_{+},\mR_{+},\mV_{+},\bcS_{+})$ is updated in \eqref{eq:ScalHT}, $C_2$ is some constant.  
  \end{lemma}
\begin{proof}
      See Appedix~\ref{pf:iter_noise}.
\end{proof}
}
 \vspace{-15pt}
\subsection{Proof of Lemma~\ref{lem:init_noise}} \label{pf:init_noise}
{
As in the initialization under the noiseless setting,
\begin{align*}
     \dist{\mF'^0}{\mF_{\star}} &\leq(\sqrt{2}+1)^{3/2}\|(\mL'^{0},\mR'^{0},{\mV}^{0})\bcdot\bcS^{0}-\bcZ_{\star}\|_F.
\end{align*}
Denote $\bcZ^0=\hat{p}^{-1}\G\P_{\Omega_0}(\mY_\star)$ and $\bcE_e=\hat{p}^{-1}\G\P_{\Omega_0}(\D(\mE))$
\vspace{-0.2cm}
\begin{align*}
    \bcZ^0_{e}=\hat{p}^{-1}\G\P_{\Omega_0}(\mY_\star+\D(\mE))=\bcZ^0+\bcE_e.
\end{align*}
Similarly to the proof of Lemma~\ref{lem:init}, 
we can establish  
\begin{align*}
& \left\|(\mL'^{0},\mR'^{0},{\mV}^{0})\bcdot\bcS^{0}-\bcZ_{\star}\right\|_{F} = \left\|\bcZ_{\star}-(\bP_{L},\bP_{R},\bP_{V})\bcdot\bcZ^0_{e}\right\|_{F} \nonumber\\
&= \big\|(\bP_{L},\bP_{R},\bP_{V})\bcdot(\bcZ_{\star}-\bcZ^0_{e})-(\bP_{L},\bP_{R},\bP_{V})\bcdot\bcZ_{\star}+\bcZ_\star\big\|_F
\nonumber 
\\& \leq  \left\|\bP_{L_\perp}\cM_{1}(\bcZ_{\star})\right\|_F {+} \left\|\bP_{R_\perp}\cM_{2}(\bcZ_{\star})\right\|_F {+} \left\|\bP_{V_\perp}\cM_{3}(\bcZ_{\star})\right\|_F 
\\&\quad{+}\left\|(\bP_{L},\bP_{R},\bP_{V})\bcdot(\bcZ^0-\bcZ_{\star})\right\|_F
{+}\left\|(\bP_{L},\bP_{R},\bP_{V})\bcdot(\bcE_e)\right\|_F
, \numberthis \label{eq:TC_init_expand_noise}
\end{align*}
where we invoke $\bcS^0=\big(({\mL'}^0)^H,({\mR'}^0)^H,(\mV^0)^H\big)\bcdot \bcZ^0$ and the previous decomposition of $\bcZ_\star$. 
The first four terms have been bounded in our noiseless setting, and we give their results directly, which are $\left\|\bP_{L_\perp}\cM_{1}(\bcZ_{\star})\right\|_F,\left\|\bP_{R_\perp}\cM_{2}(\bcZ_{\star})\right\|_F,\left\|\bP_{V_\perp}\cM_{3}(\bcZ_{\star})\right\|_F,$ ~~$\left\|(\bP_{L},\bP_{R},\bP_{V})\bcdot(\bcZ^0{-}\bcZ_{\star})\right\|_F\lesssim \varepsilon_0 \sigma_{\min}(\bcZ_\star)$.  
, provided $\hat{m}\geq O(\varepsilon_0^{-2}\mu_0 c_{\mathrm{s}} sr^3\kappa^2\log(sn))$. For the last term
\begin{align*}
    &\left\|(\bP_{L},\bP_{R},\bP_{V})\bcdot(\bcE_e)\right\|_F\leq\sqrt{r}\|\hat{p}^{-1}\cM_1(\G\P_{\Omega_0}(\D(\mE)))\| 
    \\&\lesssim  \sigma\sqrt{\frac{sn^2r\max\{s,n\}}{\hat{m}}}\lesssim \sigma\sqrt{{n^2\max\{s,n\}}},
\end{align*}
provided $\hat{m}\geq O(sr\log^2(sn))$, where 
the last inequality results from Lemma~\ref{lem:robs_noise} and Lemma~\ref{lem:weight_noise}. We conclude that 
\begin{align*}
    \dist{\mF'^0}{\mF_{\star}} \leq  \varepsilon_0\sigma_{\min}(\bcZ_\star)/3+C_1\sigma\sqrt{{n^2\max\{s,n\}}}.
\end{align*}
When ${\sigma}\leq 
c_1\frac{\sigma_{\max}(\bcZ_\star)}{\kappa\sqrt{n^2\max\{s,n\}}}$ for some sufficiently small constant $c_1$ and $\varepsilon_0<1$, invoking {Lemma~\ref{lem:proj}}, we obtain
\begin{align*}
    \dist{\mF^0}{\mF_{\star}} &\leq \dist{\mF'^0}{\mF_{\star}} . 
\end{align*}
and 
$ \max\{\|\mL^0(\breve{\mL}^0)^H\|_{2,\infty},\|\mR^0(\breve{\mR}^0)^H\|_{2,\infty}\}\leq B/\sqrt{n}$. 
}
 \vspace{-15pt}
 \subsection{Proof of Lemma~\ref{lem:iter_converge_noise}}\label{pf:iter_noise}
{
 We introduce the notations such as $\bDelta_{L},\bcE_a,\bcE,\bcE_1,\mP_L,\mP_R$, as in the noiseless setting. Additionally, denote $\bcE_e=\tilde{p}^{-1}\G\P_{\tilde{\Omega}}\D(\mE)$. 
Recall the expansion of distance metric in \eqref{eq:dist_updt_splitbd} and suppose the factor quadruple  $\mF=(\mL, \mR , \mV,\bcS)$ is aligned with $\mF_\star$ as before.  
We bound the first term in \eqref{eq:dist_updt_splitbd} as follows: 
\begin{align*}
     &\|({\mL} _{+}\bQ_{1}{-}\mL_{\star})\bSigma_{\star,1}\|_F^2
     \\&=\|\big(\bDelta_L{-}\eta(\cM_1(\bcE+\bcE_1+\bcE_e))\breve{\mL}(\breve{\mL}^{H}\breve{\mL})^{-1}\big)\bSigma_{\star,1}\|_F^2
     \\&=\| \mI_0-(\mI_1+\mI_2)\|_F^2\leq (1+t)\|\mI_0\|_F^2+(1+1/t)\|\mI_1+\mI_2\|_F^2\\
     &\leq (1+t)\|\mI_0\|_F^2+2(1+1/t)(\|\mI_1\|_F^2+\|\mI_2\|_F^2),
\end{align*}
where $\mI_0=\bDelta_L-\eta\cM_1(\bcE)\breve{\mL}(\breve{\mL}^{H}\breve{\mL})^{-1}\big)\bSigma_{\star,1}$, $\mI_1=\eta\cM_1(\bcE_1)\breve{\mL}(\breve{\mL}^{H}\breve{\mL})^{-1}\bSigma_{\star,1},$ have been defined before, and
\begin{align*}
  \mI_2& = \eta\cM_1(\bcE_e)\breve{\mL}(\breve{\mL}^{H}\breve{\mL})^{-1}\bSigma_{\star,1}, 
\end{align*}
   Similarly, we can define $\mJ_0, \mJ_1, \mJ_2$ for $ \|({\mR} _{+}\bQ_{2}{-}\mR_{\star})\bSigma_{\star,2}\|_F^2$ and $\mK_0, \mK_1,\mK_2$ for $\|({\mV} _{+}\bQ_{3}{-}\mV_{\star})\bSigma_{\star,3}\|_F^2$.  

For the last term in \eqref{eq:TC_init_expand_noise} mainly associated 
with $\bcS$:
\begin{align*}
&\left\|(\bQ_{1}^{-1},\bQ_{2}^{-1},\bQ_{3}^{-1})\bcdot\bcS _{+}-\bcS_{\star}\right\|_{F}^2
\\&=\|\bDelta_S-\eta(\mL^{\dagger},\mR^{\dagger},\mV^{\dagger})\bcdot(\bcE+\bcE_1+\bcE_e)\|_F^2
=\|\mP_0{-}\mP_1{-}\mP_2\|_F^2
\\&\leq (1+t)\|\mP_0\|_F^2+2(1+1/t)(\|\mP_1\|_F^2+\|\mP_2\|_F^2),
\end{align*}
where we denote $\mL^{\dagger}=(\mL^H\mL)^{-1}\mL^H$, and $\mR^{\dagger}$, $\mV^{\dagger}$ are similarly defined. Also, we denote $\mP_0=\bDelta_S-\eta(\mL^{\dagger},\mR^{\dagger},\mV^{\dagger})\bcdot\bcE,~\mP_1=\eta(\mL^{\dagger},\mR^{\dagger},\mV^{\dagger})\bcdot\bcE_1,$ and
\begin{align*}
    \mP_2=\eta(\mL^{\dagger},\mR^{\dagger},\mV^{\dagger})\bcdot\bcE_e. 
\end{align*}


In previous derivations, $\mI_0, \mJ_0, \mK_0, \mP_0$ are the quantities associated with the tensor Tucker factorization problem, and we have listed their results in Appendix~\ref{pf_lem_lin_cvg}. 

$\mI_1$, $\mJ_1$, $\mK_1$, and $\mP_1$ are the quantities associated with the noiseless Hankel tensor completion problem, and have been bounded before in Appendix~\ref{pf_lem_lin_cvg}, provide $m\geq O(\varepsilon^{-2}\mu_0 c_{\mathrm{s}} sr\kappa^2\log(sn))$ and $\varepsilon\leq 0.2$.  
For the noise part
\begin{align*}
&\|\mI_2\|_F^2 = \|\eta\cM_1(\bcE_e)\breve{\mL}(\breve{\mL}^{H}\breve{\mL})^{-1}\bSigma_{\star,1}\|_F^2
\\&\leq \eta^2 \|\cM_1(\bcE_e)\|^2 \|\breve{\mL}(\breve{\mL}^{H}\breve{\mL})^{-1}\bSigma_{\star,1}\|_F^2\leq (1{-}\varepsilon)^{-6}\eta^2 \sigma^2{n^2\max\{s,n\}},
\end{align*}
where we invoke  $\|\breve{\mL}(\breve{\mL}^{H}\breve{\mL})^{-1}\bSigma_{\star,1}\|_F\leq \sqrt{r}\|\breve{\mL}(\breve{\mL}^{H}\breve{\mL})^{-1}\bSigma_{\star,1}\|\leq \sqrt{r}(1-\varepsilon)^{-3}$ from Lemma~\ref{lemma:perturb_bounds} and combine Lemma~\ref{lem:robs_noise},\ref{lem:weight_noise} about 
noise such that
\begin{align*}
    \|\cM_1(\bcE_e)\|&=\|\tilde{p}^{-1}\cM_1(\G\P_{\tilde{\Omega}}(\D(\mE)))\|\\
    &\lesssim  \sigma\sqrt{\frac{sn^2r\max\{s,n\}}{\tilde{m}}}\lesssim \sigma\sqrt{{n^2\max\{s,n\}}},
\end{align*}
provided $\tilde{m}\geq sr\log^2(sn)$. The similar bound holds for $\|\mJ_2\|_F^2$, $\|\mK_2\|_F^2$. For $\mP_2$
\begin{align*}
    \|\mP_2\|_F^2&\leq \eta^2 \|\mR^\dagger\|^2\|\mV^\dagger\|^2\|(\mL^\dagger,\mI_{n_2},\mI_{s})\bcdot\bcE_e\|_F^2
    \\&\leq  \eta^2 r \|\mL^\dagger\|^2\|\mR^\dagger\|^2\|\mV^\dagger\|^2\|\|\cM_1(\bcE_e)\|^2
    \\&\leq (1-\varepsilon)^{-6}\eta^2 \sigma^2{n^2\max\{s,n\}}.
\end{align*}

Combining the above pieces, we give an upper bound of $\distsq{\bF_{+}}{\bF_{\star}}$ and set $t=\eta/10$:
\begin{align*}
    &\distsq{\bF_{+}}{\bF_{\star}} {\leq} (1{+}\frac{\eta}{10})\(\|\mI_0\|_F^2+\|\mJ_0\|_F^2+\|\mK_0\|_F^2+\|\mP_0\|_F^2\)\\&\quad+2(1+\frac{10}{\eta})\big((\|\mI_1\|_F^2+\|\mJ_1\|_F^2+\|\mK_1\|_F^2+\|\mP_1\|_F^2)
   \\& \quad+(\|\mI_2\|_F^2+\|\mJ_2\|_F^2+\|\mK_2\|_F^2+\|\mP_2\|_F^2)\big)
    \\&\leq (1-0.3\eta)^2\distsq{\mF}{\mF_{\star}}{+}2\eta(\eta{+}10)(1{-}\varepsilon)^{-6}\sigma^2{{n^2\max\{s,n\}}}
    \\&\leq (1-0.3\eta)^2\distsq{\mF}{\mF_{\star}}+C_{2}^2\sigma^2{{n^2\max\{s,n\}}},
\end{align*}
provided $\varepsilon\leq 1/40$, $\eta\leq 0.4$ and $C_2$ is some constant. Invoking the fact that $\sqrt{a^2+b^2}\leq a+b$ for  $a,b>0$, we establish that 
\begin{align*}
    \dist{\bF_{+}}{\bF_{\star}}\leq (1-0.3\eta)\dist{\mF}{\mF_{\star}}+C_{2}\sigma\sqrt{{n^2\max\{s,n\}}}.
\end{align*}
 }
 \vspace{-20pt}
\subsection{Auxiliary lemmas}
{
\begin{lemma} \label{lem:robs_noise}
Suppose the noise matrix $\mE\in\C^{s\times n}$ has independent sub-Gaussian entries with parameter $\sigma$ \cite{Vershynin2018}. When  $p\cdot\max\{s,n\}\gtrsim \log^2(sn)$, with probability at least $1-O((sn)^{-2})$ we have
\begin{align*}
        \| \P_{\Omega}(\mE)\|\lesssim \sigma \sqrt{p\max\{s,n\}}, 
\end{align*}
where $\Omega$ is the index set with $m$ random samples independent of the noise, and $p=m/sn$. 
\end{lemma}
\begin{proof}
    This lemma is 
    adapted from \cite[Lemma~3]{Chen2020}. We distinguish the dimensions $s$ and $n$ of the rectangular matrix $\mE$, while in \cite[Lemma~3]{Chen2020}, the two dimensions of the rectangular matrix exhibit the same order of  $n$. 
\end{proof}
\begin{remark}
    If $m\geq s\log^2(sn)$, it is guaranteed that $p\cdot\max\{s,n\}\gtrsim \log^2(sn)$ where $p=m/sn$. 
\end{remark}
\begin{lemma} \label{lem:weight_noise}
        \begin{align*}
        \|\cM_i\(\G(\P_\Omega(\D(\mE)))\)\|\leq \sqrt{n}\|\P_\Omega(\mE)\|,~i=1,2,3.
    \end{align*}
\end{lemma}
  \begin{proof}
    From the 
    spectral norm that $\|\cM_i\|\leq 1$ where $i=1,2,3,$ $\|\G\|\leq 1$, and $\|\D\|\leq \sqrt{n}$, it is obvious that 
    \begin{align*}
        \|\cM_i\(\G(\P_\Omega(\D(\mE)))\)\|&=\|\cM_i\(\G(\D(\P_\Omega(\mE)))\)\|
        \\&\leq \|\D\| \|\cM_i\|\|\G\|\|\P_\Omega(\mE)\|\leq \sqrt{n}\|\P_\Omega(\mE)\|. 
    \end{align*}
  \end{proof}
}
\vspace{-15pt}
 \section{}
\begin{lemma}\cite[Lemma~14]{Tong2022} \label{lem:dist_fullerr_upbd}
    For $\mF=(\mL,\mR,\mV,\bcS)$, the distance metric satisfies
    \begin{align*}
 \dist{\mF}{\mF_{\star}}\le (\sqrt{2}+1)^{3/2}\left\|(\mL,\mR,\mV)\bcdot\bcS-\bcZ_{\star}\right\|_{F}.
\end{align*}
\end{lemma}

We present the results of perturbation bounds, and the short notations in this lemma are introduced in \eqref{eq:short_notations}. 
\begin{lemma}[Local perturbation results]\label{lemma:perturb_bounds}
Suppose $\bF=(\mL,\mR,\mV,\bcS)$ and $\bF_{\star}=(\mL_{\star},\mR_{\star},\mV_{\star},\bcS_{\star})$ are aligned and satisfy $\dist{\bF}{\bF_{\star}}\le\epsilon\sigma_{\min}(\bcZ_{\star})$ for some $\epsilon<1$. Then the following bounds hold 
\vspace{-5pt}
\begin{subequations} 
\begin{align}
1/\sigma_r(\mL)=\|\mL(\mL^{H}\mL)^{-1}\|& \le(1-\epsilon)^{-1};\label{eq:perturb_Lsigmar}
\\
\left\Vert \breve{\mL}(\breve{\mL}^{H}\breve{\mL})^{-1}\bSigma_{\star,1}\right\Vert & \le(1-\epsilon)^{-3}. \label{eq:perturb_Rinv}
\end{align}
\end{subequations}
By symmetry, the previous bounds holds for $\mR,\breve{\mR}$ and $\mV,\breve{\mV}$. 
If the incoherence condition \eqref{eq:incoh_iter_full} holds, we have 
\begin{subequations} 
\begin{align}
     \|\mL\|_{2,\infty}&\leq (1-\varepsilon)^{-3}C_B\kappa\sqrt{\mu_0 c_{\mathrm{s}} r}/\sqrt{n};\label{eq:incoh_factorL}
     \\ \|\mR\|_{2,\infty}&\leq (1-\varepsilon)^{-3}C_B\kappa\sqrt{\mu_0 c_{\mathrm{s}} r}/\sqrt{n}. \label{eq:incoh_factorR}
\end{align}
\end{subequations}
Besides, when $\varepsilon\leq 0.2$ we have 
\begin{align}
 \|\bcE\|_F\leq \|\bcE_a\|_F+\|\bcE_b\|_F+\|\bcE_c&\|_F+\|\bcE_d\|_F\leq 3\dist{\bF}{\bF_{\star}}.\label{eq:upbd_deltaTsplit}
\end{align}
\begin{proof}
    The inequalities \eqref{eq:incoh_factorL} and \eqref{eq:incoh_factorR} can be found in \cite[ Lemma~17]{Tong2022} with slight modifications. Other inequalities can be found in \cite[Lemma~16]{Tong2022} and its proof. 
\end{proof}
\end{lemma}
\vspace{-15pt}
\end{document}